%% file: main.tex
\documentclass[10pt]{article}
\usepackage{xspace}
\usepackage{tikz}
\usepackage{pgfplots}
\usepackage{multirow}
\usepackage{marvosym}
\usepackage{enumitem}
\usepackage{fullpage}
\usepackage{hyperref}
\usepackage{amsthm}
\usepackage{amssymb}
\usepackage{mathtools}
\usepackage{stmaryrd}
\usepackage{amsfonts}
\usepackage{booktabs}
\usepackage{subcaption}
\usepackage{tikz-3dplot} 

\definecolor{light-gray}{gray}{0.7.2}
\definecolor{goodgreen}{rgb}{0.1, 0.5, 0.1}
\newcommand{\bigO}[1]{\mathcal{O}(#1)}
\newcommand{\softO}[1]{\widetilde{\mathcal{O}}(#1)}
\newcommand{\smallO}[1]{o(#1)}
\newcommand{\eps}{\epsilon}    
\newcommand{\ivme}{\text{IVM}$^{\eps}$\xspace}
\newcommand{\dbeps}{\mathbf{P}}
\newcommand{\Dom}{\mathsf{Dom}}
\newcommand{\inst}[1]{\mathbf{#1}}
\newcommand{\db}{\inst{D}}
\newcommand{\astate}{\mathcal{Z}}
\newcommand{\upd}{\mathit{u}}
\newcommand{\update}{\mathit{apply}}
\newcommand{\minor}{\mathit{minor}}
\newcommand{\major}{\mathit{major}}
\newcommand{\OuMv}{\textsf{OuMv}\xspace}

\newcommand{\OMv}{\textsf{OMv}\xspace}

\newcommand{\ztimes}{\cdot}
\newcommand{\textsmaller}[1]{\text{\scalebox{0.85}{#1}}}
\renewcommand{\H}{\protect\textsmaller{$H$}}
\renewcommand{\L}{\protect\textsmaller{$L$}}

\newcommand{\boxminusspaced}{\text{\scalebox{0.75}{\hspace{0.1mm}$\boxminus$\hspace{0.1mm}}}}
\newcommand{\F}{\protect\boxminusspaced}

\newcommand{\vecnormal}[1]{\textnormal{\bf #1}\xspace}

\newcommand{\floor}[1]{\left\lfloor #1 \right\rfloor}
\newcommand{\ceil}[1]{\left\lceil #1 \right\rceil}

\newcommand{\deltaA}{\alpha}
\newcommand{\deltaB}{\beta}

\newcommand{\p}{\mathit{m}}

\newcommand{\linenumber}{\makebox[2ex][r]{\rownumber\TAB}}

\newcommand{\calR}{\mathcal{R}}
\newcommand{\calT}{\mathcal{T}}

\newcommand{\calF}{\mathcal{F}}

\newcommand{\calS}{\mathcal{S}}

\newcommand{\TAB}{\makebox[2.5ex][r]{}}%
\newcommand{\STAB}{\makebox[1.5ex][r]{}}%
\newcommand{\SPACE}{\makebox[0.75ex][r]{}}%
\newcommand{\OUTPUT}{\textbf{output}\xspace}%
\newcommand{\LET}{\textbf{let}\xspace}%
\newcommand{\IF}{\textbf{if}\xspace}%
\newcommand{\ELSE}{\textbf{else}\xspace}%
\newcommand{\WHILE}{\textbf{while}\xspace}%
\newcommand{\FOREACH}{\textbf{foreach}\xspace}%
\newcommand{\DO}{\textbf{do}\xspace}%
\newcommand{\AND}{\textbf{and}\xspace}%
\newcommand{\OR}{\textbf{or}\xspace}%
\newcommand{\NOT}{\textbf{not}\xspace}%
\newcommand{\EOF}{\textbf{EOF}\xspace}%
\newcommand{\RETURN}{\textbf{return}\xspace}%
\newcommand{\skipp}{\mathit{skipTo}\xspace}%
\newcommand{\backskipp}{\mathit{skippedFrom}\xspace}%
\newcommand{\nop}[1]{}

\usepackage{todonotes}

\theoremstyle{plain}                  
\newtheorem{theorem}{Theorem}
\newtheorem{lemma}[theorem]{Lemma}
\newtheorem{proposition}[theorem]{Proposition}
\newtheorem{conjecture}[theorem]{Conjecture}
\newtheorem{corollary}[theorem]{Corollary}         

\newtheorem{definition}[theorem]{Definition}
\newtheorem{example}[theorem]{Example}

\newcounter{magicrownumbers}
\newcommand\rownumber{\footnotesize\stepcounter{magicrownumbers}\arabic{magicrownumbers}}

\title{Maintaining Triangle Queries under Updates}

\author{
Ahmet Kara$^1$, 
Milos Nikolic$^2$, 
Hung Q. Ngo$^3$,
Dan Olteanu$^1$, 
Haozhe Zhang$^1$ 
\\ \\
$^1$University of Oxford  
\enspace\enspace 
$^2$University of Edinburgh
\enspace\enspace 
$^3$RelationalAI, Inc.
}

\date{}
\begin{document}

\maketitle
\begin{abstract}
We consider the problem of incrementally maintaining the triangle queries with arbitrary free variables under single-tuple updates to the input relations. 

We introduce an approach called \ivme that exhibits a trade-off between the update time, the space, and the delay for the enumeration of the query result, such that the update time ranges from the square root to linear in the database size while the delay ranges from constant to linear time.

\ivme achieves Pareto worst-case  optimality in the update-delay space conditioned on the Online Matrix-Vector Multiplication conjecture. 
It is strongly Pareto optimal for the triangle queries with zero or three free variables and weakly Pareto optimal for the  triangle queries with one or two free variables.

\end{abstract}

\paragraph{Acknowledgements}
This project has received funding from the European Union's Horizon 2020 research and innovation programme under grant agreement No 682588.

\maketitle

\input{introduction}

\input{preliminaries}

\input{count}

\input{full}

\input{binary}

\input{unary}

\input{rebalancing}

\input{amortization}

\input{lowerBound}

\input{recovery}

\input{related}

\input{extensions}

\input{conclusion}

\bibliographystyle{abbrv}
\bibliography{bibliography}
 





\end{document}

%% file: introduction.tex
\section{Introduction}
\label{sec:intro}
In this article we consider the problem of incrementally maintaining 
triangle queries \nop{and their Loomis-Whitney generalizations} under single-tuple updates to the input relations. 
We introduce an approach to this problem that 
expresses a trade-off between the update time, space, and enumeration delay.
The {\em update time} is the time needed to maintain the data structure encoding the 
query result upon a  single-tuple update. 
The {\em space} is the overall memory needed by the used data structure. 
The {\em enumeration delay} is the maximal time
needed 
from starting the enumeration or reporting one result tuple to reporting the next result tuple or ending the enumeration.

We consider the triangle queries written in FAQ notation~\cite{FAQ:PODS:2016}. 
Let $R$, $S$, and $T$ be relations that have schemas $(A,B)$, $(B,C)$, and 
$(C,A)$, respectively, and are given as functions mapping tuples over their schemas 
to tuple multiplicities.
The {\em ternary} triangle query
$$\triangle_3(a,b,c) = R(a,b)\ztimes S(b,c) \ztimes T(c,a)$$ 
returns each triangle and its multiplicity in the join of the three relations.
The {\em binary} triangle query 
$$\triangle_2(a,b) = \sum_{c\in\Dom(C)} R(a,b)\ztimes S(b,c) \ztimes T(c,a)$$ returns each $(A,B)$-pair that occurs in a triangle and its multiplicity.
The {\em unary} triangle query 
$$\triangle_1(a) = \sum_{b\in\Dom(B)} \sum_{c\in\Dom(C)} R(a,b)\ztimes S(b,c) \ztimes T(c,a)$$ returns each $A$-value that occurs in a triangle and its multiplicity.
Finally, the {\em nullary} triangle query
$$\triangle_0()  = \sum_{a\in\Dom(A)} \sum_{b\in\Dom(B)} \sum_{c\in\Dom(C)} R(a,b)\ztimes S(b,c) \ztimes T(c,a)$$ returns the number of triangles.
There are further unary and binary triangle queries, e.g., $\triangle_1(b)$ or $\triangle_2(b,c)$, 
yet they can be treated similarly since the join of the three relations is symmetric in $A$, $B$, and $C$.

\nop{
The Loomis-Whitney (LW) queries generalize triangle queries from cliques of degree three to cliques of degree $n\geq 3$; they encode the 
Loomis Whitney inequality~\cite{LW:1949}.
Let $A_1,\ldots,A_n$ be the query variables and $R_1,\ldots,R_n$ relations over schemas ${\bf X}_1,\ldots,{\bf X}_n$, 
where $\forall i\in[n]: {\bf X}_i = (A_{((i+j)\mod n) + 1})_{-1\leq j\leq n-3}$. 
That is, the schema of $R_1$ is $(A_1,\ldots,A_{n-1})$, whereas the schema of $R_n$ is $(A_n,A_1,\ldots,A_{n-2})$.
\nop{
if we go from 0 to n-1, then we can use for recurrence: $(i+j) \mod n$ for $0\leq j\leq n-2$.
}
The nullary LW query (of degree $n$)  has the form 
$$\Diamond_n() = \sum_{a_1\in\Dom(A_1)}\ldots\sum_{a_n\in\Dom(A_n)}  
R_1({\bf x}_1)\cdots R_n({\bf x}_n),$$ where $\forall i\in[n]: {\bf x}_i$ is a value from the domain of the tuple ${\bf X}_i$ of variables.
As for triangle queries, a LW query of arity $0\leq j\leq n$ has the same body as for arity $n$ but only keeps the first $j$ values in the result. For instance, for $n=4$ the $4$-ary LW query is 
$$\Diamond_4(a_1,a_2,a_3,a_4)=R_1(a_1,a_2,a_3)\cdot R_2(a_2,a_3,a_4)\cdot R_3(a_3,a_4,a_1)\cdot R_4(a_4,a_1,a_2).$$ 
In case $n=3$, each LW query $\Diamond_j$ becomes the triangle query $\triangle_j$, for $0\leq j\leq n$.
}

The ternary triangle query has served as a milestone for the worst-case optimality of join algorithms in the centralized and parallel settings. Likewise, 
the nullary triangle query is a working horse for randomized approximation schemes for data processing. They showcase the suboptimality of mainstream join algorithms used currently by virtually all commercial database systems. 
For a database $\db$ consisting of relations $R$, $S$, and $T$, standard binary join plans implementing these queries may take $O(|\db|^2)$ time, 
yet the ternary and nullary triangle queries can be solved in 
$\bigO{|\db|^{\frac{3}{2}}}$ \cite{NgoPRR18} and respectively
$\bigO{|\db|^{1.41}}$ 
time~\cite{AYZ:Counting:1997}.
This observation motivated a new line of work on worst-case optimal algorithms for arbitrary join queries~\cite{NgoPRR18}. Triangle queries have also served as a yardstick for understanding the optimal communication cost for parallel query evaluation in the Massively Parallel Communication model~\cite{Koutris:FTDB:2018}. They 
have witnessed the development of randomized approximation schemes with increasingly lower time and space requirements~\cite{Eden:approximately:FOCS:2015}.

In our prior work we introduced a worst-case optimal approach for incrementally maintaining 
the exact result of the nullary triangle query~\cite{KaraNNOZ19}. This article extends that work 
with an investigation of Pareto worst-case  optimality for the triangle queries in the update-delay space.

Incremental maintenance algorithms may benefit from a range of processing techniques whose combinations make it more challenging to reason about optimality.
Such techniques include algorithms for aggregate-join queries with low complexity developed for the non-incremental case~\cite{NgoPRR18}; pre-materialization of views to reduce the maintenance of a query to that of subqueries~\cite{KochAKNNLS14}; and delta processing that allows to only compute the change to the result instead of the entire result~\cite{Chirkova:Views:2012:FTD}.

\subsection{Existing Incremental View Maintenance (IVM) Approaches}

The problem of incrementally maintaining triangle queries 
has received a fair amount of attention. We next discuss the na\"ive approach, which
recomputes the query result from scratch, and several IVM approaches.

We consider the single-tuple update $\delta R = \{(\deltaA,\deltaB) \mapsto m\}$ to 
a binary relation $R$ that maps the tuple $(\deltaA,\deltaB)$ to a 
nonzero multiplicity $m$, which is positive for inserts and negative for deletes. 

The na\"ive approach incurs constant-time updates: Each update is executed on a relation of the input database $\inst{D}$. Whenever we need the query result, we recompute it  
in time $\bigO{|\db|^{\frac{3}{2}}}$~\cite{AYZ:Counting:1997,NgoPRR18}. 
The number of distinct tuples in the result is at most 
$|\db|^{\frac{3}{2}}$~\cite{LW:1949}. 

We next exemplify the classical first-order IVM~\cite{Chirkova:Views:2012:FTD} on the 
nullary triangle query $\triangle_0$ 
under the aforementioned single-tuple update $\delta R$;
all other triangle queries are treated similarly.
The classical IVM approach materializes the query result, computes on the 
fly a delta query $\delta \triangle_0$, and then
updates the query result:
\begin{align*}
\delta \triangle_0() = \delta R(\alpha,\beta)\ztimes \sum_{c\in\Dom(C)}   
 S(\beta,c) \ztimes T(c,\alpha),
\hspace*{6em} \triangle_0() = \triangle_0() + \delta\triangle_0().
\end{align*} 
The delta computation takes $\bigO{|\db|}$ time since it needs to intersect two lists of possibly linearly many $C$-values that are paired with $\deltaB$ in $S$ and with $\deltaA$ in $T$ (i.e., the multiplicity of such pairs in $S$ and $T$ is nonzero).  
Since the query result is materialized, it can be enumerated with constant delay.

The recursive IVM~\cite{KochAKNNLS14} speeds up the delta computation by precomputing three auxiliary views representing the update-independent parts of the delta queries:
\begin{align*}
V_{ST}(b,a) &= \sum\limits_{c \in \Dom(C)} S(b,c) \ztimes T(c,a) \\
V_{TR}(c,b) &= \sum\limits_{a \in \Dom(A)} T(c,a) \ztimes R(a,b) \\
V_{RS}(a,c) &= \sum\limits_{b \in \Dom(B)} R(a,b) \ztimes S(b,c).
\end{align*}
These three views take $\bigO{|\db|^2}$ space but allow to compute the delta query for single-tuple updates to the input relations in $\bigO{1}$ time. 
Computing the delta $\delta\triangle_0() = \delta R(\deltaA,\deltaB) \ztimes  V_{ST}(\deltaB,\deltaA)$ requires just a constant-time lookup in $V_{ST}$;
however, maintaining the views $V_{RS}$ and $V_{TR}$, which refer to $R$, still requires $\bigO{|\db|}$ time.
The factorized IVM~\cite{Nikolic:SIGMOD:18} materializes only one of the three views, for instance, $V_{ST}$. In this case, the maintenance under updates to $R$ takes $\bigO{1}$ time, but the maintenance under updates to $S$ and $T$ still takes $\bigO{|\db|}$ time.

Further exact IVM approaches focus on acyclic conjunctive queries. For free-connex acyclic conjunctive queries, the dynamic Yannakakis approach allows for enumeration of result tuples with constant delay after single-tuple updates in linear time~\cite{Idris:dynamic:SIGMOD:2017}. 
For databases with or without integrity constraints, it is known that a strict, small subset of the class of acyclic conjunctive queries admit constant-time update, while all other conjunctive queries have update times dependent on the size of the input database~\cite{BerkholzKS17,Berkholz:ICDT:2018}.

A line of work relevant to our result unveils structure in the PTIME complexity class by giving lower bounds on the complexity of problems under various conjectures~\cite{Henzinger:OMv:2015,Williams:2018:finegrained}.

\begin{definition}[Online Matrix-Vector Multiplication (\OMv)~\cite{Henzinger:OMv:2015}]\label{def:OMv}
We are given an $n \times n$ Boolean matrix $\vecnormal{M}$ and  receive $n$ column vectors of size $n$, denoted by $\vecnormal{v}_1, \ldots, \vecnormal{v}_n$, one by one; after seeing each vector $\vecnormal{v}_i$, we output the product $\vecnormal{M} \vecnormal{v}_i$ before we see the next vector.
\end{definition}

\begin{conjecture}[\OMv Conjecture, Theorem 2.4 in~\cite{Henzinger:OMv:2015}]\label{conj:omv}
For any $\gamma > 0$, there is no algorithm that solves \OMv in time $\bigO{n^{3-\gamma}}$.
\end{conjecture}

The \OMv conjecture has been used to exhibit conditional lower bounds 
for many dynamic problems, including those previously based on other popular problems and conjectures, such as 3SUM and combinatorial Boolean matrix multiplication~\cite{Henzinger:OMv:2015}. This also applies to the nullary triangle query:
For any $\gamma > 0$ and database of domain size $n$, 
there is no algorithm that incrementally maintains
the query result under single-tuple updates
with arbitrary preprocessing time, $\bigO{n^{1-\gamma}}$ update time,
and $\bigO{n^{2-\gamma}}$ answer time,
unless the \OMv conjecture fails~\cite{BerkholzKS17}. All aforementioned prior approaches to maintaining triangle queries do not meet this (conditional) lower bound and are thus not worst-case optimal.

\subsection{Contributions of This Article}

This article introduces \ivme, an IVM approach for triangle queries with arbitrary free variables 
that exhibits a trade-off between the update time, the space, and the enumeration delay.

\begin{theorem}
\label{theo:main_result_triangle}
Given a database $\inst{D}$ and $\eps \in [0,1]$, 
\ivme incrementally maintains the triangle queries
under single-tuple updates to $\inst{D}$ with
$\bigO{|\inst{D}|^{\frac{3}{2}}}$ preprocessing time and
$\bigO{|\inst{D}|^{\max\{\eps,1-\eps\}}}$ amortized 
update time.
The space complexity and enumeration delay
are given in Table \ref{table:complexities_triangle}:

\begin{table}[h!]
  \vspace{-0.5em}
\begin{center}
\begin{tabular}{@{\hskip 0.0in}l@{\hskip 0.2in}l@{\hskip 0.2in}l@{\hskip 0.2in}l@{\hskip 0.2in}l@{\hskip 0.0in}}
& $\triangle_0$ & $\triangle_1$ & $\triangle_2$ & $\triangle_3$  \\
\midrule
Space 
& $\bigO{|\inst{D}|^{1 + \min\{\eps , 1- \eps\}}}$ 
& $\bigO{|\inst{D}|^{1 + \min\{\eps , 1- \eps\}}}$ 
& $\bigO{|\inst{D}|^{1 + \min\{\eps , 1- \eps\}}}$ 
& $\bigO{|\inst{D}|^{\frac{3}{2}}}$\\
Enumeration delay  
&  $\bigO{1}$ 
& $\bigO{|\inst{D}|^{2\min\{\eps, 1-\eps\}}}$
& $\bigO{|\inst{D}|^{\min\{\eps, 1-\eps\}}}$
&  $\bigO{1}$
\end{tabular}
\end{center}
\vspace{-1em}
\caption{\ivme's space and enumeration delay 
for maintaining triangle queries.}
 \label{table:complexities_triangle}
\end{table}
\end{theorem}
 
The preprocessing time is the time to compute the query result 
on the initial database before the updates; if we start with the empty database, then this is $\bigO{1}$.
\ivme maintains triangle queries with repeating relation symbols
 with the same complexities from Theorem~\ref{theo:main_result_triangle}.

\ivme uses a data structure that partitions each input relation 
based on the degrees of data values. 
The degree of an $A$-value $a$ in relation $R$ is 
the number of $B$-values paired with $a$ in $R$. 
The degree of $B$- and $C$-values is defined analogously. 
Depending on whether a combination of relation parts
includes data values with high or low degrees,
\ivme uses a different maintenance strategy.     
Thanks to this degree-based adaptive processing, the overall update time 
of \ivme
is kept sublinear. 
As the database evolves under updates, \ivme needs to rebalance the 
relation partitions to account for updated degrees of data values. While this rebalancing may take superlinear time, it remains sublinear per single-tuple update. The overall update time is therefore amortized.

We distinguish two types of relation partitioning. In {\em single partitioning},
relations are partitioned based on the degrees of data values 
in one column. In {\em double partitioning},
relations are partitioned based on the degrees of data values 
in two columns.
Unary and binary triangle queries
require double partitioning to obtain the complexity results
in Theorem~\ref{theo:main_result_triangle}. 
For the nullary and ternary triangle queries,
single partitioning suffices to obtain 
these complexity results. Nevertheless, 
double partitioning can lower the space complexity in case of the nullary triangle query,
as stated next.

\begin{proposition}
\label{prop:tighter-upper-bound-space-nullary}
Given a database $\inst{D}$ and $\eps \in [0,1]$, 
\ivme incrementally maintains the nullary triangle query
under single-tuple updates to $\inst{D}$ with
$\bigO{|\inst{D}|^{\frac{3}{2}}}$ preprocessing time,
$\bigO{|\inst{D}|^{\max\{\eps,1-\eps\}}}$ amortized 
update time, $\bigO{|\inst{D}|^{\max\{1,\min\{1+\eps,2-2\eps\}\}}}$
space complexity, and $\bigO{1}$ enumeration delay.
\end{proposition}

For $\eps=0$ and $\eps\geq\frac{1}{2}$, the space complexity needed by \ivme to maintain the nullary triangle query becomes linear; its maximum is $\bigO{|\inst{D}|^{4/3}}$ for $\eps=\frac{1}{3}$.


\begin{figure}[t]
\begin{center}
\begin{minipage}{\textwidth}
\begin{center}
\begin{tikzpicture}
\begin{axis}[
grid=major, 
    grid style={dotted,line width=1pt},
xmin=0, xmax=1, ymin=0, ymax=1.5,
every axis plot post/.append style={mark=none},
  xtick ={0,0.5, 1},
  ytick ={0, 0.5, 1,1.5},
  xticklabels={\footnotesize{$0$},
  \footnotesize{$\frac{1}{2}$},\footnotesize{$1$}},
   yticklabels={$ $,\footnotesize{$\frac{1}{2}$},
   \footnotesize{$1$},\footnotesize{$\frac{3}{2}$}},  
y=2.5cm,
    x=3.7cm,
axis lines=middle,
    axis line style={->},
    x label style={at={(axis description cs:1.15,-0.06)}},
    xlabel={\footnotesize{$\eps$}},
    y label style={at={(axis description cs:-0.7,1.1)},align=center},
     ylabel= \footnotesize{Complexity} \\ $\bigO{|\inst{D}|^y}$,
  axis x line*=bottom,
  axis y line*=left,
 ]

  \addplot[color = blue, mark=none,domain=0:1,thick,dashed] coordinates{ 
 (0, 1.5) 
  (1, 1.5)
};

  \addplot[color = blue, mark=none,domain=0:1,thick,dashed] coordinates{ 
 (0, 1) 
  (0.5, 1.5)
    (1, 1)
};

  \addplot[color=red,mark=none,domain=0:1,thick] coordinates{ 
  (0, 1) 
  (1/2, 1/2)
  (1,1) 
}; 

  \addplot[color=goodgreen,mark=none,domain=0:1,very thick,dotted] coordinates{ 
  (0, 0.02) 
  (1,0.02) 
}; 

  \addplot[color=goodgreen,mark=none,domain=0:1,very thick,dotted] coordinates{ 
  (0, 0)
  (1/2, 1/2) 
  (1,0) 
}; 

  \addplot[color=goodgreen,mark=none,domain=0:1,very thick,dotted] coordinates{ 
  (0, 0) 
  (0.5,1)
  (1,0) 
}; 

\end{axis}
\node at(0,4) {\footnotesize{$y$}};

\node at(1.9,0.25) {\color{goodgreen}\footnotesize{$\triangle_0$,  
$\triangle_3$}};

\node at(1.55,0.8) {\color{goodgreen}\footnotesize{$\triangle_2$}};

\node at(1.35,2.2) {\color{goodgreen}\footnotesize{$\triangle_1$}};

\node at(1,3.9) {\color{blue}\footnotesize{$\triangle_3$}};

\node at(3.6,3.2) {\color{blue}\footnotesize{$\triangle_0, \triangle_1, \triangle_2$}};

\node at(3.9,1.8) {\color{red}\footnotesize{$\triangle_0, \triangle_1, \triangle_2, \triangle_3$}};

\begin{scope}[xshift = 5.5cm,yshift = 4cm]
\draw [line width= 0.8pt, dashed, blue](0,0) -- (0.5,0);
\node[left] at(1.6,0) {\footnotesize{Space}};
\node[left,blue] at(5.6,-0.7) {\footnotesize 
$
y = 
\begin{cases}
1 + \min\{\eps, 1-\eps\}, & \text{ for $\triangle_0, \triangle_1, \triangle_2$} \\ 
\frac{3}{2}, & \text{ for $\triangle_3$}
\end{cases}
$
};
\end{scope}

\begin{scope}[xshift = 0.5cm,yshift = 2.2cm]
\draw [line width= 0.8pt,red](5,0) -- (5.5,0);
\node[left] at(8.85,0) {\footnotesize{Amortized update time}};
\node[red] at(7.45,-0.4) {\footnotesize $y =\max\{\eps,1-\eps\}$ for $\triangle_0, \triangle_1, \triangle_2, \triangle_3$};
\end{scope}

\begin{scope}[xshift = -3.8cm,yshift = 1cm]
\draw [line width= 1.4pt,dotted, goodgreen](9.3,0) -- (9.8,0);
\node[left] at(12.6,0) {\footnotesize{Enumeration delay}};
\node[left,goodgreen] at(14,-0.9){\footnotesize
$
y = 
\begin{cases}
0,                               & \text{for $\triangle_0, \triangle_3$} \\ 
2\min\{ \eps,\, 1 - \eps\}, & \text{for $\triangle_1$}\\
\min\{ \eps,\, 1 - \eps\},                        & \text{for $\triangle_2$}
\end{cases}
$
};
\end{scope}

\node at(1.85,-1.3) {\footnotesize{optimal static for $\triangle_3$}};
\draw [->,>=stealth,dotted, line width=0.3mm] (1.85,-1.15) -- (1.85,-0.65);

\node at(5.75,-1.3) {\footnotesize{classical IVM for $\triangle_0, \triangle_1, \triangle_2, \triangle_3$}};
\draw [->,>=stealth, dotted, line width=0.3mm] (5.25,-0.9) -- (3.9,-0.2);

\node at(-2.1,-1.3) {\footnotesize{classical IVM for $\triangle_0, \triangle_1, \triangle_2, \triangle_3$}};
\draw [->,>=stealth,dotted, line width=0.3mm] (-1.2,-0.85) -- (-0.3,-0.15);

\end{tikzpicture}
\end{center}
\end{minipage}
\end{center}
\caption{
\ivme's amortized update time, space,  and enumeration delay 
for maintaining triangle queries. 
$|\inst{D}|$ is the database size. The complexities are  
parameterized by $\eps$. The space and enumeration delay depend on 
the arity of the query result.
By setting $\eps$ to $0$ or $1$, \ivme recovers classical first-order  IVM.
For $\eps = \frac{1}{2}$, \ivme computes the ternary triangle query worst-case optimally.
}

\label{fig:complexity_plots_triangle}
\end{figure}
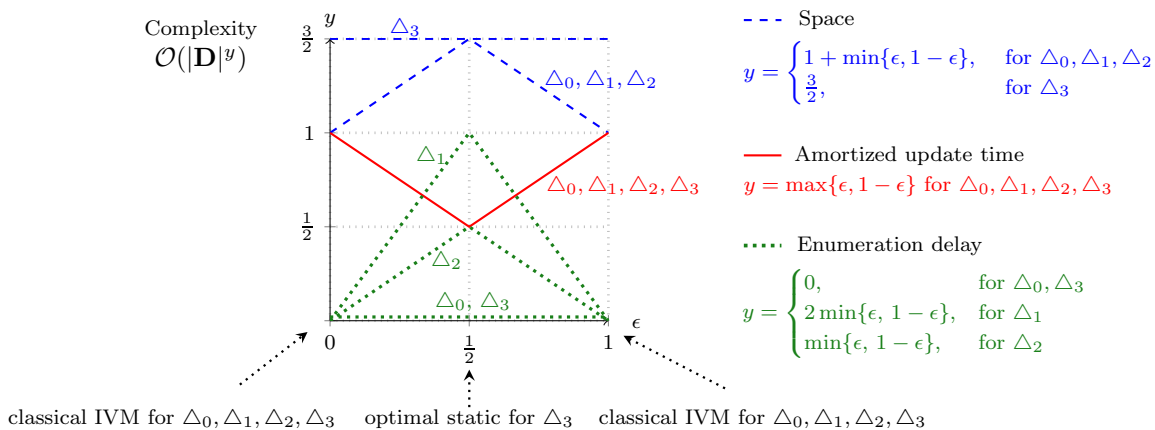 


As depicted in Figure~\ref{fig:complexity_plots_triangle}, \ivme defines a continuum of maintenance approaches that exhibit a trade-off between amortized  update time, enumeration delay, and space based on the parameter $\epsilon$, which ranges from 0 to 1.
We can recover the classical first-order IVM for all triangle queries by setting 
$\eps$ to $0$  or $1$.
For $\eps=\frac{1}{2}$, \ivme recovers the worst-case optimal time $\bigO{|\db|^{\frac{3}{2}}}$ of non-incremental algorithms 
for computing all tuples in the result of the ternary triangle query~\cite{NgoPRR18}.
Whereas these static algorithms are monolithic and require processing the input data in bulk and joining all relations at once, 
\ivme achieves the same complexity by inserting $|\db|$ tuples one at a time in initially empty relations by using its update mechanism 
and binary join plans. 
Using binary join plans in the static case is suboptimal,
since they can lead to intermediate results that 
are larger than the final result~\cite{NgoPRR18}.

The following proposition shows that some combinations of update time and delay in the update-delay space are not possible, conditioned on the \OMv Conjecture~\ref{conj:omv}.

\begin{proposition}
\label{prop:lower_bound_triangle}
For any $\gamma > 0$ and database $\db$,
there is no algorithm that incrementally maintains the result of
any triangle query under single-tuple updates to $\db$ with arbitrary preprocessing 
time, $\bigO{|\db|^{\frac{1}{2} - \gamma}}$ amortized update time, and $\bigO{|\db|^{1 - \gamma}}$ enumeration delay, unless the \OMv conjecture fails.
\end{proposition} 

\nop{Proposition~\ref{prop:lower_bound_triangle} is an adaptation of a prior lower bound result~\cite{BerkholzKS17}.}

\begin{figure}[t!]
  \begin{center} 
 \begin{minipage}{0.43\textwidth}
        \tdplotsetmaincoords{78}{175}
  \begin{tikzpicture}[xscale=2.5, yscale=1.0,tdplot_main_coords]
      \coordinate (O) at (0,0,0);


      \draw[thick,dotted,color=gray] (0,0,0) -- (1,0,0);
      \draw[thick,->] (1,0,0) -- (1.25,0,0) node[anchor=north]{\footnotesize $\log_{|\db|}${delay}};
      \draw[thick,dotted,color=gray] (0,0,0) -- (0,4,0);
      \draw[thick,->] (0,4,0) -- (0,5,0) node[anchor=north]{\footnotesize $\log_{|\db|}${space}};
      \draw[thick,dotted,color=gray] (0,0,0) -- (0,0,1);
      \draw[thick,->] (0,0,1) -- (0,0,2.25) node[anchor=south]{\footnotesize $\log_{|\db|}${update time}};

      \node[color=gray] at (-0.1,0,0) () {\footnotesize $0$};
      \coordinate (P1) at (1,1,0);
      \coordinate (P1x) at (1,0,0);
      
      \coordinate (P2) at (0,2.00,2);
      \coordinate (P2z) at (0,0,2);
      \node[] at (-0.1,0,2) () {\footnotesize $1$};
      \node[] at (-0.1,0,1) () {\footnotesize $0.5$};

      \coordinate (P3) at (0,4,2);
      \coordinate (P3z) at (0,0,2);

      \draw[dashed, color=black] (P2) -- (0,2.00,0);
 
      
      \draw[dashed, color=gray] (0,2.00,1) -- (1,2.00,1);
      \draw[dashed, color=gray] (0.5,0,1) -- (0.5,3.00,1);

      \draw[dashed, color=gray] (1,2.00,0) -- (1,2.00,1);
      \draw[dashed, color=gray] (1,2.00,0) -- (0,2.00,0);
      \draw[dashed, color=gray] (0.5,0,0) -- (0.5,0,1);
      \draw[dashed, color=gray] (0.5,3.00,0) -- (0.5,3.00,1);
      \draw[dashed, color=gray] (0.5,0,0) -- (0.5,3.00,0);

      \draw[dashed, color=gray] (1,3.00,0) -- (0,3.00,0);
      \draw[dashed, color=gray] (1,3.00,0) -- (1,3.00,1);
      \draw[dashed, color=gray] (0,3.00,0) -- (0,3.00,1);
      \draw[dashed, color=gray] (1,3.00,0) -- (0,3.00,0);
      \draw[dashed, color=gray] (1,3.00,1) -- (0,3.00,1);

      \node[] at (1.05,0,0.2) () {\color{black} \footnotesize $1$};
      \node[] at (0.5,0,-0.2) () {\color{black} \footnotesize $0.5$};
      \node[] at (-0.05,2.00,0.1) () {\color{black} \footnotesize $1$};
      \node[] at (-0.09,3.00,0.1) () {\color{black} \footnotesize $1.5$};

      \coordinate (X_half) at (1,0,0);
      \coordinate (Z_half) at (0,0,1);
      \coordinate (XZ_half) at (1,0,1);
      \fill[gray!30, opacity=0.4] (XZ_half) -- (Z_half) -- (0,4,1) -- (1,4,1) -- cycle;
      \fill[gray!30, opacity=0.4] (XZ_half) -- (X_half) -- (1,4,0) -- (1,4,1) -- cycle;
      \fill[gray!30, opacity=0.4] (1,4,1) -- (1,4,0) -- (0,4,0) -- (0,4,1) -- cycle;
      \draw[color=gray] (1,0,0) -- (1,4,0);
      \draw[color=gray] (1,0,0) -- (1,0,1);
      \draw[color=gray] (1,0,1) -- (0,0,1);
      \draw[color=gray] (1,0,1) -- (1,4,1);
      \draw[color=gray] (1,4,1) -- (1,4,0);
      \draw[color=gray] (1,4,1) -- (0,4,1);
      \draw[color=gray] (1,4,0) -- (0,4,0);
      \draw[color=gray] (0,4,1) -- (0,0,1);
      \draw[color=gray] (0,4,1) -- (0,4,0);

      \draw[dashed, color=gray] (1,2.00,1) -- (1,2.00,2);
      \draw[dashed, color=gray] (0,2.00,2) -- (1,2.00,2);
      \draw[dashed, color=gray] (1,2.00,2) -- (1,0,2);
      \draw[dashed, color=gray] (1,0,1) -- (1,0,2);
      \draw[dashed, color=gray] (1,0,2) -- (0,0,2);
      \draw[dashed, color=gray] (0,3.00,2) -- (0,0,2);

      \draw[thick, color=orange] (0,3.00,1) -- (0,3.00,2);
      \node[] at(-0.1,3.00,1.5) () {\footnotesize {\color{orange} $\triangle_3$}};

      \draw[thick, color=red] (0,3.00,1) -- (0,2.00,2);
      \node[] at(0.09,3.00,1.25) () {\footnotesize {\color{red} $\triangle_0$}};

      \filldraw[red] (0,3.00,1) ellipse(1pt and 2.5pt);
      \node[] at(-0.0,3.00,0.75) () {\footnotesize $A$};

      \draw[thick, color=goodgreen] (0.5,3.00,1) -- (0,2.00,2);
    \node[] at(0.3,3.00,1.2) () {\footnotesize {\color{goodgreen} $\triangle_2$}};

      \filldraw[goodgreen] (0.5,3.00,1) ellipse(1pt and 2.5pt);
      \node[] at(0.5,3.00,0.75) () {\footnotesize $B$};
      \node[] at(1.1,3.00,1.5) () {\footnotesize {\color{blue} $\triangle_1$}};

      \draw[thick, color=blue] (1,3.00,1) -- (0,2.00,2);

      \filldraw[blue] (1,3.00,1) ellipse(1pt and 2.5pt);
      \node[] at(1,3.00,0.75) () {\footnotesize $C$};
\nop{
\begin{pgfonlayer}{background}
    \filldraw [line width=4mm,join=round,blue!7]
      (1.5,3.5) rectangle (-1.3,-9.2);
  \end{pgfonlayer}      
}   
    \end{tikzpicture}
 \end{minipage}
 \hspace{-0.8cm}
   \begin{minipage}{0.6\textwidth}
      \centering
      \begin{small}

\begin{tabular}{c@{\hskip 0.13in}c@{\hskip 0.13in}c@{\hskip 0.13in}c@{\hskip 0.13in}c}
  \toprule
\multirow{2}{*}{$\eps$}  & \multirow{2}{*}{Query}  & Pareto & Amortized  & Enumeration   \\
& & optimality &  update time & delay  \\
  \midrule
 $\frac{1}{2}$ &  {\color{red} $\triangle_0$} and {\color{orange}$\triangle_3$}  & strong ($A$) & $\bigO{|\inst{D}|^{\frac{1}{2}}}$ & $\bigO{1}$   \\ [0.15cm]

 $\frac{1}{2}$ & {\color{goodgreen} $\triangle_2$ }& weak ($B$) & $\bigO{|\inst{D}|^{\frac{1}{2}}}$ & $\bigO{|\inst{D}|^{\frac{1}{2}}}$   \\ [0.15cm]

$\frac{1}{2}$ & {\color{blue} $\triangle_1$ } & weak ($C$) & $\bigO{|\inst{D}|^{\frac{1}{2}}}$ & $\bigO{|\inst{D}|}$   \\ [0.05cm]
    \bottomrule
      \end{tabular}
    \end{small}
  \end{minipage}
  \vspace{-1em}
  \caption{
 (left) \ivme's 
trade-offs between 
space complexity, 
amortized update time,
and enumeration delay
 for the maintenance of triangle queries.
 The preprocessing time is $\bigO{|\inst{D}|^{\frac{3}{2}}}$
for all triangle queries.
There is no algorithm that can maintain a triangle query with update time and enumeration delay representing
 a point in the gray cuboid, 
  unless the \OMv conjecture fails (Proposition \ref{prop:lower_bound_triangle}). 
  The  surface of the gray 
  cuboid
   corresponds to Pareto worst-case optimal  combinations
  of amortized update time and enumeration delay. 
(right)  
\ivme is strongly Pareto optimal at point $A$
for $\triangle_0$ and $\triangle_3$
and weakly Pareto optimal at point $B$
and $C$ for   $\triangle_2$ and respectively $\triangle_1$.
$\eps=\frac{1}{2}$  for points $A$, $B$, and $C$.
  }

  \label{fig:optimality_plot_3d}
  \end{center}
\end{figure}


Figure~\ref{fig:optimality_plot_3d} visualizes  
\ivme's 
trade-offs between 
space complexity, 
amortized update time,
and enumeration delay
 for the maintenance of triangle queries.
The preprocessing time is $\bigO{|\inst{D}|^{\frac{3}{2}}}$
for all triangle queries.
The gray cuboid is infinite in the dimension of space complexity. 
Each point strictly included in the  gray 
cuboid corresponds 
to a combination of some space complexity,  
$\bigO{|\inst{D}|^{\frac{1}{2}-\gamma}}$
amortized update time, and
$\bigO{|\inst{D}|^{1-\gamma}}$ enumeration delay
for $\gamma  > 0$ (note that $\gamma$ may be different for 
update 
and delay). 
Due to Proposition~\ref{prop:lower_bound_triangle}, there is no
maintenance  algorithm for triangle queries that admits a trade-off
corresponding to a point in the gray  cuboid, 
unless the \OMv
conjecture fails. Each point on the  surface of the 
gray  cuboid
corresponds 
to a Pareto worst-case  optimal trade-off between the 
amortized update time
and enumeration delay. 
For $\eps=\frac{1}{2}$, \ivme needs $\bigO{|\inst{D}|^{\frac{1}{2}}}$
amortized update time and, depending on the query, an enumeration delay
such that the trade-off between these two measures is 
Pareto optimal.
For the nullary and ternary triangle queries, the delay 
is $\bigO{1}$ (Point A in Figure~\ref{fig:optimality_plot_3d}).
\ivme is  strongly Pareto worst-case optimal for these queries: There can be
no tighter upper bound for any of the update time or delay measures without loosening 
the upper bound for the other measure.
For the unary and binary triangle queries, 
the delay is  $\bigO{|\inst{D}|}$ (Point C in Figure~\ref{fig:optimality_plot_3d})
and respectively $\bigO{|\inst{D}|^{\frac{1}{2}}}$
(Point B in Figure~\ref{fig:optimality_plot_3d}). 
\ivme is only weakly Pareto 
worst-case  optimal for the unary and binary triangle queries:
There are no tighter upper bounds for both the update time and delay measures. 
Nevertheless, either the update time or the delay may still 
be lowered for the unary query without contradicting the \OMv conjecture.
As for the binary query, only the update time may be lowered,
since the delay is already below the $\bigO{|\inst{D}|}$ 
threshold from Proposition~\ref{prop:lower_bound_triangle}.

Corollary \ref{cor:ivme_optimal_triangle} summarizes the above discussion on the worst-case optimality of \ivme.

\begin{corollary}[Theorem~\ref{theo:main_result_triangle} and Proposition~\ref{prop:lower_bound_triangle}]\label{cor:ivme_optimal_triangle}
Under a single-tuple update to the database $\inst{D}$, 
\ivme with $\eps = \frac{1}{2}$ is strongly Pareto worst-case   optimal for 
the nullary and ternary triangle queries and 
weakly Pareto worst-case 
 optimal for the 
unary and binary triangle queries in the update-delay space,
unless the \OMv conjecture fails. 
\end{corollary}

\subsection{Structure of This Article}
Section~\ref{sec:preliminaries} introduces the preliminaries. Sections~\ref{sec:count} to \ref{sec:unary} introduce \ivme 
 for the nullary, ternary, binary, and unary triangle queries. 
\ivme for the nullary triangle query  needs three techniques to achieve the complexities in Theorem~\ref{theo:main_result_triangle}: delta processing, materialization of auxiliary views, and adaptive maintenance strategy depending on the degree of values in one of the columns of the input relations. 
For the ternary triangle query \ivme  additionally uses the concept of view trees. 
 \ivme for unary and  binary triangle queries exploits the degree of values
  in both columns of relations. It also uses two union  
  algorithms: one for 
  enumerating the distinct tuples in projections of views
  and one for enumerating the distinct tuples in unions of views.
   The lower bound in Proposition~\ref{prop:lower_bound_triangle} is proven in 
 Section~\ref{sec:lowerbound}.
 Section~\ref{sec:recovery} details how \ivme recovers  
existing dynamic and static approaches for triangle queries.  
Section~\ref{sec:related} relates the results of this article 
to existing work.  
Section~\ref{sec:extensions} discusses several extensions 
of \ivme. 
Conclusion and future work are given in 
Section~\ref{sec:conclusion}.

%% file: preliminaries.tex
\section{Preliminaries}
\label{sec:preliminaries}
\subsection{Data Model and Query Language}
\label{sec:data_model}
A schema $\inst{X} = (X_1, \ldots , X_n)$ is a tuple of distinct variables. 
Each variable $X_i$ has a discrete domain 
$\Dom(X_i)$. 
By $\inst{F} \subseteq \inst{X}$, we mean that 
$\inst{F}$ is a schema  that consists of a subset of the variables 
in $\inst{X}$.
A tuple 
$\inst{x}$ over schema $\inst{X}$ is an element from 
$\Dom(\inst{X}) = \Dom(X_1) \times \ldots  \times \Dom(X_n)$.
We use uppercase letters for variables and lowercase letters for data values.
Likewise, we use bold uppercase letters for schemas and bold lowercase letters for tuples of data values.

A relation $K$ over schema $\inst{X}$ is a function 
$K: \Dom(\inst{X}) \to \mathbb{Z}$ mapping tuples over
$\inst{X}$ to integers such that   
$K(\inst{x}) \neq 0$ for finitely many tuples $\inst{x}$. 
A tuple $\inst{x}$ is in 
$K$, denoted by $\inst{x} \in K$, if $K(\inst{x}) \neq 0$. 
The value $K(\inst{x})$ represents the multiplicity of $\inst{x}$ in $K$.
The size $|K|$ of $K$ is the size of the set $\{ \inst{x} \mid \inst{x} \in K \}$. 
A database $\inst{D}$ is a set of relations, and its size $|\inst{D}|$ is the sum of the sizes of the relations in $\inst{D}$.

Given a tuple $\inst{x}$ over 
schema  $\inst{X}$ and $\inst{F} \subseteq \inst{X}$, 
we write $\inst{x}[\inst{F}]$ to denote the restriction of $\inst{x}$ 
onto the variables in $\inst{F}$ such that the values in 
$\inst{x}[\inst{F}]$ follow the ordering 
in $\inst{F}$. For instance, if the 
tuple $(a,b,c)$ is over the schema $(A,B,C)$, then it holds 
$(a,b,c)[(C,A)] = (c,a)$.
For a relation $K$ over $\inst{X}$, and a tuple 
$\inst{t} \in \Dom(\inst{F})$, 
$\sigma_{\hspace{0.25mm}\inst{F} = \inst{t}} K$ denotes the set of tuples in  
$K$ that agree with $\inst{t}$ on the variables in $\inst{F}$, that is,
$\sigma_{\hspace{0.25mm}\inst{F} = \inst{t}} K = 
\{\, \inst{x} \,\mid\, \inst{x} \in K \land \inst{x}[\inst{F}] = \inst{t} \,\}$.
We write $\pi_{\hspace{0.25mm}\inst{F}}K$ to denote the set of 
restrictions of the tuples in $K$ onto $\inst{F}$, that is,
$\pi_{\hspace{0.25mm}\inst{F}}K = \{\, \inst{x}[\inst{F}] \,\mid\, \inst{x} \in K \,\}$.

\paragraph*{Query Language}
We express queries and view definitions in the language of 
functional aggregate queries (FAQ)~\cite{FAQ:PODS:2016}. 
 Compared to the original FAQ definition that uses several 
commutative semirings, we define queries over the single 
commutative ring $(\mathbb{Z},+,\ztimes,0,1)$ of integers   
with the usual addition and 
multiplication\footnote{Previous work shows how the data-intensive 
computation of different applications can be captured by application-specific rings~\cite{Nikolic:SIGMOD:18}.}. 
A query Q has one of the two forms:

\begin{enumerate}
    \item Given a set $\{X_i\}_{i \in [n]}$ of variables and an index set $S \subseteq [n]$, 
let $\inst{X}_{S}$ denote a tuple $(X_i)_{i \in S}$ of variables and 
$\inst{x}_{S}$ denote a tuple of data values over the schema $\inst{X}_{S}$.
Then,

\begin{equation*}
Q(\inst{x}_{[f]}) = \sum\limits_{x_{f+1} \in \Dom(X_{f+1})} \cdots 
\sum\limits_{x_{n} \in \Dom(X_{n})}\  \ \prod_{S\in\mathcal{M}} K_S(\inst{x}_S),\text{ where:}
\end{equation*}

\begin{itemize}
\item $\mathcal{M}$ is a multiset of index sets.

\item For every index set $S \in \mathcal{M}$, 
$K_S : \Dom(\inst{X}_S) \rightarrow \mathbb{Z}$ 
is a relation over the schema $\inst{X}_S$.

\item $\inst{X}_{[f]}$ is the tuple of free variables of $Q$.
The variables $X_{f+1}, \ldots ,X_n$ are called bound.
\end{itemize}

\item $Q(\inst{x}) = Q_1(\inst{x}) + Q_2(\inst{x})$, where $Q_1$ and $Q_2$ are queries over the same tuple of free variables.
\end{enumerate}

In the following, we use $\sum_{x_i}$ as a shorthand for $\textstyle\sum_{x_i\in \Dom(X_i)}$.

\paragraph*{Updates and Delta Queries.}
An update $\delta K$ to a relation $K$ is a relation over the schema of $K$.
A single-tuple update, written as $\delta K = \{\inst{x} \mapsto \p\}$, maps the tuple $\inst{x}$ to the nonzero multiplicity $\p\in\mathbb{Z}$ and any other tuple to $0$; that is, $|\delta K| = 1$.
The data model and query language make no distinction between inserts and deletes -- these are updates represented as relations in which tuples have positive and negative multiplicities\footnote{We restrict the multiplicities of tuples in the input relations and views to be strictly positive. Multiplicity 0 means the tuple is not present. Deletes are expressed using negative multiplicities. 
A delete request for tuple $t$ with multiplicity $-m$ is rejected if $t$'s multiplicity in the relation is less than $m$.}.

Given a query $Q$ and an update $\delta K$, the delta query $\delta Q$ defines the change in the query result after applying $\delta K$ to the database. The rules for deriving delta queries follow from the associativity, commutativity, and distributivity of the ring operations. Recall that relations and queries 
are functions mapping tuples of data values to multiplicities. 

\begin{center}
\begin{tabular}{l@{\hskip 0.4in}l}
Query $Q(\inst{x})$ & Delta query $\delta Q(\inst{x})$ \\
\midrule
$Q_1(\inst{x}_1) \ztimes Q_2(\inst{x}_2)$ & 
$\delta Q_1(\inst{x}_1) \ztimes Q_2(\inst{x}_2) +
  Q_1(\inst{x}_1)\, \ztimes \delta Q_2(\inst{x}_2) + 
  \delta Q_1(\inst{x}_1) \ztimes \delta Q_2(\inst{x}_2)$ \\[0.25ex]

$\textstyle\sum_{x} Q_1(\inst{x}_1)$ & 
$\textstyle\sum_{x} \delta Q_1(\inst{x}_1)$ \\[0.25ex]

$Q_1(\inst{x}) + Q_2(\inst{x})$ & 
$\delta Q_1(\inst{x}) + \delta Q_2(\inst{x})$ \\[0.25ex]

$K'(\inst{x})$ & 
$\delta K(\inst{x})$ when $K = K'$ and $0$ otherwise \\[0.25ex]
\end{tabular}
\end{center}

\subsection{Data Partitioning}
Our maintenance approach partitions each input relation based on the degrees of its values and uses
different maintenance strategies for values of high and low frequency.

\begin{definition}[Single Relation Partition]
\label{def:loose_relation_partition}
Given a relation $K$ over schema $\inst{X}$, a variable  $X$ from the schema $\inst{X}$, 
and a threshold $\theta$, the pair $(K^\H, K^\L)$ of relations is 
a \emph{single partition} of $K$ on $X$ with threshold $\theta$ if it satisfies the following conditions: \\[6pt]
\begin{tabular}{@{\hskip 0.5in}rl}
{(union)} & 
$K(\inst{x}) = K^\H(\inst{x}) + K^\L(\inst{x})$ for $\inst{x} \in \Dom(\inst{X})$\\[4pt]
{(domain partition)} & $\pi_{X}K^\H \cap \pi_{X}K^\L = \emptyset$ \\[4pt]
{(heavy part)} & for all $x \in \pi_{X}K^\H:\; |\sigma_{X=x} K^\H| \geq 
\frac{1}{2}\,\theta$ \\[4pt]
{(light part)} & for all $x \in \pi_{X}K^\L:\; |\sigma_{X=x} K^\L| < \frac{3}{2}\,\theta$
\end{tabular}\\[6pt]
The pair $(K^\H, K^\L)$ is called a strict partition of $K$ on $X$ with threshold 
$\theta$ if it satisfies the union and 
domain partition conditions and the following strict versions
of the heavy and light part conditions: 
\\[6pt]
\begin{tabular}{@{\hskip 0.5in}rl}
{(strict heavy part)} & for all $x \in \pi_{X}K^\H:\; |\sigma_{X=x} K^\H| \geq 
\theta$ \\[4pt]
{(strict light part)} & for all $x \in \pi_{X}K^\L:\; |\sigma_{X=x} K^\L| < \theta$
\end{tabular}\\[6pt]
The relations $K^\H$ and $K^\L$ are called the heavy and light parts of $K$.
\end{definition}

Definition~\ref{def:loose_relation_partition} admits multiple ways to (non-strictly) partition a relation
$K$ with threshold $\theta$. 
For instance, assume that $|\sigma_{X=x} K| = \theta$ for some $X$-value $x$ in $K$. 
Then, all tuples in $K$ with $X$-value $x$ can be in either the heavy or light part of $K$; but they cannot be in both parts because of the domain partition condition. If the partition is strict, then all such tuples are in the heavy part of $K$. 
The strict partition of a relation $K$  is unique for a given threshold and can be computed in time linear in the size of $K$.
\nop{
Assuming $|K|=N$ and the strict partition $(K^\H,K^\L)$ of $K$ on $X$ with threshold $\theta=N^\eps$ for $\eps\in[0,1]$, we have that:
$\forall x \in \pi_{X} K^\L: |\sigma_{X=x} K^\L| < \theta = N^\eps$;
 and $|\pi_{X} K^\H| \leq \frac{|K|}{\theta} = N^{1-\eps}$.
}

To improve the time and space complexity of our maintenance approach, we may partition input relations based on the degrees of values of two variables. 

\nop{
\begin{definition}[Double Relation Partition]
\label{def:loose_double_relation_partition}
Given a relation $K$ over schema $\inst{X}$, distinct variables $X$ and $Y$ from the schema $\inst{X}$, and a threshold $\theta$,
the tuple $(K^{hh}, K^{hl}, K^{lh}, K^{ll})$ of relations is a \emph{double partition} of $K$ on $X$ and $Y$ with threshold $\theta$ if 
$(K^{hh}, K^{hl})$ and $(K^{lh}, K^{ll})$ are partitions of $K^\H$ and respectively $K^\L$ on $Y$ with threshold $\theta$, where
$(K^{h}, K^{l})$ is a partition of $K$ on $X$ with threshold $\theta$.
\end{definition}
}

\begin{definition}[Double Relation Partition]
\label{def:loose_double_relation_partition}
Given a relation $K$ over schema $\inst{X}$, distinct variables $X$ and $Y$ from the schema $\inst{X}$, and a threshold $\theta$,
let $(K^\H_{\textsmaller{$X$}},K^\L_{\textsmaller{$X$}})$ and $(K^\H_\textsmaller{$Y$},K^\L_\textsmaller{$Y$})$ be partitions of $K$ on $X$ and respectively on $Y$ with threshold $\theta$, and let 
$K^{\H\H} = K^\H_\textsmaller{$X$}\cap K^\H_\textsmaller{$Y$}$, $K^{\H\L} = K^\H_\textsmaller{$X$}\cap K^\L_\textsmaller{$Y$}$, $K^{\L\H} = K^\L_\textsmaller{$X$}\cap K^\H_\textsmaller{$Y$}$, and $K^{\L\L} \!=\! K^\L_\textsmaller{$X$}\cap K^\L_\textsmaller{$Y$}$.
The tuple $(K^{\H\H}, K^{\H\L}, K^{\L\H}, K^{\L\L})$ is a \emph{double partition} of $K$ on $(X, Y)$ with threshold $\theta$.
\end{definition}

Let  $(K^{\H}, K^{\L})$ be a single partition 
of a relation $K$ on variable $X$ and 
$(K^{\H\H}, K^{\H\L}, K^{\L\H}, K^{\L\L})$ a double partition  
of $K$ on the pair $(X,Y)$ with some threshold $\theta$.
We say that $X$ is heavy in $K^{\H}$, $K^{\H\H}$ and 
$K^{\H\L}$ and light in 
 $K^{\L}$, $K^{\L\H}$, and $K^{\L\L}$. Similarly,
$Y$ is heavy in $K^{\H\H}$ and $K^{\L\H}$ and light in 
 $K^{\H\L}$ and $K^{\L\L}$. 
 Observe the following implications of Definitions 
 \ref{def:loose_relation_partition} and 
 \ref{def:loose_double_relation_partition} 
 to the heavy variables in relation parts.  
It holds $|\sigma_{X=x} K^\H| \geq 
\frac{1}{2}\,\theta$ for any $X$-value $x$ in $K^\H$.
However, if $K' \in \{K^{\H\H}, K^{\H\L}\}$ and $x$ is an $X$-value in $K'$, 
 this means that $|\sigma_{X = x}K| \geq \frac{1}{2}\theta$, but not necessarily 
 $|\sigma_{X = x}K'| \geq \frac{1}{2}\theta$. 
 The same holds for the degrees of $Y$-values 
 in $K^{\H\H}$ and $K^{\L\H}$.

\paragraph*{Notation.}
Our maintenance approach focuses on triangle queries and constructs auxiliary views over parts of relations $R$, $S$, and $T$.
We use an indexing scheme for such views to capture which parts of $R$, $S$, and $T$ are used in their definition.
We write $V^{rst}$ to denote a view $V$ over the parts of $R$, $S$, and $T$ specified by components $r$, $s$, and $t$, respectively.
For component $r$, 
$\H$ means $R^{\H}$;
$\L$ means $R^{\L}$;
$(\H\H)$ means $R^{\H\H}$; similarly for $(\H\L)$, $(\L\H)$, and $(\L\L)$; and
symbol $\F$ means the entire relation $R$ (i.e., the union of all parts of $R$). 
A similar convention holds for $s$ and $t$    .

For example, 
$V^{\H\H\H}$ denotes a view defined over the heavy parts of $R$, $S$, and $T$; 
$V^{\F\H\L}$ denotes a view defined over $R$, $S^H$, and $T^L$; 
$V^{(\L\H)\F\H}$ denotes a view defined over $R^{\L\H}$, $S$, and $T^H$.

\subsection{Computational Model}
\label{sec:computational_model}
We consider the RAM model of computation. 
Each relation (or materialized view) $K$ over schema $\inst{X}$ is implemented by a data structure that stores key-value entries $(\inst{x},K(\inst{x}))$ for each tuple $\inst{x}$ over $\inst{X}$ with $K(\inst{x}) \neq 0$ and 
needs space linear in the number of such tuples. 
We assume that this data structure supports
(1) looking up, inserting, and deleting entries in constant time,
(2) enumerating all stored entries in $K$ with constant delay, and
(3) returning $|K|$ in constant time.
For instance, a hash table with chaining, where entries are doubly linked for efficient enumeration, 
can support these operations in constant time on average, under the assumption of simple uniform hashing. 

Given a relation $K$ over schema $\inst{X}$ and a non-empty 
schema $\inst{F}\subset\inst{X}$, 
we assume there is an index structure on $\inst{F}$ that allows:
for any $\inst{t} \in \Dom(\inst{F})$,
(4) enumerating all entries in $K$ matching $\sigma_{\hspace{0.25mm}\inst{F}=\inst{t}}K$ with constant delay,
(5) checking $\inst{t} \in \pi_{\hspace{0.25mm}\inst{F}}K$ in constant time, and 
(6) returning $|\sigma_{\hspace{0.25mm}\inst{F}=\inst{t}}K|$ in constant time, and
(7) inserting and deleting index entries in constant time.
Such an index structure can be realized, for instance, as a hash table with chaining
where each key-value entry stores a tuple $\inst{t}$ over 
$\inst{F}$ and a doubly-linked list of pointers to the entries in $K$ having the $\inst{F}$-value $\inst{t}$.
Looking up an index entry given a tuple $\inst{t}$ over schema $\inst{F}$ takes constant time on average, and its doubly-linked list enables enumeration of the matching entries in $K$ with constant delay. 
Inserting an index entry into the hash table additionally prepends a new pointer to the doubly-linked list for a given $\inst{t}$; overall, this operation takes constant time on average.
For efficient deletion of index entries, each entry in $K$ also stores back-pointers to its index entries (as many back-pointers as there are index structures for $K$). When an entry is deleted from $K$, locating and deleting its index entries takes constant time per index.

\paragraph*{Computation Time}
Our maintenance approach first constructs a data structure that represents the result of a given triangle query on a database $\inst{D}$ and then maintains the data structure under a sequence of single-tuple updates.
In our analysis, we consider the following computation times: 
(1) the \emph{preprocessing time} is the time spent on initializing the data structure using $\inst{D}$ before any update is received,
(2) the \emph{update time} is the time spent on updating the data structure after one single-tuple update, and
(3) the \emph{enumeration delay} is the time spent until reporting the first tuple, the time between reporting two consecutive tuples, and the time between reporting the last tuple and the end of enumeration. 
For the nullary triangle query, the enumeration delay is the time spent on reporting the triangle count.
We consider two types of bounds on the update time: 
\emph{worst-case bounds}, which limit the time each individual update takes in the worst case, and
\emph{amortized worst-case bounds}, which limit the average worst-case time taken by a sequence of updates. 
When referring to sublinear time, we mean $\bigO{|\inst{D}|^{1-\gamma}}$ for some $\gamma > 0$, where $|\inst{D}|$ is the database size. 

\subsection{Enumeration Algorithms}
\subsubsection{Iterators over Materialized Views}
\label{sec:view_iterators}
Each materialized view provides the iterator interface to allow the enumeration of its tuples.
Each iterator maintains a pointer to the last reported tuple and supports two functions:
$\textsc{Next}()$ returns the next tuple in the view with a non-zero multiplicity if it exists or $\EOF$ otherwise;
$\textsc{Contains}\hspace{0.1mm}(\inst{x})$ checks if a tuple $\inst{x}$ exists in the view without altering the iterator's pointer. 
The functions $\textsc{Next}()$ and $\textsc{Contains}\hspace{0.1mm}(\inst{x})$
take constant time. 
Enumerating all tuples in a view amounts to repeatedly invoking 
the function $\textsc{Next}()$ on its iterator until reaching $\EOF$.

\subsubsection{Enumerating Unions of Sets}   
\label{sec:union_algorithm}
Given
possibly non-disjoint sets $S_1, \ldots , S_n$
the union algorithm enumerates the distinct elements in 
$\bigcup_{i \in [n]}S_i$~\cite{Durand:CSL:11}.
Figure~\ref{fig:enum_union} shows the function \textsc{UnionNext} that takes as input the iterators over $S_1, \ldots , S_n$ and based on the current iterator states (i.e., iterator pointers), returns the next element in $\bigcup_{i \in [n]}S_i$ or $\EOF$ if none.
The case $n=1$ simply returns the next element in $S_n$.
For $n=2$, the algorithm returns elements from $S_1$ only if they do not exist in $S_2$ (Line 6);
otherwise, it returns the next element from $S_2$ (Line 4). The $\textsc{Next}$ call in Line 4 always succeeds as it is made $|S_1 \cap S_2|$ times before exhausting $S_1$. 
After $S_1$ is exhausted, the algorithm returns the remaining elements from $S_2$.
The case $n>2$ is reduced to the binary case by treating $\bigcup_{i \in [n-1]}S_i$ as the first set and $S_{n}$ as the second set. 

\begin{figure}[t]
\begin{center}
\renewcommand{\arraystretch}{1.1}
\setcounter{magicrownumbers}{0}
\begin{tabular}{l}
\toprule
\textsc{UnionNext}(iterators $I_1, \ldots, I_n$) : tuple \\
\midrule
\linenumber \IF($n=1$) \RETURN $I_n.\textsc{Next}\hspace{0.1mm}()$\\
\linenumber \IF $(\,(t = \textsc{UnionNext}(I_1,\ldots, I_{n-1})) \neq \EOF\,)$ \\
\linenumber \TAB \IF $(\,I_n.\textsc{Contains}\hspace{0.1mm}(t)\,)$ \\
\linenumber \TAB\TAB \RETURN $I_n.\textsc{Next}\hspace{0.1mm}()$ \\
\linenumber \TAB \ELSE \\
\linenumber \TAB\TAB \RETURN $t$\\
\linenumber \RETURN $I_n.\textsc{Next}\hspace{0.1mm}()$\\
\bottomrule
\end{tabular}
\end{center}\vspace{-1em}
\caption{
Given iterators $I_1, \ldots, I_n$ over (possibly non-disjoint) sets $S_1, \ldots , S_n$, \textsc{UnionNext} enumerates the distinct elements in $\bigcup_{i \in [n]}S_i$. 
Each iterator $I_i$ supports two functions:
$I_i.\textsc{Next}()$ returns the next element in $S_i$ if it exists and $\EOF$ otherwise; and
$I_i.\textsc{Contains}\hspace{0.1mm}(t)$ checks whether element $t$ exists in the set $S_i$.
}

\label{fig:enum_union}
\end{figure}

\begin{lemma}
\label{lem:enum_delay}
Let $I_1, \ldots, I_n$ be iterators over sets $S_1, \ldots, S_n$, respectively, 
such that each iterator $I_i$ allows lookups in $S_i$ in time $\bigO{l}$ and enumeration of the elements in $S_i$ with delay $\bigO{d}$. 
The function $\textsc{UnionNext}(I_1, \ldots, I_n)$ enumerates $\bigcup_{i \in [n]}S_i$ with $\bigO{n l +  n d}$ delay.
\end{lemma}
\begin{proof}
The case $n=1$ follows trivially from the algorithm. 
We consider the case $n=2$.
Each element in $S_1 - S_2$ is reported from $S_1$ and all remaining elements from $S_2$; hence, each element from $S_1 \cup S_2$ is reported exactly once. 
In the worst case, we need one $\textsc{Contains}()$ call in $S_2$ and two $\textsc{Next}()$ calls before reporting the next element. 
Thus, the enumeration delay is $\bigO{l + d}$. 
The general case $n >2$ follows by simple induction.
\end{proof}

\input{pointer}

%% file: pointer.tex


An alternative method for enumerating the distinct elements in a union of sets uses skip pointers~\cite{Berkholz:ICDT:2018}.
This method allows ``jumping'' over already reported values when iterating over these sets.
To capture this idea, we first introduce the abstraction of a hop iterator, an extension of the classical iterator capable of invalidating values and omitting them during iteration.  
We then show how to enumerate the distinct elements in a union of sets using hop iterators.

\begin{figure}[t]
\begin{center}
\small
\begin{tikzpicture}


\node at(11,-3) [anchor=north west]
{
\renewcommand{\arraystretch}{1.1}
\setcounter{magicrownumbers}{0}
\begin{tabular}{@{\hskip 0.01in}l@{\hskip 0.01in}}
\toprule
\textsc{OpenHop}(\,) \\
\midrule
\linenumber $\mathit{curr} = \inst{BOF}$ \\
\bottomrule
\end{tabular}
};

\node at(0,-3) [anchor=north west]
{
\renewcommand{\arraystretch}{1.2}
\setcounter{magicrownumbers}{0}
\begin{tabular}{@{\hskip 0.01in}l@{\hskip 0.01in}}
\toprule
\textsc{NextHop}(\,): value \\
\midrule
\linenumber $\mathit{curr} = \textsc{Hop}(\inst{C}.\textsc{Next}(curr))$ \\
\linenumber \RETURN $\mathit{curr}$ \\
\bottomrule
\end{tabular}
};

\node at(5.5,-3) [anchor=north west]
{
\renewcommand{\arraystretch}{1.2}
\setcounter{magicrownumbers}{0}
\begin{tabular}{@{\hskip 0.01in}l@{\hskip 0.01in}}
\toprule
\textsc{IsEmpty}(\,): bool \\
\midrule
\linenumber $\mathit{first} = \inst{C}.\textsc{Next}(\inst{BOF})$ \\
\linenumber \RETURN $\textsc{Hop}(\mathit{first}) = \EOF$ \\
\bottomrule
\end{tabular}
};

\node at(0,-5.3) [anchor=north west]
{
\renewcommand{\arraystretch}{1.2}
\setcounter{magicrownumbers}{0}
\begin{tabular}{@{\hskip 0.01in}l@{\hskip 0.01in}}
\toprule
\textsc{Hop}(value $\mathit{x}$)$:$ value \\
\midrule
\linenumber \IF ($\mathit{x} \in \mathit{skipTo}$) \\
\linenumber \TAB \RETURN $\mathit{skipTo}[\mathit{x}]$ \\
\linenumber \RETURN $\mathit{x}$ \\
\bottomrule
\end{tabular}
};

\node at(5.5,-5.3) [anchor=north west]
{
\renewcommand{\arraystretch}{1.2}
\setcounter{magicrownumbers}{0}
\begin{tabular}{@{\hskip 0.01in}l@{\hskip 0.01in}}
\toprule
\textsc{HopBack}(value $\mathit{x}$)$:$ value \\
\midrule
\linenumber \IF ($\mathit{x} \in \mathit{skippedFrom}$) \\
\linenumber \TAB \RETURN $\mathit{skippedFrom}[\mathit{x}]$ \\
\linenumber \RETURN $\mathit{x}$ \\
\bottomrule
\end{tabular}
};

\node at(11,-4.5) [anchor=north west] 
{
\renewcommand{\arraystretch}{1.15}
\setcounter{magicrownumbers}{0}
\begin{tabular}{@{\hskip 0.01in}l@{\hskip 0.01in}}
\toprule
\textsc{Exclude}(value $\mathit{x}$) \\
\midrule
\linenumber \IF (\NOT $\inst{C}.\textsc{Contains}(\mathit{x})$) \RETURN \\
\linenumber $\mathit{to} = \textsc{Hop}(\inst{C}.\textsc{Next}(x))$ \\
\linenumber $\mathit{from} = \textsc{HopBack}(\mathit{x})$ \\
\linenumber $\mathit{skipTo}[\mathit{from}] = \mathit{to}$ \\
\linenumber $\mathit{skippedFrom}[\mathit{to}] = \mathit{from}$ \\
\bottomrule
\end{tabular}
};

\end{tikzpicture}
\end{center}
\vspace{-1em}
\caption{
Hop iterator over a collection $\inst{C}$ of values with no duplicates.
The iterator maintains a pointer $\mathit{curr}$ to the current value 
and two initially-empty dictionaries $\mathit{skipTo}$ and $\mathit{skippedFrom}$ mapping values to values. $\inst{BOF}$ and $\inst{EOF}$ represent special values before the first and after the last value in $\inst{C}$. The collection $\inst{C}$ supports $\inst{C}.\textsc{Contains}(\inst{x})$ for checking the existence of $x$ in $\inst{C}$ and $\inst{C}.\textsc{Next}(x)$ for finding the successor of $x$ in $\inst{C}$.
}
\label{fig:hop_iterator}
\end{figure}

\subsubsection{Hop Iterators over Collections}
\label{sec:hop_iterators}
Consider a collection $\inst{C}$ of values with no duplicates. The collection supports $\inst{C}.\textsc{Contains}(x)$ for checking the existence of $x$ in $\inst{C}$ and $\inst{C}.\textsc{Next}(x)$ for finding the successor of $x$ in $\inst{C}$.
An iterator over $\inst{C}$ allows enumerating the values in $\inst{C}$ using the standard Volcano-style $\textsc{Open}(\,)$ and $\textsc{Next}(\,)$ functions. In addition to that, a {\em hop iterator} can invalidate an arbitrary value $x$ in $\inst{C}$ using the $\textsc{Exclude}(x)$ function. Such invalidated values are omitted during iteration. The hop iterator also ensures a constant amount of work per reported value. 

Figure~\ref{fig:hop_iterator} defines the operations of a hop iterator over collection $\inst{C}$. The hop iterator maintains a pointer $\mathit{curr}$ to the current value in $\inst{C}$. Upon opening the iterator via $\textsc{OpenHop}(\,)$, $\mathit{curr}$ points to before the first element in $\inst{C}$, denoted by $\inst{BOF}$. The $\textsc{Next}(\,)$ function returns the next valid value from $\inst{C}$ if it exists or $\inst{EOF}$ otherwise. 
The $\textsc{Exclude}(x)$ procedure invalidates $x \in \inst{C}$ and records this information using dictionaries $\mathit{skipTo}$ and $\mathit{skippedFrom}$.
The former consists of $(x,y)$ pairs encoding that $x$ is invalid and its next value is $y$, while the latter is the inverse dictionary of the former.
$\textsc{Exclude}(x)$ computes a range of skipped values that includes $x$ but potentially also values before and after $x$, ensuring there are no consecutive ranges of skipped values. 
This property guarantees that reporting the next valid value or $\inst{EOF}$ during iteration takes constant time.

\begin{lemma}
\label{lem:enumerate_hop_iterator}
Let $\inst{C}$ be a collection of values with no duplicates that allows lookups in time $\bigO{l}$ and returns the successor of a value in time $\bigO{d}$. 
Constructing a hop iterator over $\inst{C}$ takes constant time, 
and the hop iterator
can exclude an arbitrary value from $\inst{C}$ in $\bigO{l+d}$ time and enumerate the non-excluded values from $\inst{C}$ with $\bigO{d}$ delay, using $\bigO{|\inst{C}|}$ space.
\end{lemma}
\begin{proof}
Figure~\ref{fig:hop_iterator} defines the operations of a hop iterator. 
$\textsc{OpenHop}(\,)$, $\textsc{Hop}(x)$, and $\textsc{HopBack}(x)$ run in constant time, assuming constant-time dictionary operations over $\mathit{skipTo}$ and $\mathit{skippedFrom}$. 
$\textsc{Next}(\,)$ looks for the valid successor of the current value in $\bigO{d}$ time. 
$\textsc{Exclude}(x)$ checks if $x$ exists in $\inst{C}$, finds the valid successor of $x$ in $\inst{C}$, and stores the range of skipped elements in $\bigO{l+d}$ total time. 
The iterator state includes the pointer $\mathit{curr}$ of constant size and two dictionaries, $\mathit{skipTo}$ and $\mathit{skippedFrom}$, of size at most the size of $\inst{C}$.
The pointer $\mathit{curr}$ is initialized to $\inst{BOF}$, and the two dictionaries are initially empty. 
Thus, constructing the iterator state takes constant time. 
\end{proof}

\subsubsection{Enumerating Unions of Sets using Hop Iterators}
\label{sec:skip_pointers}
We now design an iterator that uses hop iterators to enumerate the distinct elements in the union $\bigcup_{i\in[n]}S_i$ of possibly non-disjoint sets $S_1, \ldots S_n$. 
This union iterator first enumerates the elements from $S_1$, then those from $S_2 - S_1$, then those from $S_3 - S_2 - S_1$, and so on. 
Using classical iterators, this strategy would incur an enumeration delay linear in the size of these sets. 
Using hop iterators, however, this strategy can skip over already reported elements, for example, omit the elements from $S_2$ that also exist in $S_1$ when enumerating $S_2-S_1$.
The enumeration delay in this case would depend on the time needed to exclude a just reported element from those sets containing that element. 

Figure~\ref{fig:enum_union_with_pointer_new} defines the iterator for enumerating the distinct elements in the union of sets $S_1, \ldots S_n$. 
The iterator state includes a collection of hop iterators, one for each set $S_i$, called buckets, an iterator $I_{\mathit{buckets}}$ over this collection, and an iterator $I_{current}$ denoting the current hop iterator in this collection.
The $\textsc{Open}(\,)$ procedure allocates the buckets and initializes $I_{\mathit{current}}$ with the hop iterator for $S_1$. The hop iterators are lazily initialized on their first access to allow $\textsc{Open}(\,)$ to run in constant time.
The $\textsc{Next}(\,)$ function reports the next valid element using $I_{\mathit{current}}$. 
On exhausting the current iterator, $I_{\mathit{current}}$ moves on to the next bucket if it exists or returns $\inst{EOF}$ otherwise (Lines 2-6). 

For each returned element $t$, $\textsc{Next}(\,)$ also excludes $t$ from all the buckets containing $t$ (Lines 7-10).  
The $\textsc{CandidateBuckets}(t)$ function identifies the set of buckets to be examined when excluding $t$. 
This function is a parameter of the union iterator. Its default implementation returns the set $[n]$ for any element $t$, as in prior work~\cite{Berkholz:ICDT:2018}.
However, providing a context-specific implementation of this function may restrict the number of buckets that need to be examined to exclude $t$, further improving the enumeration delay,  as demonstrated in Sections~\ref{sec:enumeration_binary} and \ref{sec:enumeration_unary}.
Excluding $t$ may leave a hop iterator with no valid elements. In this case, the hop iterator itself is also excluded from the collection of hop iterators (Lines 9-10).

\begin{figure}[t]
\begin{center}
\small
\begin{tikzpicture}

\node at(0,0) [anchor=north west]
{
\renewcommand{\arraystretch}{1.2}
\setcounter{magicrownumbers}{0}
\begin{tabular}{@{\hskip 0.01in}l@{\hskip 0.01in}}
\toprule
Iterator state \\
\midrule
$\mathit{buckets}[i] = \text{iterator over elements of set } S_i, i\in[n]$ \\
$I_{\mathit{buckets}} = \text{iterator over } \mathit{buckets}$, \\
$I_{\mathit{current}} = \text{iterator over elements of current bucket}$\\
\bottomrule
\end{tabular}
};

\node at(0,-2.79) [anchor=north west]
{
\renewcommand{\arraystretch}{1.2}
\setcounter{magicrownumbers}{0}
\begin{tabular}{@{\hskip 0.01in}l@{\hskip 0.01in}}
\toprule
\textsc{Open}(\,) \\
\midrule
\linenumber $\mathit{buckets} = \text{allocate iterators for sets } \{S_i\}_{i\in[n]}$ \\ 
\linenumber $I_{\mathit{buckets}} = \text{create iterator over } \mathit{buckets}$ \\
\linenumber $I_{\mathit{buckets}}.\textsc{OpenHop}(\,)$ \\
\linenumber $I_{\mathit{current}} = I_{\mathit{buckets}}.\textsc{NextHop}(\,)$ \\
\linenumber $I_{\mathit{current}}$.\textsc{OpenHop}(\,) \\
\bottomrule
\end{tabular}
};

\node at(8.5,0) [anchor=north west] 
{
\renewcommand{\arraystretch}{1.2}
\setcounter{magicrownumbers}{0}
\begin{tabular}{@{\hskip 0.01in}l@{\hskip 0.01in}}
\toprule
\textsc{Next}(\,): tuple \\
\midrule
\linenumber  $t = I_{\mathit{current}}.\textsc{NextHop}(\,)$ \\
\linenumber  \IF ($t = \EOF$) \\
\linenumber \TAB $I_{\mathit{current}} = I_{\mathit{buckets}}.\textsc{NextHop}(\,)$ \\
\linenumber \TAB \IF ($I_{\mathit{current}} = \EOF$) \RETURN \EOF \\
\linenumber \TAB $I_\mathit{currrent}.\textsc{OpenHop}(\,)$ \\
\linenumber \TAB $t = I_\mathit{currrent}.\textsc{NextHop}(\,)$ \\
\linenumber \FOREACH $i \in \textsc{CandidateBuckets}(t)$ \\
\linenumber \TAB $\mathit{buckets}[i].\textsc{Exclude}(t)$ \\
\linenumber \TAB \IF ($\mathit{buckets}[i].\textsc{IsEmpty}(\,)$) \\
\linenumber \TAB\TAB $I_{\mathit{buckets}}.\textsc{Exclude}(\mathit{buckets}[i])$ \\
\linenumber \RETURN $t$ \\ 
\bottomrule
\end{tabular}
};
\end{tikzpicture}
\end{center}
\vspace{-1em}
\caption{
Iterator for enumerating the distinct elements in the union $\bigcup_{i \in [n]}S_i$ of (possibly non-disjoint) sets $S_1, \ldots, S_n$ using hop iterators. Each set $S_i$ is an iterable collection (bucket) of values.
The function $\textsc{CandidateBuckets}$ parameterizes the iterator and serves to restrict the set of buckets that may contain a given element $t$; the default implementation of this function returns the set $[n]$ for any element $t$.
}

\label{fig:enum_union_with_pointer_new}
\end{figure}



\begin{lemma}
\label{lem:enumerate_pointer_new}
Let $S_1, \ldots, S_n$ be collections of elements with no duplicates such that each collection $S_i$ allows lookups in time $\bigO{l}$ and returns the successor of a value in time $\bigO{d}$. 
Let $\textsc{CandidateBuckets}(t)$ be a function that returns a set $B \subseteq [n]$ in time $\bigO{b}$, for any value $t$.
Constructing an iterator as per Figure~\ref{fig:enum_union_with_pointer_new} takes constant time, 
and the iterator can enumerate the elements from $\bigcup_{i\in[n]}S_i$ with $\bigO{|B|l +|B|d + b}$ delay, using $\bigO{\sum_{i\in[n]}|S_i|}$ space. 
\end{lemma}
\begin{proof}
$\textsc{Open}(\,)$ creates a hop iterator $\mathit{buckets}[i]$ with a unique index $i$ for each collection $S_i$. The hop iterators form an array with index-based constant-time lookup and successor operations. Each hop iterator is initialized on its first access.
Opening the iterator $I_{\mathit{buckets}}$ and getting the first hop iterator from the array take constant time. Overall, $\textsc{Open}(\,)$ runs in constant time. 

The $\textsc{Next}(\,)$ function gets the next tuple from $I_{\mathit{current}}$ in $\bigO{d}$ time, per Lemma~\ref{lem:enumerate_hop_iterator} (Lines 1 and 6). Moving on to the next bucket if it exists or returning $\EOF$ otherwise take constant time (Lines 3-5). The loop (Lines 7-10) runs $|B|$ times, and each loop iteration takes $\bigO{l+d}$ time to exclude $t$ from a bucket (Line 8), $\bigO{d}$ time to check if the bucket is empty (Line 9), and constant time to exclude that bucket (Line 10), per Lemma~\ref{lem:enumerate_hop_iterator}. Given that $\textsc{CandidateBuckets}$ runs in $\bigO{b}$ time, $\textsc{Next}(\,)$ takes $\bigO{|B|l + |B|d + b}$ total time. The overall space complexity directly follows from Lemma~\ref{lem:enumerate_hop_iterator}.
\end{proof}


\begin{example}
\rm
We illustrate the iterators for enumerating unions of sets using hop iterators described in Figures~\ref{fig:hop_iterator} and~\ref{fig:enum_union_with_pointer_new}.
Given the non-materialized view $V$ with schema $(A,B)$ presented in Figure~\ref{fig:example_pointer_enum}, 
we show how a hop-based iterator can enumerate the distinct 
$B$-values in $\pi_B V$. 
We assume that the set 
$\{ \pi_{B}\sigma_{A=a_i} V \mid a_i \in \pi_AV \}$
and each set $V(a_i,B) = \pi_{B}\sigma_{A=a_i} V$ of $B$-values for $i \in [4]$
support the operators $\textsc{Next}(x)$ for returning the successor of $x$ 
and $\textsc{Contains}(x)$ for checking the existence of $x$.

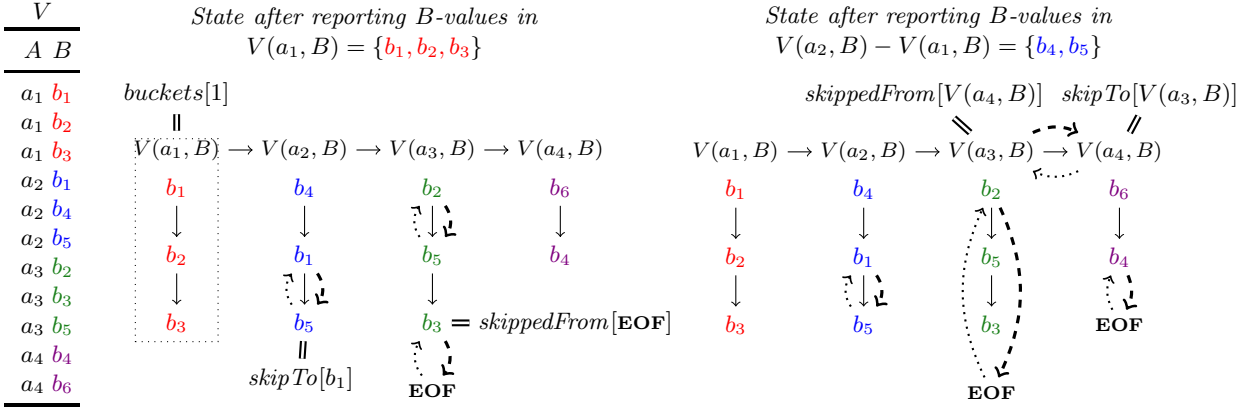
\begin{figure}[t]
\begin{center}
\begin{small}
\begin{minipage}{1cm}
\begin{center}
  \hspace{-3.8cm}
\begin{tabular}{c@{\;}c@{\;}}
  
  \multicolumn{2}{c}{$V$} \\
  \toprule
  $A$ & $B$  \\\midrule
  $a_1$ & $\color{red}b_1$\\
  $a_1$ & $\color{red}b_2$\\
  $a_1$ & $\color{red}b_3$\\
  $a_2$ & $\color{blue}b_1$\\
  $a_2$ & $\color{blue}b_4$\\
  $a_2$ & $\color{blue}b_5$\\    
  $a_3$ & $\color{goodgreen}b_2$\\
  $a_3$ & $\color{goodgreen}b_3$\\
  $a_3$ & $\color{goodgreen}b_5$\\    
  $a_4$ & $\color{violet}b_4$\\
  $a_4$ & $\color{violet}b_6$\\    
  \bottomrule
\end{tabular}
\end{center}
\end{minipage}
\hspace{0.5cm}
\begin{minipage}{11cm}
\begin{minipage}{6cm}
  \hspace{-2.2cm}
\begin{tikzpicture}
\node  at(1.3, 0.8) {};
   
      \node (start)  at(-2, 0.75) {$buckets[1]$};
            
      \draw[dotted] (-2.55,0.15) rectangle (-1.45,-2.55);

         \node   at(0.5, 1.75) {\textit{State after reporting $B$-values in}};
         \node   at(0.5, 1.35) {\textit{$V(a_1,B) = \{{\color{red} b_1,b_2,b_3}\}$}};
    \node (backskip3)  at(3.3, -2.3) {$\backskipp[\scriptsize{\EOF}]$};
    \node (skip3)  at(-0.35, -3.05) {$\skipp[b_1]$};
        
      \node (a1)  at(-2, 0.0) {\footnotesize $V(a_1,B)$};
      \node (a2) [right of = a1, node distance=1.7cm] {\footnotesize $V(a_2,B)$};
      \node (a3) [right of = a2, node distance=1.7cm] {\footnotesize $V(a_3,B)$};
      \node (a4) [right of = a3, node distance=1.7cm] {\footnotesize $V(a_4,B)$};       
      
      \node (b1) [below of = a1, node distance=0.5cm] {$\color{red}b_1$};
      \node (b2) [below of = b1, node distance=0.9cm] {$\color{red}b_2$};
      \node (b3) [below of = b2, node distance=0.9cm] {$\color{red}b_3$};     
       
      \node (b4) [below of = a2, node distance=0.5cm] {$\color{blue}b_4$};
      \node (b1') [below of = b4, node distance=0.9cm] {$\color{blue}b_1$};
      \node (b5) [below of = b1', node distance=0.9cm] {$\color{blue}b_5$};
      
      \node (b2') [below of = a3, node distance=0.5cm] {$\color{goodgreen}b_2$};
      \node (b5') [below of = b2', node distance=0.9cm] {$\color{goodgreen}b_5$};
      \node (b3') [below of = b5', node distance=0.9cm] {$\color{goodgreen}b_3$};                  
      \node (EOF3) [below of = b3', node distance=0.9cm] {$\scriptsize{\EOF}$};  
                
     \node (b6) [below of = a4, node distance=0.5cm] {$\color{violet}b_6$};
     \node (b4') [below of = b6, node distance=0.9cm] {$\color{violet}b_4$};

          
       \draw[thick,double, double distance=1pt] (start)--(a1);
 
   \draw[thick,double, double distance=1pt] (backskip3)--(b3');
   \draw[thick,double, double distance=1pt] (skip3)--(b5);
        
   \draw[->] (a1)--(a2);
   \draw[->] (a2)--(a3);
   \draw[->] (a3)--(a4);
     
        \draw[->] (b1)--(b2);
        \draw[->] (b2)--(b3);        
            
        \draw[->] (b4)--(b1');
        \draw[->] (b1')--(b5);

\draw[->] (b2')--(b5');
        \draw[->] (b5')--(b3');
        
        \draw[->] (b6)--(b4');
        
       \draw[->,very thick,dashed,bend angle=35, bend left] (b1') to (b5);
         \draw[->,thick, dotted, bend angle=35, bend left] (b5) to (b1');
        \draw[->,very thick,dashed,bend angle=35, bend left] (b2') to (b5');
         \draw[->,thick, dotted, bend angle=35, bend left] (b5') to (b2');
          
          \draw[->,very thick,dashed,bend angle=35, bend left] (b3') to (EOF3);
          \draw[->,thick,dotted,bend angle=35, bend left] (EOF3) to (b3');             
                 
\end{tikzpicture}
\end{minipage}
\hspace{-1.1cm}
\begin{minipage}{6cm}
\begin{tikzpicture}
  
    \node   at(-2.7, 1) {}; 
             \node   at(1, 2.3) {\textit{State after reporting $B$-values in}};
     \node   at(1, 1.9) {\textit{$V(a_2,B) - V(a_1,B)= \{{\color{blue} b_4,b_5}\}$}};          
    \node (backskip)  at(0.8, 1.25) {$\backskipp[V(a_4,B)]$};
    \node (skip)  at(3.8, 1.25) {$\skipp[V(a_3,B)]$};

      \node (a1)  at(-1.7, 0.5) {\footnotesize $V(a_1,B)$};
      \node (a2) [right of = a1, node distance=1.7cm] {\footnotesize $V(a_2,B)$};
      \node (a3) [right of = a2, node distance=1.7cm] {\footnotesize $V(a_3,B)$};
      \node (a4) [right of = a3, node distance=1.7cm] {\footnotesize $V(a_4,B)$};

      \node (b1) [below of = a1, node distance=0.5cm] {$\color{red}b_1$};
      \node (b2) [below of = b1, node distance=0.9cm] {$\color{red}b_2$};
      \node (b3) [below of = b2, node distance=0.9cm] {$\color{red}b_3$};

      \node (b4) [below of = a2, node distance=0.5cm] {$\color{blue}b_4$};
      \node (b1') [below of = b4, node distance=0.9cm] {$\color{blue}b_1$};
      \node (b5) [below of = b1', node distance=0.9cm] {$\color{blue}b_5$};

      \node (b2') [below of = a3, node distance=0.5cm] {$\color{goodgreen}b_2$};
      \node (b5') [below of = b2', node distance=0.9cm] {$\color{goodgreen}b_5$};
      \node (b3') [below of = b5', node distance=0.9cm] {$\color{goodgreen}b_3$};                  
      \node (EOF3) [below of = b3', node distance=0.9cm] {$\scriptsize{\EOF}$};

     \node (b6) [below of = a4, node distance=0.5cm] {$\color{violet}b_6$};
     \node (b4') [below of = b6, node distance=0.9cm] {$\color{violet}b_4$};
     \node (EOF4) [below of = b4', node distance=0.9cm] {$\scriptsize{\EOF}$};


                          
     \draw[thick, double, double distance=1pt] (1.2, 0.95)--(a3);      
         \draw[thick, double, double distance=1pt] (skip)--(a4);

   \draw[->] (a1)--(a2);
        \draw[->] (a2)--(a3);
        \draw[->] (a3)--(a4);
        
                \draw[->] (b1)--(b2);
        \draw[->] (b2)--(b3);
        
        \draw[->] (b4)--(b1');
        \draw[->] (b1')--(b5);
        
        \draw[->] (b2')--(b5');
        \draw[->] (b5')--(b3');
                       
        \draw[->] (b6)--(b4');
        
             \draw[->,very thick,dashed,bend angle=35, bend left] (b1') to (b5);
         \draw[->,thick, dotted, bend angle=35, bend left] (b5) to (b1');

        \draw[->,very thick,dashed,bend angle=25, bend left] (b2') to (EOF3);
         \draw[->,thick, dotted, bend angle=25, bend left] (EOF3) to (b2');
         
             \draw[->,very thick,dashed,bend angle=25, bend left] (a3) to (a4);
         \draw[->,thick, dotted, bend angle=25, bend left] (a4) to (a3);
         
        \draw[->,very thick,dashed,bend angle=25, bend left] (b4') to (EOF4);
         \draw[->,thick, dotted, bend angle=25, bend left] (EOF4) to (b4');   
\end{tikzpicture}
\end{minipage}

\nop{
\vspace{0.2cm}
\hspace{-0.4cm}
\begin{minipage}{11cm}
\begin{center}
\begin{tikzpicture}
      \node (2)  at(0, 0) {\textit{after exhausting $\pi_B\sigma_{A=a_1} V$}};
      \node (3)  at(0, -0.4) {$\color{red}b_1$, $\color{red}b_2$, $\color{red}b_3$};                          
      \node (2)  at(3.8, 0) {\textit{after exhausting $\pi_B\sigma_{A=a_2} V$}};
      \node (3)  at(3.8, -0.4) {$\color{blue}b_4$, $\color{blue}b_5$};                                  
      \node (2)  at(7.6, 0) {\textit{after exhausting $\pi_B\sigma_{A=a_3} V$}};
      \node (3)  at(7.6, -0.4) {$\color{violet}b_6$};                                  
    \node (3)  at(-1.5, -0.4) {\textit{reported:}};
    
\end{tikzpicture}
\end{center}
\end{minipage}
}
\end{minipage}
\end{small}
\end{center}
\vspace{-1em}
\caption{
Using a hop-based iterator to enumerate the distinct $B$-values from the non-materialized view $V$ over schema $(A,B)$. 
Solid arrows represent the successor relationship among the values of $V$.
Dotted and bold dashed arrows are hops and back hops added by the iterator during the enumeration 
of the distinct $B$-values in $\pi_{B}V$.
}

\label{fig:example_pointer_enum}
\end{figure}

Figure~\ref{fig:example_pointer_enum} visualizes two states of the hop-based iterator during the enumeration of the distinct $B$-values from the given view $V$.
A vertical or horizontal solid arrow from $x$ to $y$ means $\textsc{Next}(x) = y$. 
Dotted and bold dashed arrows visualize hops:
a dotted arrow from $x$ to $y$ represents $\mathit{skipTo}[x] = y$, 
while a bold dashed arrow from $y$ to $x$ represents $\mathit{skippedFrom}[y] = x$.

The $B$-values are reported in three stages. 
In Stage 1, the iterator 
for $\pi_BV$ reports all $B$-values paired with $a_1$;
in Stage 2, it reports all $B$-values paired with $a_2$
but not with $a_1$; in Stage 3, it reports  
all $B$-values paired with $a_4$ but not with
$a_1$, $a_2$, or $a_3$.  
Since all $B$-values paired with $a_3$
are also paired with $a_1$ or $a_2$, there is  
no stage for reporting $B$-values paired with $a_3$.
The first state in Figure~\ref{fig:example_pointer_enum}
visualizes the hop iterators at the end of Stage 1,
and the second state shows the hop iterators at the end of Stages 2 and 3.  
We explain the three stages in more detail. 

\smallskip 
\textit{Stage 1:} 
The \textsc{Open} procedure from Figure~\ref{fig:hop_iterator} initializes the iterator state by 
allocating an iterator $buckets[i]$ for each set in $\{V(a_i,B)\}_{i\in[4]}$ and
positioning $I_{\mathit{buckets}}$ at $\mathit{buckets}[1]$ and $I_{\mathit{current}}$ before $b_1$ in the bucket for $V(a_1,B)$. 
The iterator then reports $b_1$, $b_2$, and $b_3$ from $V(a_1,B)$ 
and excludes $b_1$ from $\mathit{buckets}[2]$, and $b_2$ and $b_3$ from $\mathit{buckets}[3]$ by 
adding hops to their candidate buckets. 
At the end of Stage 1, 
$\mathit{buckets}[2]$ contains
$\mathit{skipTo}[b_1] = b_5$ and
$\mathit{skippedFrom}[b_5] = b_1$,
and $\mathit{buckets}[3]$ contains
$\mathit{skipTo}[b_2] = b_5$,
$\mathit{skippedFrom}[b_5] = b_2$,
$\mathit{skipTo}[b_3] = \EOF$, and
$\mathit{skippedFrom}[\EOF] = b_3$.

\smallskip

\textit{Stage 2:} 
The iterator moves $I_{buckets}$ to $buckets[2]$ and $I_{current}$ to $b_4$ in $V(a_2, B)$.
Then, it reports the values $b_4$ and $b_5$ in $V(a_2, B)$ but skips $b_1$ using the hop at this value. 
It excludes $b_4$ from $buckets[4]$ and $b_5$ from $buckets[3]$;
for the latter, since $b_5$ has a hop back to $b_2$, and its successor $b_3$ has a hop to $\EOF$, 
the iterator connects $b_2$ and $\EOF$.
Since all the $B$-values in $buckets[3]$ are now excluded, the iterator excludes $V(a_3, B)$ from $I_{buckets}$.

\smallskip
 
\textit{Stage 3:} 
The iterator $I_{buckets}$ skips $V(a_3,B)$ and reaches $V(a_4,B)$. 
The iterator then reports $b_6$ while skipping $b_4$. 
The value $b_6$ does not appear under other $A$-value, hence, 
no hop has to be added. 
Since the set of $A$-values is exhausted, the 
iterator returns $\EOF$ and terminates.
\end{example}

%% file: count.tex
\section{Maintaining the Nullary Triangle Query}\label{sec:count}
In this section, we present our strategy for maintaining the nullary triangle query 
$$\triangle_0() = \sum_{a,b,c} R(a,b)\ztimes S(b,c) \ztimes T(c,a)$$
under a single-tuple update. 
We start with a high-level overview.
Consider a database $\db$ consisting of three relations 
$R$, $S$, and $T$ with schemas $(A,B)$, $(B,C)$, and $(C,A)$, respectively.
We partition $R$, $S$, and $T$ on variables $A$, $B$, and $C$, respectively, for 
a given threshold.
We then decompose the nullary triangle query 
into eight skew-aware views expressed over these relation parts:

\begin{align*}
\triangle_0^{rst}() = \sum\limits_{a,b,c} R^r(a,b)\ztimes S^s(b,c) \ztimes T^t(c,a), \quad\text{ for } 
r,s,t \in \{\H,\L\}.
\end{align*}
The nullary triangle query is then the sum of these skew-aware views: 
$\triangle_0() = 
\textstyle\sum_{r,s,t \in \{\H,\L\}} \triangle_0^{rst}()$.

\ivme adapts its maintenance strategy to each skew-aware view $\triangle_0^{rst}$ to allow for amortized update time that is sublinear in the database size. While most of these views may admit sublinear delta computation over the relation parts, few exceptions require linear-time maintenance in worst case. For these exceptions, \ivme precomputes the update-independent parts of the delta queries as \emph{auxiliary materialized views} and then exploits these views to speed up the delta computation. 

One such exception is the view $\triangle_0^{\H\H\L}$. Consider a single-tuple update 
$\delta R^\H = \{(\deltaA,\deltaB) \mapsto \p\}$ to the heavy part $R^\H$ of relation $R$, 
where $\deltaA$ and $\deltaB$ are fixed data values.
Computing the delta view 
$\delta \triangle_0^{\H\H\L}() = \delta R^\H(\deltaA,\deltaB) \ztimes \textstyle\sum_c  S^\H(\deltaB,c) \ztimes T^\L(c,\deltaA)$ requires iterating over all the $C$-values $c$ paired with $\deltaB$ in $S^\H$ and with $\deltaA$ in $T^\L$; the number of such $C$-values can be linear in the size of the database. To avoid this iteration, \ivme precomputes the view 
$V_{ST}(b,a) = \sum_c S^\H(b,c) \ztimes T^\L(c,a)$  and uses this view to evaluate 
$\delta \triangle_0^{\H\H\L}() = \delta R^\H(\deltaA,\deltaB) \ztimes V_{ST}(\deltaB,\deltaA)$ in constant time.

Such auxiliary views, however, also require maintenance. All such views 
created by \ivme can be maintained in sublinear time under single-tuple updates to the input relations. Figure~\ref{fig:view_definitions} summarizes these views used by \ivme to maintain the nullary triangle query: $V_{RS}$, $V_{ST}$ and $V_{TR}$. They serve to avoid linear-time delta computation for updates to $T$, $R$, and $S$, respectively. \ivme also materializes the result of the nullary triangle query, which ensures constant 
enumeration delay.

\begin{figure}[t]
  \begin{center}
    \renewcommand{\arraystretch}{1.2}  
    \begin{tabular}{@{\hskip 0.05in}l@{\hskip 0.4in}l@{\hskip 0.05in}}
      \toprule
      Materialized View Definition & Space Complexity \\    
      \midrule
      $\triangle_0() = \sum\limits_{r,s,t \in \{\H,\L\}} \,
      \sum\limits_{a,b,c} R^r(a,b) \ztimes S^s(b,c) \ztimes T^t(c,a)$ & 
      $\bigO{1}$ \\
      $V_{RS}(a,c) = \sum_{b} R^\H(a,b) \ztimes S^\L(b,c)$ & 
      $\bigO{|\inst{D}|^{1+\min{\{\,\eps, 1-\eps \,\}}}}$ \\
      $V_{ST}(b,a) = \sum_{c} S^\H(b,c) \ztimes T^\L(c,a)$ & 
      $\bigO{|\inst{D}|^{1+\min{\{\,\eps, 1-\eps \,\}}}}$ \\
      $V_{TR}(c,b) = \sum_{a} T^\H(c,a) \ztimes R^\L(a,b)$ & 
      $\bigO{|\inst{D}|^{1+\min{\{\,\eps, 1-\eps \,\}}}}$ \\
      \bottomrule    
    \end{tabular}
  \end{center}
  \caption{The definition and space complexity of the materialized views 
  $\inst{V} = \{ \triangle_0, V_{RS}, V_{ST}, V_{TR} \}$ for the nullary triangle query. 
  The set $\inst{V}$ is part of an \ivme state of a database $\inst{D}$
  partitioned for $\eps \in [0,1]$.} 
  \label{fig:view_definitions}
\end{figure}

We now describe our strategy in detail. We start by defining the state that \ivme initially creates and maintains upon each update. Then, we specify the procedure for processing a single-tuple update to any input relation, followed by the space complexity analysis of \ivme. Section~\ref{sec:rebalancing} gives the procedure for rebalancing the partitions after a sequence of such updates.

\begin{definition}[\ivme State]\label{def:ivme_state}
Let $\db=\{R,S,T\}$ be a database, $\triangle$ a triangle query
and $\eps \in [0,1]$. 
An \ivme  state of $\inst{D}$ supporting the maintenance of $\triangle$ 
is a tuple $\astate = (\eps,N, \dbeps, \inst{V})$, where: 
\begin{itemize}
\item $N$ is a natural number such that the size 
invariant $\floor{\frac{1}{4}N} \leq |\db| < N$ holds.
$N$ is called the threshold base.  
\item $\dbeps = \calR \cup \calS \cup \calT$ where $\calR$, $\calS$, and 
$\calT$ are partitions 
of the database relations $R$, $S$, and $T$, respectively,  with  threshold $\theta = N^{\eps}$.
\item $\inst{V}$ is a set of materialized views. 
\end{itemize} 
The initial state $\astate$ of $\db$ has $N = 2\ztimes|\db| + 1$ and the three partitions $\calR$, $\calS$, and $\calT$ are strict.
\end{definition}

By construction, $|\dbeps|=|\db|$. The size invariant implies $|\db| = \Theta(N)$ and, together with the heavy and light part conditions, 
it facilitates the amortized analysis of \ivme in Section~\ref{sec:amortization}.

For the nullary triangle query, the \ivme state has: the partitions $\dbeps = \{R^\H, R^\L, S^\H, S^\L, T^\H, T^\L\}$ of $R$, $S$, and $T$ on variables $A$, $B$, and $C$; and the set of materialized views $\inst{V}=\{ \triangle_0, V_{RS}, V_{ST}, V_{TR} \}$ as defined in Figure~\ref{fig:view_definitions}.
Definition~\ref{def:loose_relation_partition} provides two essential upper bounds for each relation partition in an \ivme state:
The number of distinct $A$-values in $R^\H$ is at most $\frac{N}{\frac{1}{2}N^{\eps}} = 2N^{1-\eps}$, that is, $|\pi_A R^\H| \leq 2N^{1-\eps}$, and the number of tuples in $R^\L$ with an $A$-value $a$ is less than 
$\frac{3}{2}N^{\eps}$, that is, $|\sigma_{A = a}R^\L| < \frac{3}{2}N^{\eps}$, for any $a \in \Dom(A)$. The same bounds hold for $B$-values in $\{S^\H,S^\L\}$ and $C$-values in $\{T^\H,T^\L\}$.

\subsection{Preprocessing Stage}\label{sec:preprocessing-nullary}

The preprocessing stage for the nullary triangle query constructs the initial \ivme state given a database $\db$ and $\eps\in[0,1]$.

\begin{proposition}\label{prop:preprocessing_step}
Given a database $\db$ and $\eps\in[0,1]$, constructing the initial \ivme state of $\inst{D}$ supporting the maintenance of the nullary triangle query takes $\bigO{|\db|^{\frac{3}{2}}}$ time.
\end{proposition}

\begin{proof}
We analyze the time to construct the initial state $\astate = (\eps, N, \inst{P}, \inst{V})$ of $\inst{D}$.
Retrieving the size $|\inst{D}|$ and computing $N = 2\ztimes|\db| + 1$ take constant time. 
Strictly partitioning the input relations from $\db$ using the threshold $N^{\eps}$, as described in Definition~\ref{def:loose_relation_partition}, takes $\bigO{|\db|}$ time.
Computing the result of the nullary triangle query  on $\db$ (or $\dbeps$) 
using the algorithms Leapfrog TrieJoin or Recursive-Join  takes $\bigO{|\db|^{\frac{3}{2}}}$ time~\cite{NgoPRR18}. Computing the auxiliary views $V_{RS}$, $V_{ST}$, and $V_{TR}$ takes $\bigO{|\db|^{1+\min\{\eps, 1-\eps\}}}$ time, as shown next.  
Consider the view 
$V_{RS}(a,c) = \textstyle\sum_b R^\H(a,b) \ztimes S^\L(b,c)$. 
To compute $V_{RS}$, one can iterate over all $(a,b)$ pairs in $R^\H$ and then find the $C$-values in $S^\L$ for each $b$. 
The relation part $S^\L$ contains at most $N^{\eps}$ distinct $C$-values for any $B$-value, which gives an upper bound of $|R^\H| \ztimes N^{\eps}$ on the size of $V_{RS}$. Alternatively, one can iterate over all $(b,c)$ pairs in $S^\L$ and then find the $A$-values in $R^\H$ for each $b$. The relation part $R^\H$ contains at most $N^{1-\eps}$ distinct $A$-values, which gives an upper bound of $|S^\L| \ztimes N^{1-\eps}$ on the size of $V_{RS}$. 
The number of steps needed to compute this result is upper-bounded by $\min\{\,|R^\H| \ztimes N^{\eps},\, |S^\L| \ztimes N^{1-\eps}\,\} < \min\{\, N \ztimes N^{\eps},\, N \ztimes N^{1-\eps}\,\} = N^{1+\min\{\eps,1-\eps\}}$. From $|\db| = \Theta(N)$ follows that computing $V_{RS}$ on the database partition $\inst{P}$ takes $\bigO{|\db|^{1+\min\{\eps,1-\eps\}}}$ time; the analysis for $V_{ST}$ and $V_{TR}$ is analogous.  
Note that $\max_{\eps\in[0,1]}\{1+\min\{\eps, 1-\eps\}\} = \frac{3}{2}$.
Overall, the initial state $\astate$ of $\db$ can be constructed in $\bigO{|\db|^{\frac{3}{2}}}$ time.
\end{proof}

The preprocessing stage of \ivme happens \emph{before} any update is received. In case we start from an empty database, the preprocessing cost of \ivme is $\bigO{1}$.

\subsection{Space Complexity}\label{sec:space-nullary}
We analyze the space complexity of the \ivme maintenance strategy for the nullary triangle query.
  
\begin{proposition}\label{prop:space_complexity}
Given a database $\inst{D}$ and $\eps\in[0,1]$, an \ivme state 
of $\inst{D}$ supporting the maintenance of 
the nullary triangle query takes $\bigO{|\db|^{1 +\min\{\eps,1-\eps\}}}$ space.
\end{proposition}  
\begin{proof}
We consider a state $\astate = (\eps, N, \dbeps, \inst{V})$ of database $\db$. 
$N$ and $\eps$ take constant space and $|\dbeps| = |\inst{D}|$. 
Figure~\ref{fig:view_definitions} summarizes the space complexity of the materialized views $\triangle_0$, $V_{RS}$, $V_{ST}$, and $V_{TR}$ from $\inst{V}$. 
The result of $\triangle_0$ takes constant space.
 As discussed in the proof of Proposition~\ref{prop:preprocessing_step}, to compute 
 the view $V_{RS}(a,c) = \textstyle\sum_b R^\H(a,b) \ztimes 
S^\L(b,c)$, we can use either $R^\H$ or $S^\L$ as the outer relation:
\begin{align*}
\!\!\!\!\!|V_{RS}|
\,\leq\, \min\{\, |R^\H| \ztimes\!\! \max_{b\in\pi_B S^\L}\!|\sigma_{B=b}S^\L|,\, |S^\L| \ztimes\!\! \max_{b\in\pi_B R^\H}\!|\sigma_{B=b}R^\H| \,\} 
\,<\, \min\{\, N \ztimes \frac{3}{2}N^{\eps}, N \ztimes 2N^{1-\eps} \,\}
\end{align*} 
The size of $V_{RS}$ is thus $\bigO{N^{1+\min\{\eps, 1-\eps\}}}$. From $|\db|=\Theta(N)$ follows that $V_{RS}$ takes $\bigO{|\db|^{1+\min\{\eps,1-\eps\}}}$ space; the space analysis for $V_{ST}$ and $V_{TR}$ is analogous. Overall, the state $\astate$ of $\db$ supporting the maintenance of the nullary triangle query takes $\bigO{|\db|^{1+\min\{\eps,1-\eps\}}}$ space.
\end{proof}

\subsection{Processing a Single-Tuple Update}\label{sec:single-update-nullary}
We describe the \ivme strategy for maintaining the nullary triangle query under a single-tuple update to the relation $R$. This update can affect either the heavy or light part of $R$ partitioned on $A$, 
hence we write $\delta R^r$, where $r$ stands for $\H$ or $\L$. 
We can check in constant time 
whether the update affects $R^\H$ or $R^\L$
 (cf.\@ computational model in Section~\ref{sec:computational_model}). 
The update is represented as a relation $\delta R^r=\{\, (\deltaA,\deltaB) \mapsto \p \,\}$, where $\deltaA$ and $\deltaB$ are data values and $\p\in\mathbb{Z}$. 
Due to the symmetry of the nullary triangle query and auxiliary views, updates to $S$ and $T$ are handled similarly.

\begin{figure}[t]
\begin{center}
\renewcommand{\arraystretch}{1.2}
\setcounter{magicrownumbers}{0}
\begin{tabular}{ll@{\hskip 0.25in}l@{\hspace{0.6cm}}l}
\toprule
\multicolumn{2}{l}{$\textsc{ApplyUpdate}\hspace{0.1mm}(\hspace{0.25mm} \text{update } \delta R^r,\hspace{0.25mm} \text{state } \astate \hspace{0.25mm})$}& & Time \\
\cmidrule{1-2} \cmidrule{4-4}
\rownumber & \LET $\delta R^r = \{(\deltaA,\deltaB) \mapsto \p\}$ \\
\rownumber & \LET $\astate = (\eps, N, 
\{R^\H, R^\L, S^\H, S^\L, T^\H, T^\L\},
\{\triangle_0,V_{RS}, V_{ST}, V_{TR}\})$ \\

\rownumber & $\delta \triangle_0^{r\H\H}() = \delta{R^r(\deltaA,\deltaB)} \ztimes \textstyle\sum_c 
S^\H(\deltaB,c) \ztimes T^\H(c,\deltaA)$
&& $\bigO{|\inst{D}|^{1-\eps}}$ \\

\rownumber & 
$\delta \triangle_0^{r\H\L}() = \delta{R^r(\deltaA,\deltaB)} \ztimes V_{ST}(\deltaB,\deltaA)$ & &
$\bigO{1}$ \\

\rownumber & $\delta \triangle_0^{r\L\H}() = \delta R^r(\deltaA,\deltaB) \ztimes \textstyle\sum_c 
S^\L(\deltaB,c) \ztimes T^\H(c,\deltaA)$  & &
$\bigO{|\inst{D}|^{\min{\{\eps, 1-\eps\}}}}$ \\ 

\rownumber & $\delta \triangle_0^{r\L\L}() = \delta R^r(\deltaA,\deltaB) \ztimes \textstyle\sum_c 
S^\L(\deltaB,c) \ztimes T^\L(c,\deltaA)$  & &
$\bigO{|\inst{D}|^{\eps}}$ \\

\rownumber & 
$\triangle_0() = \triangle_0() + \delta \triangle_0^{r\H\H}() + \delta \triangle_0^{r\H\L}() + \delta \triangle_0^{r\L\H}() + 
\delta \triangle_0^{r\L\L}()$ & &
$\bigO{1}$ \\

\rownumber & \IF ($r$ is $\H$) &\\  
\rownumber & \TAB  
$V_{RS}(\deltaA,c) = V_{RS}(\deltaA,c) + \delta R^\H(\deltaA,\deltaB) \ztimes 
S^\L(\deltaB,c)$  & &
$\bigO{|\inst{D}|^{\eps}}$ \\

\rownumber & \ELSE &\\
\rownumber & \TAB 
 $V_{TR}(c,\deltaB) = V_{TR}(c,\deltaB) + T^\H(c,\deltaA) \ztimes \delta R^\L(\deltaA,\deltaB)$ & &
$\bigO{|\inst{D}|^{1-\eps}}$ \\

\rownumber & $R^r(\deltaA,\deltaB) = R^r(\deltaA,\deltaB) + \delta{R}^r(\deltaA,\deltaB)$ & &
$\bigO{1}$ \\

\rownumber & \RETURN 
$\astate$ & &
 \\
\midrule
\multicolumn{2}{r}{Total update time:} & &
$\bigO{|\inst{D}|^{\max\{\eps, 1-\eps\}}}$ \\
\bottomrule
\end{tabular}
\end{center}\vspace{-1em}
\caption{
 (left) Maintaining the nullary triangle query 
 under a single-tuple update.
 \textsc{ApplyUpdate} takes as input an update $\delta R^r$ to one of the 
 parts $R^{\H}$ and $R^{\L}$ of relation $R$, 
hence $r \in \{\H,\L\}$, and 
  the current \ivme state $\astate$ of a database $\inst{D}$ partitioned using $\eps\in[0,1]$.
It returns a new state that results from applying $\delta R^r$ to $\astate$. 
Lines 3-6 compute the deltas of the affected skew-aware views, and Line 7 
maintains $\triangle_0$.
Lines 9 and 11 maintain the auxiliary views $V_{RS}$ and $V_{TR}$, respectively. Line 12 maintains the affected part $R^r$.
(right) The time complexity of computing and applying deltas.  
The evaluation strategy for computing $\delta \triangle_0^{r\L\H}$ in Line 5 may choose either 
$S^\L$ or $T^\H$ to bound $C$-values, depending on $\eps$. The total time is the maximum of all individual times. The maintenance procedures for $S$ and $T$ are similar.
}
\label{fig:applyUpdate}
\vspace{-5pt}
\end{figure}

Figure~\ref{fig:applyUpdate} gives the procedure {\textsc{ApplyUpdate}} that 
takes as input a current \ivme state $\astate$ and the update 
$\delta R^r$, and returns a new state that results from applying $\delta R^r$ to $\astate$.
The procedure computes the deltas of the skew-aware views referencing $R^r$, which are   
$\delta \triangle_0^{r\H\H}$ (Line 3), $\delta \triangle_0^{r\H\L}$ (Line 4), 
$\delta \triangle_0^{r\L\H}$ (Line 5), and $\delta \triangle_0^{r\L\L}$ (Line 6), 
and uses these deltas to maintain the 
 nullary triangle query (Line 7).
These skew-aware views are not materialized, but their 
deltas facilitate the maintenance of the nullary triangle query.  
If the update affects the heavy part $R^\H$ of $R$, the procedure maintains 
$V_{RS}$ (Line 9) and $R^\H$ (Line 12); otherwise, it maintains $V_{TR}$ (Line 11) and
$R^\L$ (Line 12).
The view $V_{ST}$ remains unchanged as it has no reference to $R^\H$ or $R^\L$.

Figure~\ref{fig:applyUpdate} also gives the time complexity of computing these deltas and applying them to $\astate$. This complexity is either constant or dependent on the number of $C$-values for which matching tuples in the parts of $S$ and $T$ have nonzero multiplicities.  

\begin{proposition}\label{prop:single_step_time}
Given a database $\db$, $\eps\in[0,1]$, and an \ivme state 
$\astate$ of $\inst{D}$ supporting the maintenance of
the nullary triangle query, \ivme maintains $\astate$ under a single-tuple update to any input relation in $\bigO{|\db|^{\max\{\eps,1-\eps\}}}$ time.
\end{proposition} 
\begin{proof}
We analyze the running time of the procedure from 
Figure~\ref{fig:applyUpdate} given a single-tuple update $\delta R^r = \{(\deltaA,\deltaB) \mapsto \p\}$ and a state 
$\astate=(\eps, N, \dbeps, \inst{V})$ of $\db$. Since the query and auxiliary views are symmetric, the analysis for updates to $S$ and $T$ is similar.

We first analyze the evaluation strategies for the deltas of the skew-aware views $\triangle_0^{rst}$:
\begin{itemize}
\item (Line 3) Computing $\delta \triangle_0^{r\H\H}$ requires summing over $C$-values ($\deltaA$ and $\deltaB$ are fixed). The minimum degree of each $C$-value in $T^\H$ is $\frac{1}{2}N^{\eps}$, which means the number of distinct $C$-values in $T^\H$ is at most $\frac{N}{\frac{1}{2}N^{\eps}} = 2N^{1-\eps}$. Thus, this delta evaluation takes $\bigO{N^{1-\eps}}$ time.

\item (Line 4) Computing $\delta \triangle_0^{r\H\L}$ requires constant-time lookups in $\delta R^r$ and $V_{ST}$.

\item (Line 5) Computing $\delta \triangle_0^{r\L\H}$ can be done in two ways, depending on $\eps$: either sum over at most $2N^{1-\eps}$ $C$-values in $T^\H$ for the given $\deltaA$ or sum over at most $\frac{3}{2}N^{\eps}$ $C$-values in $S^\L$ for the given $\deltaB$. This delta computation takes at most $\min\{2N^{1-\eps}, \frac{3}{2}N^{\eps}\}$ constant-time operations, thus $\bigO{N^{\min{\{\eps, 1-\eps\}}}}$ time.

\item (Line 6) Computing $\delta \triangle_0^{r\L\L}$ requires summing over at most $\frac{3}{2}N^{\eps}$ $C$-values in $S^\L$ for the given $\deltaB$. This delta computation takes $\bigO{N^{\eps}}$ time. 
\end{itemize}
Maintaining the nullary triangle query using these deltas takes constant time (Line 7).
The views $V_{RS}$ and $V_{TR}$ are maintained for updates to distinct parts of R.
Maintaining $V_{RS}$ requires iterating over at most $\frac{3}{2}N^{\eps}$ $C$-values in $S^\L$ for the given $\deltaB$ (Line 9);
similarly, maintaining $V_{TR}$ requires iterating over at most $2N^{1-\eps}$ $C$-values in $T^\H$ for the given $\deltaA$ (Line 11).
Finally, maintaining the part of $R$ affected by $\delta R^r$ takes constant time (Line 12).  
The total update time is 
$\bigO{\max\{1,N^{\eps}, N^{1-\eps}, N^{\min\{\eps,1-\eps\}}\}} = \bigO{N^{\max\{\eps,1-\eps\}}}$. 
From the invariant $|\db| = \Theta(N)$ follows the claimed time complexity $\bigO{|\db|^{\max\{\eps,1-\eps\}}}$.
\end{proof}

\subsection{Improving Space by Double Partitioning}
\label{sec:nullary_triangle_double_partitioning}
We show how 
the space complexity of maintaining $\triangle_0$
can be improved  
to $\bigO{|\inst{D}|^{\max\{1,\min\{1+\eps,2-2\eps\}\}}}$
by double partitioning each input relation 
(cf.\@ Proposition~\ref{prop:tighter-upper-bound-space-nullary}).
This partitioning strategy allows us to obtain tighter bounds on the sizes of the materialized views.
For $\eps=0$ and $\eps\geq\frac{1}{2}$, the space complexity becomes linear;
for $\eps=\frac{1}{3}$ it reaches its maximum $\bigO{|\inst{D}|^{4/3}}$. 
Recall that the maximum space complexity under single partitioning is $\bigO{|\inst{D}|^{3/2}}$ (Proposition~\ref{prop:space_complexity}).

We double partition the input relations $R$, $S$, and $T$ on $(A,B)$, $(B,C)$, and $(C,A)$, respectively, 
with the threshold $N^{\eps}$. 
We decompose the nullary triangle query into a union of skew-aware views:
\begin{align*}
\triangle_0^{rst}() = \sum_{a,b,c} R^r(a,b) \cdot S^s(b,c) \cdot T^t(c,a), \quad\text{ for } 
r,s,t \in \{\H,\L\}^2.
\end{align*}

Figure \ref{fig:view_definitions_double_partitioning} gives the definitions of the materialized views under double partitioning. 
Under this refined partitioning strategy, each of the auxiliary views $V_{RS}$, $V_{ST}$, and 
$V_{TR}$ has both of its free variables heavy in one of the relation parts defining the view. 
For instance, 
the view $V_{RS}$ has the free variable $A$ heavy in $R^{\H\L}$ and the free variable $C$ heavy in $S^{\L\H}$.

\label{sec:space_nullary_double_partition}
\begin{figure}[t]
  \begin{center}
    \renewcommand{\arraystretch}{1.2}  
    \begin{tabular}{@{\hskip 0.05in}l@{\hskip 0.4in}l@{\hskip 0.05in}}
      \toprule
      Materialized View Definition & Space Complexity \\    
      \midrule
      $\triangle_0() = \sum\limits_{r,s,t \in \{\H,\L\}^2} \,
      \sum\limits_{a,b,c} R^r(a,b) \ztimes S^s(b,c) \ztimes T^t(c,a)$ & 
      $\bigO{1}$ \\
      $V_{RS}(a,c) = \sum_{b} R^{\H\L}(a,b) \ztimes S^{\L\H}(b,c)$ & 
      $\bigO{|\inst{D}|^{\min\{1+\eps,2-2\eps\}}}$ \\
      $V_{ST}(b,a) = \sum_{c} S^{\H\L}(b,c) \ztimes T^{\L\H}(c,a)$ & 
      $\bigO{|\inst{D}|^{\min\{1+\eps,2-2\eps\}}}$ \\
      $V_{TR}(c,b) = \sum_{a} T^{\H\L}(c,a)  \ztimes R^{\L\H}(a,b)$ & 
      $\bigO{|\inst{D}|^{\min\{1+\eps,2-2\eps\}}}$ \\
      \bottomrule    
    \end{tabular}
  \end{center}
  \caption{The definition and space complexity of the materialized views 
  for the nullary triangle query under double partitioning. 
  The set of views are part of an \ivme state of database $\inst{D}$
  partitioned for $\eps \in [0,1]$.} 
  \label{fig:view_definitions_double_partitioning}
\end{figure}

The \ivme state supporting the maintenance 
of the nullary triangle query under double partitioning 
has the partitions $\inst{P} = \{R^r, S^s, T^t\}_{r,s,t\in\{\H,\L\}^2}$ of $R$, $S$, and $T$ on 
$(A,B)$, $(B,C)$, and $(C,A)$, respectively;
and the materialized views 
$\inst{V} = \{\triangle_0, V_{RS}, V_{ST}, V_{TR}\}$
defined in Figure~\ref{fig:view_definitions_double_partitioning}.

The complexity analysis of maintaining the nullary triangle query under double partitioning is similar to that from the proofs of Propositions~\ref{prop:preprocessing_step}, \ref{prop:space_complexity}, and \ref{prop:single_step_time}. 
The preprocessing time and the maintenance time under a single-tuple update are the same as in the case of single partitioning. But the space complexity under double partitioning is improved.

\begin{proposition}\label{prop:preprocessing_step_space_double_partition}
Let $\inst{D}$ be a database and $\eps\in[0,1]$. 
\begin{itemize}
\item The initial \ivme state with double partitioning 
for the maintenance of 
the nullary triangle query 
can be constructed in $\bigO{|\db|^{\frac{3}{2}}}$ time.

\item Any \ivme state with double partitioning for the maintenance 
of the nullary triangle query
takes 
$\bigO{|\inst{D}|^{\max\{1,\min\{1+\eps,2-2\eps\}\}}}$ space.
\end{itemize}
\end{proposition}

\begin{proof}
Consider an \ivme state $\astate = (\eps, N, \inst{P}, \inst{V})$ of $\inst{D}$
with double partitioning.
Assume first that $\astate$ is the initial \ivme state.
We analyze the time to construct $\astate$.
Retrieving the database size $|\db|$ and computing $N = 2\ztimes|\db| + 1$ take constant time. 
For each input relation, strictly partitioning on both variables and then intersecting the relation parts to form the double partition (see Definition~\ref{def:loose_double_relation_partition}) take linear time. Thus, computing the partitions from $\inst{P}$ takes linear time.
The materialized views in $\inst{V}$ can be computed in time 
 $\bigO{N^{\frac{3}{2}}}$ using the same strategies as in the proof of 
Proposition~\ref{prop:preprocessing_step} and treating $R$, $S$, and $T$ as partitioned only on $A$, $B$, and $C$, respectively. 

Now, assume that $\astate$ is {\em any} \ivme state of $\db$.
We investigate its space complexity. 
The components $\eps$ and $N$ need constant space, and $|\dbeps| = |\inst{D}|$. 
Figure~\ref{fig:view_definitions_double_partitioning} gives the definition and space complexity of each  materialized view from $\inst{V}$. The size of $\triangle_0$ is constant. 

We analyze the space complexity of the view 
$V_{RS}(a,c) = \textstyle\sum_b R^{\H\L}(a,b) \ztimes S^{\L\H}(b,c)$. 
From the proof of Proposition~\ref{prop:space_complexity} follows that the size of $V_{RS}$ under single partitioning is bounded by $\bigO{N^{1 + \min\{\eps, 1-\eps\}}}$.
The double partitioning of $R$ and $S$ tightens this upper bound.
Since $A$ is heavy in $R^{\H\L}$ and $C$ is heavy in $S^{\L\H}$, 
the number of $(A,C)$-values in the result of $V_{RS}$ is bounded by 
$2N^{1- \eps}\cdot 2N^{1- \eps} = 4N^{2-2\eps}$. 
Then, the size of $V_{RS}$ is 
$\bigO{\min\{N^{1+ \min \{\eps, 1-\eps\}}, N^{2- 2\eps}\}}$, 
which simplifies to $\bigO{N^{\min\{1+\eps,2-2\eps\}}}$ since $2-2\eps \leq 2-\eps$ for $\eps\in[0,1]$.
The analyses for $V_{ST}$ and $V_{TR}$ are similar.

Considering all the components of state $\astate$,
the size of $\astate$ is $\bigO{\max\{1, N,  N^{\min\{1+\eps,2-2\eps\}}\}}$, 
which simplifies to $\bigO{N^{\max\{1,\min\{1+\eps,2-2\eps\}\}}}$.

From $|\inst{D}| = \Theta(N)$ follows the claimed preprocessing time 
and space complexity.
\end{proof}

\begin{proposition}
\label{prop:single_step_time_double_partition}
Given a database $\db$, $\eps\in[0,1]$, and an \ivme state 
$\astate$ of $\inst{D}$ supporting the maintenance of
the nullary triangle query with double partitioning, \ivme maintains $\astate$ under a single-tuple update to any input relation in $\bigO{|\db|^{\max\{\eps,1-\eps\}}}$ time.
\end{proposition} 

\begin{proof}
Consider an \ivme state $\astate = (\eps, N, \inst{P}, \inst{V})$
and an update $\delta R^r = \{(\alpha, \beta) \mapsto \p\}$,
for $r \in \{\H,\L\}^2$.
Most deltas of the skew-aware views can be computed in time $\bigO{N^{\max\{\eps, 1-\eps\}}}$
using the same strategies as in the proof of Proposition~\ref{prop:single_step_time} and
treating the relations as single partitioned.
The refined partitioning strategy splits the problematic case involving $S^\H$ and $T^\L$ into new cases involving $S^{\H\H}$ and $S^{\H\L}$ on one side and $T^{\L\H}$ and $T^{\L\L}$ on the other side. 
We next analyze the complexity of computing the deltas in these four cases:

\begin{itemize}
\item 
Computing $\delta \triangle_0^{r(\H\H)(\L\H)}$ and $\delta \triangle_0^{r(HH)(LL)}$ 
requires summing over at most $2N^{1-\eps}$ $C$-values 
paired with $\beta$ in $S^{\H\H}$; thus, computing these deltas takes $\bigO{N^{1- \eps}}$ time.

\item 
Computing $\delta \triangle_0^{r(\H\L)(\L\L)}$ requires summing over less than $\frac{3}{2} N^{\eps}$ $C$-values paired with $\deltaA$ in $T^{\L\L}$; thus, computing this delta takes $\bigO{N^{\eps}}$ time.

\item 
Computing $\delta \triangle_0^{r(HL)(LH)}$ requires a constant-time lookup in the view $V_{ST}$ from 
Figure~\ref{fig:view_definitions_double_partitioning}.
\end{itemize}

From $|\inst{}D| = \Theta(N)$ follows that $\astate$ can be maintained in time $\bigO{|\inst{D}|^{\max\{\eps, 1-\eps\}}}$ under the single-tuple update $\delta R^r$. 
The analyses for updates to $S$ and $T$ are analogous. 
\end{proof}

\subsection{Summing Up}
Materializing the query result in the \ivme state ensures constant-delay enumeration
of the result.
Then, our main result in Theorem~\ref{theo:main_result_triangle} for the nullary triangle query follows from Propositions~\ref{prop:preprocessing_step}, \ref{prop:space_complexity}, and \ref{prop:single_step_time} shown in the previous subsections, complemented by 
Proposition~\ref{prop:amortized_update_time}, which shows that the amortized
rebalancing time is $\bigO{|\inst{D}|^{\max\{\eps, 1-\eps\}}}$.

Proposition~\ref{prop:tighter-upper-bound-space-nullary}, which 
gives
an improved space complexity for the maintenance of the 
nullary triangle query using double partitioning,     
follows from Propositions 
\ref{prop:preprocessing_step_space_double_partition},
\ref{prop:single_step_time_double_partition}, and
\ref{prop:amortized_update_time}.

%% file: full.tex
\section{Maintaining the Ternary Triangle Query}
\label{sec:full}
We now focus on the maintenance of the ternary triangle query 
$$\triangle_3(a,b,c) = R(a,b)\ztimes S(b,c) \ztimes T(c,a)$$
under a single-tuple update.
We employ a similar adaptive maintenance strategy as with the nullary triangle query. 
We first partition the relations $R$, $S$, and $T$ on variables $A$, $B$, and $C$, respectively, with the threshold $N^\eps$. 
We then decompose $\triangle_3$ into skew-aware views defined over the relation parts:
\begin{align*}
\triangle_3^{\H\H\H}(a,b,c) &= R^\H(a,b) \cdot S^\H(b,c) \cdot T^\H(c,a),\\[2pt]
\triangle_3^{\L\L\L}(a,b,c) &= R^\L(a,b) \cdot S^\L(b,c) \cdot T^\L(c,a),\\[2pt]
\triangle_3^{\F\H\L}(a,b,c) &= \sum\limits_{r\in\{\H,\L\}}R^r(a,b) \cdot S^\H(b,c) \cdot T^\L(c,a),\\
\triangle_3^{\L\F\H}(a,b,c) &= \sum\limits_{s\in\{\H,\L\}}R^\L(a,b) \cdot S^s(b,c) \cdot T^\H(c,a),\\
\triangle_3^{\H\L\F}(a,b,c) &= \sum_{t\in\{\H,\L\}}R^\H(a,b) \cdot S^\L(b,c) \cdot T^t(c,a).
\end{align*}
The result of $\triangle_3$ is the union of the disjoint results of these skew-aware views.
To enumerate the result of $\triangle_3$, we can thus enumerate the results of these views one after the other.

\begin{figure}[t!]
  \begin{minipage}[b]{\textwidth}
  \begin{center}
    \renewcommand{\arraystretch}{1.3}  
    \begin{tabular}{@{\hskip 0.05in}l@{\hskip 0.4in}l@{\hskip 0.05in}}
      \toprule
      Materialized View Definition & Space Complexity \\    
      \midrule 
      $\triangle_3^{\H\H\H}(a,b,c) = 
      R^\H(a,b) \cdot S^\H(b,c) \cdot T^\H(c,a)$ & 
      $\bigO{|\inst{D}|^{\frac{3}{2}}}$ \\
       $\triangle_3^{\L\L\L}(a,b,c) = 
      R^\L(a,b) \cdot S^\L(b,c) \cdot T^\L(c,a)$  & 
      $\bigO{|\inst{D}|^{\frac{3}{2}}}$ \\[3pt]
      
      View tree for $\triangle_3^{\H\L\F}(a,b,c) = 
      \sum_{t\in\{\H,\L\}}R^\H(a,b) \cdot S^\L(b,c) \cdot T^t(c,a)$\\

      $\TAB\TAB V_{RS}(a,b,c) = R^\H(a,b) \cdot S^\L(b,c)$ & 
      $\bigO{|\inst{D}|^{1+\min{\{\,\eps, 1-\eps \,\}}}}$ \\

      $\TAB\TAB \hat{V}_{RS}(a,c) = \sum_b V_{RS}(a,b,c)$ & 
      $\bigO{|\inst{D}|^{1+\min{\{\,\eps, 1-\eps \,\}}}}$ \\

      $\TAB\TAB V^{\H\L\F}(a,c) =
      \sum_{t\in\{\H,\L\}} \hat{V}_{RS}(a,c) \cdot T^t(c,a)$ & 
      $\bigO{|\inst{D}|}$\\[3pt]

      View tree for $\triangle_3^{\F\H\L}(a,b,c) =  
      \sum_{r\in\{\H,\L\}}R^r(a,b) \cdot S^\H(b,c) \cdot T^\L(c,a)$ & \\

      $\TAB\TAB V_{ST}(b,c,a) = S^\H(b,c) \cdot T^\L(c,a)$ & 
      $\bigO{|\inst{D}|^{1+\min{\{\,\eps, 1-\eps \,\}}}}$ \\

      $\TAB\TAB \hat{V}_{ST}(b,a) = \sum_{c} V_{ST}(b,c,a)$ & 
      $\bigO{|\inst{D}|^{1+\min{\{\,\eps, 1-\eps \,\}}}}$ \\

      $\TAB\TAB V^{\F\H\L}(a,b) =
      \sum_{r\in\{\H,\L\}}R^r(a,b) \cdot \hat{V}_{ST}(b,a)$ & 
      $\bigO{|\inst{D}|}$\\[3pt]

      View tree for $\triangle_3^{\L\F\H}(a,b,c) =  
      \sum_{s \in \{\H,\L\}} R^\L(a,b) \cdot S^s(b,c) \cdot T^\H(c,a)$ & \\

      $\TAB\TAB V_{TR}(c,a,b) = T^\H(c,a) \cdot R^\L(a,b)$ & 
      $\bigO{|\inst{D}|^{1+\min{\{\,\eps, 1-\eps \,\}}}}$ \\

      $\TAB\TAB \hat{V}_{TR}(c,b) = \sum_a V_{TR}(c,a,b)$ & 
      $\bigO{|\inst{D}|^{1+\min{\{\,\eps, 1-\eps \,\}}}}$ \\

      $\TAB\TAB V^{\L\F\H}(b,c) =  
      \sum_{s \in \{\H,\L\}} S^s(b,c) \cdot \hat{V}_{TR}(c,b)$ & 
      $\bigO{|\inst{D}|}$\\

      \bottomrule    
    \end{tabular}\vspace*{-0.5em}
  \end{center}
  \end{minipage}  
  
  \vspace{0.3cm}
  \begin{minipage}{\textwidth}
\begin{center}
    \begin{tikzpicture}
    \begin{scope}
     \node at (0,0.1) (G) {View tree for $\triangle_3^{\H\L\F}$};
      \node at (0, -0.5) (A) {$V^{\H\L\F}(a,c)$};
      \node at (-1, -1.5) (C) {$\hat{V}_{RS}(a,c)$} edge[-] (A);
      \node at (1, -1.65) (B) {$\sum\limits_{t\in\{\H,\L\}}\!\!\!\!\!T^t(c,a)$} edge[-] (A);
      \node at (-1, -2.5) (D) {$V_{RS}(a,b,c)$} edge[-] (C);
      \node at (-2, -3.5) (E) {$R^\H(a,b)$} edge[-] (D);
      \node at (0, -3.5) (F) {$S^\L(b,c)$} edge[-] (D);            
    \end{scope}

    \begin{scope}[xshift=4.5cm]
      \node at (0,0.1) (G) {View tree for $\triangle_3^{\F\H\L}$};
      \node at (0, -0.5) (A) {$V^{\F\H\L}(a,b)$};
      \node at (-1, -1.5) (C) {$\hat{V}_{ST}(b,a)$} edge[-] (A);
      \node at (1, -1.65) (B) {$\sum\limits_{r\in\{\H,\L\}}\!\!\!\!\!R^r(a,b)$} edge[-] (A);      
      \node at (-1, -2.5) (D) {$V_{ST}(b,c,a)$} edge[-] (C);
      \node at (-2, -3.5) (E) {$S^\H(b,c)$} edge[-] (D);
      \node at (0, -3.5) (F) {$T^\L(c,a)$} edge[-] (D);      
    \end{scope}

    \begin{scope}[xshift=9cm]
      \node at (0,0.1) (G) {View tree for $\triangle_3^{\L\F\H}$};
      \node at (0, -0.5) (A) {$V^{\L\F\H}(b,c)$};
      \node at (-1, -1.5) (C) {$\hat{V}_{TR}(c,b)$} edge[-] (A);
      \node at (1, -1.65) (B) {$\sum\limits_{s\in\{\H,\L\}}\!\!\!\!\!S^s(b,c)$} edge[-] (A);
      \node at (-1, -2.5) (D) {$V_{TR}(c,a,b)$} edge[-] (C);
      \node at (-2, -3.5) (E) {$T^\H(c,a)$} edge[-] (D);
      \node at (0, -3.5) (F) {$R^\L(a,b)$} edge[-] (D);     
    \end{scope}
    \end{tikzpicture}
    \end{center}
  \end{minipage}

  \caption{
  (top) The materialized views
  $\inst{V} = \{ 
  \triangle_3^{\H\H\H}, \triangle_3^{\L\L\L}, 
  V_{RS}, \hat{V}_{RS}, V^{\H\L\F},
  V_{ST}, \hat{V}_{ST}, V^{\F\H\L},  
  V_{TR}, \hat{V}_{TR}, V^{\L\F\H} \}$ 
  supporting the maintenance of the ternary triangle query.
  The set $\inst{V}$ is part of an \ivme state of database $\inst{D}$. 
  The views $\triangle_3^{\H\H\H}$ and $\triangle_3^{\L\L\L}$ are materialized, while the views
  $\triangle_3^{\H\L\F}$, $\triangle_3^{\F\H\L}$, and $\triangle_3^{\L\F\H}$ 
  allow for enumeration with constant delay using their auxiliary views denoted by indentation. 
  (bottom) The view trees supporting the maintenance and enumeration of the results of 
  $\triangle_3^{\H\L\F}$, $\triangle_3^{\F\H\L}$, and $\triangle_3^{\L\F\H}$. 
  } 
  \label{fig:view_definitions_full}
\end{figure}
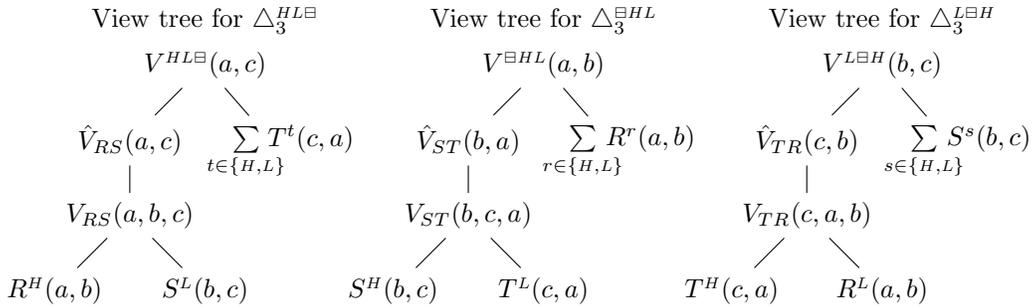

As with the nullary triangle query, \ivme  customizes the maintenance strategy 
for each of these skew-aware views and relies on auxiliary views to speed up the view maintenance.

The \ivme strategy for the nullary triangle query, however, fails to achieve sublinear maintenance time for most of these skew-aware views. 
Consider for instance the view $\triangle_3^{\F\H\L}$ and a single-tuple update 
$\delta{R^\H} = \{(\deltaA,\deltaB) \mapsto \p\}$ to the heavy part $R^\H$ of relation $R$. 
The delta $\delta \triangle_3^{\F\H\L}(\deltaA,\deltaB,c) = \delta R^\H(\deltaA,\deltaB) \cdot 
S^\H(\deltaB,c) \cdot T^\L(c,\deltaA)$ iterates over linearly many $C$-values in the worst case. 
Precomputing the view $V_{ST}(b,c,a) = S^\H(b,c) \cdot T^\L(c,a)$ and 
rewriting the delta as $\delta \triangle_3^{\F\H\L}(\deltaA,\deltaB,c) = \delta R^\H(\deltaA,\deltaB) \cdot V_{ST}(\deltaB,c,\deltaA)$ makes no improvement in the worst-case running time. In contrast, for the nullary triangle query, the view $V_{ST}(b,a) = S^\H(b,c) \cdot T^\L(c,a)$ enables computing $\delta \triangle_0^{\H\H\L}$ in constant time.

The skew-aware views of the ternary triangle query can be maintained in sublinear time by avoiding the listing (tabular) form of the view results. For that purpose, 
the result of a skew-aware view can be maintained in {\em factorized form}:
 Instead of using one materialized view, a hierarchy of materialized views
 is created  such that each of them admits sublinear maintenance time and all of them together guarantee constant-delay enumeration of the result of the skew-aware view. 
Factorized evaluation has been previously used in the context of incremental view maintenance~\cite{BerkholzKS17,Idris:dynamic:SIGMOD:2017,Nikolic:SIGMOD:18}.

Figure~\ref{fig:view_definitions_full} (top) presents the views used by \ivme to maintain the ternary triangle query under updates to the base relations. 
The results of the skew-aware views $\triangle_3^{\H\H\H}$ and $\triangle_3^{\L\L\L}$ are materialized 
in listing form. The remaining skew-aware views 
$\triangle_3^{\H\L\F}$, $\triangle_3^{\F\H\L}$, and $\triangle_3^{\L\F\H}$ avoid materialization altogether but ensure constant-delay enumeration of their results using other auxiliary materialized views (denoted by indentation).

Figure~\ref{fig:view_definitions_full} (bottom) shows for each of the skew-aware views 
$\triangle_3^{\H\L\F}$, $\triangle_3^{\F\H\L}$, and $\triangle_3^{\L\F\H}$, 
the materialized auxiliary views needed to maintain the results of the skew-aware view in factorized form. 
These auxiliary views make a view tree with input relations as leaves and updates propagating in a bottom-up manner. 
The result of $\triangle_3^{\H\L\F}$ is distributed among two auxiliary materialized views, $V^{\H\L\F}$ and $V_{RS}$. The former stores all $(a,c)$ pairs that would appear in the result of 
$\triangle_3^{\H\L\F}$, while the latter provides the matching $B$-values for each $(a,c)$ pair. 
The two views together provide constant-delay enumeration of the result of 
$\triangle_3^{\H\L\F}$. In addition to them, the view $\hat{V}_{RS}$ serves to support constant-time updates to $T^t$. The view trees for $\triangle_3^{\F\H\L}$ and $\triangle_3^{\L\F\H}$ are analogous.

The \ivme state supporting the maintenance of the ternary triangle query has the partitions
$\dbeps = \{R^\H, R^\L, S^\H, S^\L, T^\H, T^\L\}$ of $R$, $S$, and $T$ on variables $A$, $B$, and $C$; and the materialized views
$\inst{V} = \{ 
\triangle_3^{\H\H\H}, \triangle_3^{\L\L\L}, 
V_{RS}, \hat{V}_{RS}, V^{\H\L\F}, 
V_{ST}, \hat{V}_{ST}, V^{\F\H\L}, 
V_{TR}, \hat{V}_{TR}, V^{\L\F\H}\}$.

\subsection{Preprocessing Stage} 
The preprocessing stage builds the initial \ivme state $\astate = (\eps, \dbeps, \inst{V}, N)$ of database $\db$. This step partitions the input relations and computes the materialized views in $\inst{V}$ from Figure~\ref{fig:view_definitions_full} before processing any update.

\begin{proposition}\label{prop:preprocessing_step_full}
Given a database $\db$ and $\eps\in[0,1]$, constructing the initial \ivme state of $\inst{D}$ supporting the maintenance of the ternary triangle query takes $\bigO{|\db|^{\frac{3}{2}}}$ time.
\end{proposition}

\begin{proof}
Partitioning the input relations takes $\bigO{|\db|}$ time. 
The queries  $\triangle_3^{\H\H\H}$
and $\triangle_3^{\L\L\L}$ can be computed 
using a worst-case optimal join algorithm
like Leapfrog TrieJoin or Recursive-Join 
in $\bigO{|\db|^{\frac{3}{2}}}$ time~\cite{NgoPRR18}. The remaining 
skew-aware views $\triangle_3^{\H\L\F}$, $\triangle_3^{\F\H\L}$, and $\triangle_3^{\L\F\H}$ are not materialized but represented using auxiliary views. 
Consider the views in the view tree for $\triangle_3^{\H\L\F}$.
Computing $V_{RS}$ and $\hat{V}_{RS}$ takes $\bigO{|\db|^{1+\min\{\eps,1-\eps\}}}$ time, as explained in the proof of Proposition~\ref{prop:preprocessing_step}. 
The view $V^{\H\L\F}$ 
is computed by intersecting $\hat{V}_{RS}$ and $T$
in linear time. 
The same holds for the views in the view trees of $\triangle_3^{\F\H\L}$ and  $\triangle_3^{\L\F\H}$. 
Overall, the preprocessing time is $\bigO{|\db|^{\frac{3}{2}}}$.
\end{proof}

\subsection{Space Complexity}
We analyze the space complexity of the \ivme maintenance strategy for the ternary triangle query.

\begin{proposition}\label{prop:space_complexity_full}
Given a database $\inst{D}$, an \ivme state 
of $\inst{D}$ supporting the maintenance of 
the ternary triangle query takes 
$\bigO{|\db|^{\frac{3}{2}}}$ space.
\end{proposition}  

\begin{proof}
Let $\astate=(\eps, \dbeps, \inst{V}, N)$ be a state of $\db$.
The size of $\eps$ and $N$ is constant while the size of 
$\dbeps$ is $\bigO{|\inst{D}|}$.
Figure~\ref{fig:view_definitions_full} summarizes the space complexities of the materialized views in $\inst{V}$.
The size of each of the skew-aware views $\triangle_3^{\H\H\H}$ and $\triangle_3^{\L\L\L}$ is upper-bounded by $N^{\frac{3}{2}}$, the maximum number of  triangles in a database of size $N$~\cite{LW:1949}. 
The space complexity of the auxiliary views 
$V_{RS}$, $\hat{V}_{RS}$, $V_{ST}$, $\hat{V}_{ST}$, $V_{TR}$, and $\hat{V}_{TR}$ is $\bigO{N^{1+\min\{\eps,1-\eps\}}}$, as discussed in the proof of Proposition~\ref{prop:space_complexity}.
The sizes of the auxiliary views $V^{\H\L\F}$, $V^{\F\H\L}$, and $V^{\L\F\H}$ are upper-bounded by the sizes of $T$, $R$, and $S$, respectively; 
hence, these auxiliary views take $\bigO{|\inst{D}|}$ space.  
From the invariant $|\db|=\Theta(N)$ follows the claimed space complexity $\bigO{|\db|^{\frac{3}{2}}}$.
\end{proof}

\begin{figure}[t]
\begin{center}
\renewcommand{\arraystretch}{1.3}
\setcounter{magicrownumbers}{0}
\begin{tabular}{ll@{\hskip 0.05in}l@{\hspace{0.3cm}}c}
\toprule
\multicolumn{2}{l}{$\textsc{ApplyUpdate}\hspace{0.1mm}(\hspace{0.25mm} \text{update } \delta R^r,\hspace{0.25mm} \text{state } \astate \hspace{0.25mm})$}& & Time \\
\cmidrule{1-2} \cmidrule{4-4}

\rownumber & \LET $\delta R^r = \{(\deltaA,\deltaB) \mapsto \p\}$ \\
\rownumber & \LET $\astate = (\eps, N, 
\{R^\H, R^\L, S^\H, S^\L, T^\H, T^\L\},$ \\
& \TAB\TAB\TAB\SPACE $\{ 
\triangle_3^{\H\H\H}, \triangle_3^{\L\L\L}, 
V_{RS}, \hat{V}_{RS}, V^{\H\L\F}, 
V_{ST}, \hat{V}_{ST}, V^{\F\H\L}, 
V_{TR}, \hat{V}_{TR}, V^{\L\F\H}\})$ \\

\rownumber & \IF ($r$ is $\H$) &\\  

\rownumber & \TAB $\triangle_3^{\H\H\H}(\deltaA,\deltaB,c) = \triangle_3^{\H\H\H}(\deltaA,\deltaB,c) + \delta{R^\H(\deltaA,\deltaB)} \cdot S^\H(\deltaB,c) \cdot T^\H(c,\deltaA)$ & & $\bigO{|\inst{D}|^{1-\eps}}$ \\

\rownumber & \TAB $V_{RS}(\deltaA,\deltaB,c) = V_{RS}(\deltaA,\deltaB,c) + 
\delta{R^\H(\deltaA,\deltaB)} \cdot S^\L(\deltaB,c)$ & & $\bigO{|\inst{D}|^\eps}$ \\

\rownumber & \TAB $\hat{V}_{RS}(\deltaA,c) = \hat{V}_{RS}(\deltaA,c) + 
\delta{R^\H(\deltaA,\deltaB)} \cdot S^\L(\deltaB,c)$ & & $\bigO{|\inst{D}|^\eps}$ \\

\rownumber & \TAB $V^{\H\L\F}(\deltaA,c) = V^{\H\L\F}(\deltaA,c) +
\sum_{t\in\{\H,\L\}}\delta{R^\H(\deltaA,\deltaB)} \cdot S^\L(\deltaB,c) \cdot T^t(c,\deltaA)$ & & $\bigO{|\inst{D}|^\eps}$ \\


\rownumber & \ELSE &\\
\rownumber & \TAB $\triangle_3^{\L\L\L}(\deltaA,\deltaB,c) = \triangle_3^{\L\L\L}(\deltaA,\deltaB,c) + \delta{R^\L(\deltaA,\deltaB)} \cdot S^\L(\deltaB,c) \cdot T^\L(c,\deltaA)$ & & $\bigO{|\inst{D}|^{\eps}}$ \\

\rownumber & \TAB $V_{TR}(c,\deltaA,\deltaB) = V_{TR}(c,\deltaA,\deltaB) + T^\H(c,\deltaA) \cdot \delta{R^\L(\deltaA,\deltaB)}$ & & $\bigO{|\inst{D}|^{1-\eps}}$ \\

\rownumber & \TAB $\hat{V}_{TR}(c,\deltaB) = \hat{V}_{TR}(c,\deltaB) + T^\H(c,\deltaA) \cdot \delta{R^\L(\deltaA,\deltaB)}$ & & $\bigO{|\inst{D}|^{1-\eps}}$ \\

\rownumber & \TAB $V^{\L\F\H}(\deltaB,c) = V^{\L\F\H}(\deltaB,c) +  \sum_{s\in\{\H,\L\}}T^\H(c,\deltaA) \cdot \delta{R^\L(\deltaA,\deltaB)} \cdot S^s(\deltaB,c)$ & & $\bigO{|\inst{D}|^{1-\eps}}$ \\

\rownumber & $V^{\F\H\L}(\deltaA,\deltaB) = V^{\F\H\L}(\deltaA,\deltaB) + \hat{V}_{ST}(\deltaB, \deltaA) \cdot \delta{R^r(\deltaA,\deltaB)}$ & & $\bigO{1}$ \\

\rownumber & $R^r(\deltaA,\deltaB) = R^r(\deltaA,\deltaB) + 
\delta{R}^r(\deltaA,\deltaB)$ & &
$\bigO{1}$ \\

\rownumber & \RETURN $\astate$ \\

\midrule
\multicolumn{2}{r}{Total update time:} & &
$\bigO{|\inst{D}|^{\max\{\eps, 1-\eps\}}}$ \\
\bottomrule
\end{tabular}
\end{center}\vspace{-1em}
\caption{
  (left) Maintaining an \ivme state under a single-tuple update to 
  support constant-delay enumeration of the result 
  of the ternary triangle query. 
\textsc{ApplyUpdate} takes as input an update
$\delta R^r$ to the heavy or light part of $R$, hence $r\in\{\H,\L\}$, and the current 
\ivme state $\astate$ of database $\inst{D}$. 
It returns a new state that results from applying $\delta R^r$ to $\astate$.
(right) The time complexity of computing and applying deltas. 
The procedures for updates to $S$ and $T$ are similar.
}
\label{fig:applyUpdate_full}
\end{figure}

\subsection{Processing a Single-Tuple Update}
\label{sec:single_update_full}
Figure~\ref{fig:applyUpdate_full} shows the procedure for maintaining a current state $\astate$ of the ternary triangle query under an update $\delta R^r(a,b)$. 
If the update affects the heavy part $R^\H$ of $R$, the procedure maintains $\triangle_3^{\H\H\H}$ (Line~4) and propagates $\delta{R^\H}$ through the view tree for $\triangle_3^{\H\L\F}$ (Lines~5-7). If the update affects the light part $R^\L$ of $R$, the procedure maintains $\triangle_3^{\L\L\L}$ (Line~9) and propagates 
$\delta{R^\L}$ through the view tree for $\triangle_3^{\L\F\H}$ (Lines~10-12). Finally, it updates 
$V^{\F\H\L}$ (Line~13) and the part of $R$ affected by $\delta{R^r}$ (Line~14).
The views $V_{ST}$ and $\hat{V}_{ST}$ remain unchanged as they have no reference to $R^\H$ or $R^\L$.

\begin{proposition}\label{prop:single_step_time_full}
Given a database $\db$, $\eps\in[0,1]$, and an \ivme state 
$\astate$ of $\inst{D}$ supporting the maintenance of 
the ternary triangle query, \ivme maintains $\astate$ under a single-tuple update to any input relation in $\bigO{|\db|^{\max\{\eps,1-\eps\}}}$ time.
\end{proposition}

\begin{proof}
Figure~\ref{fig:applyUpdate_full} shows the time complexity of each maintenance statement in the \textsc{ApplyUpdate} procedure, for a given single-tuple update $\delta R^r = \{(\deltaA, \deltaB) \mapsto \p\}$ with $r \in \{\H,\L\}$ and a state 
$\astate=(\eps, \dbeps, \inst{V}, N)$ of $\db$. 
This complexity is determined by the number of $C$-values that need to be iterated over during computing and applying the deltas of skew-aware views. 

We first analyze the case when $\delta{R^r}$ affects the heavy part $R^\H$ of $R$.
The skew-aware view $\triangle_3^{\H\H\H}$ (Line~4) is maintained by 
iterating over $C$-values
paired with $\deltaA$ in $T^\H$ and for each such $C$-value, 
 doing constant-time lookups in the other relations and views in 
 the maintenance statement. Since $T^\H$ is heavy 
on $C$, the number of distinct $C$-values iterated over in 
$T^\H$ is at most $2N^{1-\eps}$. Hence, the maintenance 
requires $\bigO{N^{1-\eps}}$ time. 
Each of the auxiliary views $V_{RS}$, $\hat{V}_{RS}$, and $V^{\H\L\F}$ (Lines~5-7)  
is maintained by iterating over the $C$-values paired with $\deltaB$ in
$S^\L$ and doing constant-time lookups in the remaining relations 
and views in the corresponding maintenance statement. 
Since $S^\L$ is light on $B$,  the $B$-value $\deltaB$
is paired with less than $\frac{3}{2}N^{\eps}$ $C$-values in $S^\L$.
Thus, the auxiliary views $V_{RS}$, $\hat{V}_{RS}$, and $V^{\H\L\F}$
are maintained in $\bigO{N^{\eps}}$ time. 

We now consider the case when $\delta{R^r}$ affects the light part $R^\L$ of $R$. Maintaining $\triangle_3^{\L\L\L}$ (Line~9) requires iterating over less than $\frac{3}{2}N^{\eps}$ 
distinct $C$-values paired with $\deltaB$ in $S^\L$, which means that the maintenance 
requires $\bigO{N^{\eps}}$ time. Maintaining each of the auxiliary views 
 $V_{TR}$, $\hat{V}_{TR}$, and $V^{\L\F\H}$ (Line~10) requires 
 iterating over at most $2N^{1-\eps}$ distinct $C$-values paired with $\deltaA$ in $T^\H$. 
 Thus, these views can be maintained in time $\bigO{N^{1-\eps}}$.

Maintaining $V^{\F\H\L}$ and the part of $R$ affected by $\delta R^r$ takes constant time. 
Then, the total execution time of the procedure \textsc{ApplyUpdate}
in Figure~\ref{fig:applyUpdate_full} is $\bigO{N^{\max\{\eps,1-\eps\}}}$. 
From the invariant $|\db| = \Theta(N)$ follows the claimed time complexity $\bigO{|\db|^{\max\{\eps,1-\eps\}}}$. Due to the symmetry of the triangle query, the analysis for updates to parts of relations $S$ and $T$ is similar.
\end{proof}

\subsection{Enumeration Delay}

The materialized views stored in an \ivme state allow us to enumerate the tuples in the
result of the ternary triangle query with constant delay.

\begin{proposition}\label{prop:delay_full}
Given an \ivme state
$\astate$ supporting the maintenance of 
the ternary triangle query, \ivme enumerates
 the result of the query from $\astate$
with  $\bigO{1}$ delay.
\end{proposition}

\begin{proof}
The results of skew-aware views are disjoint, so the result of the ternary triangle query can be enumerated by enumerating the result 
of each skew-aware view, one after the other. 
Since the number of such skew-aware views is independent of the data size, 
it suffices to show that the result of each skew-aware view can be enumerated 
with constant delay to achieve an overall constant delay enumeration for the ternary triangle query.

The results of the skew-aware views $\triangle_3^{\H\H\H}$ and $\triangle_3^{\L\L\L}$ are materialized using the listing representation, so they admit constant-delay enumeration.
 
We next focus on the enumeration of the result of the skew-aware view $\triangle_3^{\H\L\F}$.
The remaining skew-aware views, $\triangle_3^{\F\H\L}$ and $\triangle_3^{\L\F\H}$, are treated similarly.  
The enumeration of the result of $\triangle_3^{\H\L\F}$ is supported 
by the materialized views in its view tree from Figure~\ref{fig:view_definitions_full} (left).  
The root $V^{\H\L\F}$  materializes the set of all tuples $(a,c)$
in the projection of the result of $\triangle_3^{\H\L\F}$ onto $(A,C)$. 
The view $V_{RS}$ serves to retrieve all $B$-values in the result that are paired with a given tuple $(a,c)$.
Thus, enumerating the result of $\triangle_3^{\H\L\F}$ requires iterating over the $(A,C)$-values in $V^{\H\L\F}$, and for each such tuple $(a,c)$, iterating over the $B$-values paired with $(a,c)$ in $V_{RS}$.
Based on our computational model (see Section~\ref{sec:computational_model}), 
the $B$-values paired with $(a,c)$ in $V_{RS}$ are enumerable with constant delay.  
For each obtained triple $(a,b,c)$, \ivme retrieves the correct multiplicity by looking up 
the multiplicities of the tuples $(a,b)$, $(b,c)$, and $(c,a)$ in the leaf relations
$R^\H$, $S^\L$, and $T$ (i.e., the sum of the multiplicities of $(c,a)$ in $T^\H$ and $T^\L$), respectively, and multiplying them. 
These lookups are constant-time operations. Hence, the overall enumeration delay is constant.  
\end{proof}

\subsection{Summing Up}

Our main result in Theorem~\ref{theo:main_result_triangle} for the ternary triangle query
follows from Propositions~\ref{prop:preprocessing_step_full}, \ref{prop:space_complexity_full}, \ref{prop:single_step_time_full}, and \ref{prop:delay_full} shown in the previous subsections, complemented by 
Proposition~\ref{prop:amortized_update_time}, which shows that the amortized
rebalancing time is $\bigO{|\inst{D}|^{\max\{\eps, 1-\eps\}}}$.

%% file: binary.tex

\section{Maintaining the Binary Triangle Query}
\label{sec:binary}

\begin{figure}[t!]
  \begin{minipage}[b]{\textwidth}
  \begin{center}
    \renewcommand{\arraystretch}{1.3}  
    \begin{tabular}{@{\hskip 0.0in}l@{\hskip 0.3in}l@{\hskip 0.0in}}
      \toprule
      Materialized View Definition & Space Complexity \\    
      \midrule 
      $\triangle_2^{\H\H\H}(a,b) = 
      \sum_{s,t\in\{\H,\L\}}\sum_{c} R^\H(a,b) \cdot S^{\H{s}}(b,c) \cdot T^{\H{t}}(c,a)$ & 
      $\bigO{|\inst{D}|^{\min\{1,2-2\eps\}}}$ \\[2pt]

      $\triangle_2^{\L\L\L}(a,b) = 
      \sum_{s,t\in\{\H,\L\}}\sum_{c} R^\L(a,b) \cdot S^{\L{s}}(b,c) \cdot T^{\L{t}}(c,a)$  & 
      $\bigO{|\inst{D}|}$ \\[2pt]

      $\triangle_2^{\H(\L\L)\F}(a,b) = 
      \sum_{t\in\{\H,\L\}^2}\sum_{c} R^\H(a,b) \cdot S^{\L\L}(b,c) \cdot T^t(c,a)$  & 
      $\bigO{|\inst{D}|}$ \\[3pt]
      
            $\triangle_2^{\L\F(\H\H)(a,b) = 
      \sum_{s\in\{\H,\L\}^2}\sum_{c} R^\L(a,b) \cdot S^{s}(b,c) \cdot T^{\H\H}(c,a)}$  & 
      $\bigO{|\inst{D}|}$ \\[3pt]
      
      View tree for $\triangle_2^{\H(\L\H)\F}(a,b) = 
      \sum_{t\in\{\H,\L\}^2}\sum_{c} R^\H(a,b) \cdot S^{\L\H}(b,c) \cdot T^t(c,a)$\\[2pt]

      $\TAB\TAB V_{RS}(a,b,c) = R^\H(a,b) \cdot S^{\L\H}(b,c)$ & 
      $\bigO{|\inst{D}|^{1+\min{\{\,\eps, 1-\eps \,\}}}}$ \\[2pt]

      $\TAB\TAB \hat{V}_{RS}(a,c) = \sum_b V_{RS}(a,b,c)$ & 
      $\bigO{|\inst{D}|^{1+\min{\{\,\eps, 1-\eps \,\}}}}$ \\[2pt]

      $\TAB\TAB V^{\H(\L\H)\F}(a,c) =
      \sum_{t\in\{\H,\L\}^2}\hat{V}_{RS}(a,c) \cdot T^t(c,a)$ & 
      $\bigO{|\inst{D}|}$\\[2pt]

      $\TAB\TAB \hat{V}^{\H(\L\H)\F}(c) =
      \sum_{a} V^{\H(\L\H)\F}(a,c)$ & 
      $\bigO{|\inst{D}|^{1-\eps}}$\\[3pt]

      View tree for $\triangle_2^{\F\H\L}(a,b) =  
      \sum_{r,s\in\{\H,\L\}}\sum_{c} R^r(a,b) \cdot S^{\H{s}}(b,c) \cdot T^\L(c,a)$ & \\[2pt]

      $\TAB\TAB V_{ST}(b,a) = \sum_{s,t\in\{\H,\L\}} \sum_{c} S^{\H{s}}(b,c) \cdot T^{\L{t}}(c,a)$ & 
      $\bigO{|\inst{D}|^{1+\min{\{\,\eps, 1-\eps \,\}}}}$ \\[2pt]

      $\TAB\TAB V^{\F\H\L}(a,b) =
      \sum_{r\in\{\H,\L\}} R^r(a,b) \cdot V_{ST}(b,a)$ & 
      $\bigO{|\inst{D}|}$\\[3pt]

      View tree for $\triangle_2^{\L\F(\H\L)}(a,b) =  
       \sum_{s\in\{\H,\L\}^2} \sum_{c} R^\L(a,b) \cdot S^s(b,c) \cdot T^{\H\L}(c,a)$ & \\[2pt]

      $\TAB\TAB  V_{TR}(c,a,b) = T^{\H\L}(c,a) \cdot R^\L(a,b)$ & 
      $\bigO{|\inst{D}|^{1+\min{\{\,\eps, 1-\eps \,\}}}}$ \\[2pt]

      $\TAB\TAB \hat{V}_{TR}(c,b) = \sum_a V_{TR}(c,a,b)$ & 
      $\bigO{|\inst{D}|^{1+\min{\{\,\eps, 1-\eps \,\}}}}$ \\[2pt]

      $\TAB\TAB V^{\L\F(\H\L})(b,c) =  
      \sum_{s\in\{\H,\L\}^2} S^s(b,c) \cdot \hat{V}_{TR}(c,b)$ & 
      $\bigO{|\inst{D}|}$\\[2pt]

      $\TAB\TAB \hat{V}^{\L\F(\H\L})(c) =  
      \sum_{b} V^{\L\F\H}(b,c)$ & 
      $\bigO{|\inst{D}|^{1-\eps}}$\\[3pt]

      \bottomrule    
    \end{tabular}
  \end{center}
  \end{minipage}
      
   \vspace{0.2cm}
  \begin{minipage}{\textwidth}
  \begin{center}
    \begin{tikzpicture}
    \begin{scope}[xshift=-1cm]
      \node at (0,1.1) (G) {View tree for $\triangle_2^{\H(\L\H)\F}$};
      \node at (0, 0.5) (root) {$\hat{V}^{\H(\L\H)\F}(c)$};
      \node at (0, -0.5) (A) {$V^{\H(\L\H)\F}(a,c)$} edge[-] (root);
      \node at (-1, -1.5) (C) {$\hat{V}_{RS}(a,c)$} edge[-] (A);
      \node at (1, -1.65) (B) {$\sum\limits_{t\in\{\H,\L\}^2}\!\!\!\!\!\!T^t(c,a)$} edge[-] (A);
      \node at (-1, -2.5) (D) {$V_{RS}(a,b,c)$} edge[-] (C);
      \node at (-2, -3.5) (E) {$R^\H(a,b)$} edge[-] (D);
      \node at (0, -3.5) (F) {$S^{\L\H}(b,c)$} edge[-] (D);            
    \end{scope}

    \begin{scope}[xshift=4cm,yshift=1cm]
      \node at (0,0.1) (G) {View tree for $\triangle_2^{\F\H\L}$};
      \node at (0, -0.5) (A) {$V^{\F\H\L}(a,b)$};
      \node at (1, -1.65) (B) {$\sum\limits_{r\in\{\H,\L\}}\!\!\!\!\!R^r(a,b)$} edge[-] (A);      
      \node at (-1, -1.5) (D) {$V_{ST}(b,a)$} edge[-] (A);
      \node at (-2, -2.65) (E) {$\sum\limits_{s\in\{\H,\L\}}\!\!\!\!\!S^{\H{s}}(b,c)$} edge[-] (D);
      \node at (0, -2.65) (F) {$\sum\limits_{t\in\{\H,\L\}}\!\!\!\!\!T^{\L{t}}(c,a)$} edge[-] (D);      
    \end{scope}

    \begin{scope}[xshift=8.3cm]
      \node at (0,1.1) (G) {View tree for $\triangle_2^{\L\F(\H\L)}$};
      \node at (0, 0.5) (root) {$\hat{V}^{\L\F(\H\L)}(c)$};
      \node at (0, -0.5) (A) {$V^{\L\F(\H\L)}(b,c)$} edge[-] (root);
      \node at (-1, -1.5) (C) {$\hat{V}_{TR}(c,b)$} edge[-] (A);
      \node at (1, -1.65) (B) {$\sum\limits_{s\in\{\H,\L\}^2}\!\!\!\!\!S^s(b,c)$} edge[-] (A);
      \node at (-1, -2.5) (D) {$V_{TR}(c,a,b)$} edge[-] (C);
      \node at (-2, -3.5) (E) {$T^{\H\L}(c,a)$} edge[-] (D);
      \node at (0, -3.5) (F) {$R^\L(a,b)$} edge[-] (D);     
    \end{scope}
    \end{tikzpicture}
    \end{center}
\end{minipage}

\vspace{-0.35em}
  \caption{
    (top) 
    The materialized views
    $\inst{V} = \{ 
    \triangle_2^{\H\H\H},$ $\triangle_2^{\L\L\L},$ $\triangle_2^{\H(\L\L)\F},$
      $\triangle_2^{\L\F(\H\H)},$
    $V_{RS}, \hat{V}_{RS},$ $V^{\H(\L\H)\F},$ $\hat{V}^{\H(\L\H)\F},$ 
    $V_{ST},$ $V^{\F\H\L},$
    $V_{TR},$ $\hat{V}_{TR},$ $V^{\L\F(\H\L)},$ $\hat{V}^{\L\F(\H\L)} \}$ 
    supporting the maintenance of the binary triangle query.
    The set $\inst{V}$ is part of an \ivme state of database $\inst{D}$. 
    (bottom) The view trees supporting the maintenance and enumeration of the results of
    $\triangle_2^{\H(\L\H)\F}$, $\triangle_2^{\F\H\L}$, and $\triangle_2^{\L\F(\H\L)}$.  }

  \label{fig:view_definitions_binary}
\end{figure}
We now consider the maintenance of the binary triangle query 
$$\triangle_2(a,b) = \sum_{c} R(a,b)\ztimes S(b,c) \ztimes T(c,a)$$
under a single-tuple update. 
Compared to the strategy for the ternary triangle query, the maintenance 
of the binary query faces two new challenges.
First, the results of the skew-aware views are not disjoint anymore, which causes difficulties in the enumeration of distinct $(A,B)$-values with correct multiplicities.
Second, among the view trees created for the ternary triangle query from Figure~\ref{fig:view_definitions_full}, only the view tree for $\triangle_3^{\F\H\L}$ allows constant-delay enumeration of $(A,B)$-values, while the view trees for 
$\triangle_3^{\H\L\F}$ and $\triangle_3^{\L\F\H}$ allow constant-delay enumeration of $(A,C)$- and respectively $(B,C)$-values but not $(A,B)$-values. 

To overcome the first difficulty, we use the union algorithm~\cite{Durand:CSL:11}
presented in Section~\ref{sec:union_algorithm}.
We modify this algorithm to report distinct tuples in the union of the skew-aware views together with their multiplicity. 
Since the number of skew-aware views is independent of the data size, the overall enumeration delay is the maximum delay of the individual skew-aware views. 

To overcome the second difficulty, we observe that the view trees for $\triangle_3^{\H\L\F}$ and $\triangle_3^{\L\F\H}$ from Figure~\ref{fig:view_definitions_full} both support constant-time lookups and constant-delay enumeration of $(A,B)$-values for a fixed $C$-value.
Based on this observation, we can decompose each of the two view trees into a union of view trees instantiated for the distinct $C$-values appearing at its root view. 
For each union of instantiated view trees,
we can use the union algorithm to enumerate the distinct $(A, B)$ pairs with the delay that is linear in the number of these view trees, that is, the number of distinct $C$-values at the root view.
In the view tree for $\triangle_3^{\H\L\F}$, the number of distinct $C$-values at the root can be linear in the database size; thus, the enumeration delay for $\triangle_3^{\H\L\F}$ is $\bigO{N}$. 
In the view tree for $\triangle_3^{\L\F\H}$, the number of distinct $C$-values is at most $2N^{1-\eps}$ due to the heavy part condition on $C$ in $T^\H$; thus, the enumeration
delay for $\triangle_3^{\L\F\H}$ is $\bigO{N^{1-\eps}}$. Overall, the enumeration delay in this case is linear.

We can improve this enumeration delay using the enumeration algorithm with hop iterators described in Section~\ref{sec:skip_pointers}. 
In this case, this algorithm can enumerate the distinct $(A, B)$ pairs with the delay determined by the $\textsc{CandidateBuckets}$ function, see Lemma~\ref{lem:enumerate_pointer_new}. 
The $\textsc{CandidateBuckets}$ function takes any $(A,B)$-value and returns a set of indices that identify the instantiated view trees that may contain the given $(A,B)$-value. The default implementation of this function considers all such view trees, but exploiting the skew information can asymptotically reduce their number.
For the view tree for $\triangle_3^{\H\L\F}$ and a fixed $(A,B)$-value, 
$\textsc{CandidateBuckets}$ can compute the matching $C$-values in the materialized view $V_{RS}$ joining $R^\H$ and $S^\L$ and retain only those $C$-values that exist in the root $V^{\H\L\F}$.
For a fixed $(A,B)$-value, the number of such $C$-values is less than $\frac{3}{2}N^{\eps}$ due to the light part condition on $B$ in $S^\L$, which gives the $\bigO{N^{\eps}}$ enumeration delay for the view $\triangle_3^{\H\L\F}$.
Similarly, for the view tree for $\triangle_3^{\L\F\H}$ and a fixed $(A,B)$-value,
$\textsc{CandidateBuckets}$ can compute the matching $C$-values in the materialized view $V_{TR}$ joining $T^\H$ and $R^\L$ and retain only those $C$-values that exist in the root $V^{\L\F\H}$.
The number of such $C$-values is at most $2N^{1-\eps}$ due to the heavy part condition on $C$ in $T^{\H}$, 
which gives the $\bigO{N^{1-\eps}}$ enumeration delay for the view $\triangle_3^{\L\F\H}$.  
Overall, the enumeration algorithm with hop pointers in this case gives $\bigO{N^{\max\{\eps,1-\eps\}}}$ delay.

To further improve the enumeration delay to $\bigO{N^{\min\{\eps, 1-\eps\}}}$
in both cases, we refine our partitioning strategy to use 
double partitioning for $S$ on $(B,C)$ and for $T$ on $(C,A)$.
This refinement allows us to further decompose 
the skew-aware view $\triangle_3^{\H\L\F}$ into two parts: 
one part that involves $S^{\L\H}$ and ensures the number of distinct $C$-values paired with any $(A,B)$-value, thus also the enumeration delay, 
is $\bigO{N^{\min\{\eps, 1-\eps\}}}$; and
another part that involves $S^{\L\L}$ and ensures the number 
of $B$-values paired with any $C$-value in $S^{\L\L}$ is $\bigO{N^\eps}$, which enables the materialization of this refined skew-aware view and enumeration with constant delay.
Similarly, 
we decompose the skew-aware view $\triangle_3^{\L\F\H}$
into one part that involves $T^{\H\L}$ and guarantees   
$\bigO{N^{\min\{\eps, 1-\eps\}}}$ enumeration delay,
and another part that involves $T^{\H\H}$ and enables 
its materialization and constant-delay enumeration.
Overall, our maintenance strategy for the binary triangle query 
that uses double partitioning for $S$ and $T$ 
achieves $\bigO{N^{\min\{\eps, 1-\eps\}}}$ enumeration delay.

We explain the \ivme strategy for the binary triangle query in more detail. 
The strategy uses single partitioning for relation  $R$
and double partitioning for relations $S$ and $T$.
The partition threshold is the same as for the nullary triangle query.
Figure~\ref{fig:view_definitions_binary} shows the definition and space complexity of the views supporting the maintenance of the binary triangle query. 
The skew-aware views $\triangle_2^{\H\H\H}$, $\triangle_2^{\L\L\L}$, 
$\triangle_2^{\H(\L\L)\F}$, $\triangle_2^{\L\F(\H\H)}$, and 
$\triangle_2^{\F\H\L}$ are materialized and enumerable with constant delay. 
The views $\triangle_2^{\H(\L\H)\F}$ and $\triangle_2^{\L\F(\H\L)}$ are represented as view trees consisting of auxiliary views that support the maintenance and enumeration of the results of 
$\triangle_2^{\H(\L\H)\F}$ and 
$\triangle_2^{\L\F(\H\L)}$.

The \ivme state supporting the maintenance of the binary triangle query has the partitions
$\dbeps = \{R^\H, R^\L, 
\\
S^{\H\H}, S^{\H\L}, S^{\L\H}, S^{\L\L}, 
T^{\H\H}, T^{\H\L}, T^{\L\H}, T^{\L\L}\}$ of $R$ on  $A$,
 of $S$ on $(B,C)$, and of $T$ on $(C,A)$; 
$\inst{V} = \{ \triangle_2^{\H\H\H}, \triangle_2^{\L\L\L}, 
\\
\triangle_2^{\H(\L\L)\F}, \triangle_2^{\L\F(\H\H)}, 
V_{RS}, \hat{V}_{RS}, V^{\H(\L\H)\F}, \hat{V}^{\H(\L\H)\F},
V_{ST}, V^{\F\H\L},
V_{TR}, \hat{V}_{TR}, V^{\L\F(\H\L)}, \hat{V}^{\L\F(\H\L)} \}$.

The following complexity results follow mainly from the analysis of the \ivme algorithm for the ternary triangle query in the proofs of Propositions~\ref{prop:preprocessing_step_full}, \ref{prop:space_complexity_full}, and \ref{prop:single_step_time_full}.

\subsection{Preprocessing Stage} 
\label{sec:preproc_binary}
The preprocessing stage builds the initial \ivme state $\astate = (\eps, \dbeps, \inst{V}, N)$ of database $\db$ supporting the maintenance of the binary triangle query. This step first partitions $R$ on $A$, $S$ on $(B,C)$, and $T$ 
on $(C,A)$ and then computes the materialized views in $\inst{V}$ from Figure~\ref{fig:view_definitions_binary} before processing any update.

\begin{proposition}\label{prop:preprocessing_step_binary}
Given a database $\db$ and $\eps\in[0,1]$, constructing the initial \ivme state of $\inst{D}$ supporting the maintenance of the binary triangle query takes $\bigO{|\db|^{\frac{3}{2}}}$ time.
\end{proposition}
\begin{proof}
Partitioning the input relations takes $\bigO{N}$ time.
The materialized skew-aware views 
$\triangle_2^{\H(\L\L)\F}$ and 
$\triangle_2^{\L\F(\H\H)}$
 can be computed in time $\bigO{N^{3/2}}$ using 
Leapfrog TrieJoin or Recursive-Join~\cite{NgoPRR18}.
All other materialized views can be computed using the same strategies as in the proof of Proposition~\ref{prop:preprocessing_step_full} and ignoring that 
$S$ and $T$ are double partitioned.
Overall, the initial \ivme state can be computed in time $\bigO{N^{\frac{3}{2}}}$
and the result follows from $N = \Theta(|\inst{D}|)$.  
\end{proof}

\subsection{Space Complexity}
\label{sec:space_binary}
We analyze the space complexity of the \ivme maintenance strategy for the binary triangle query.
  
\begin{proposition}
\label{prop:space_complexity_binary}
Given a database $\inst{D}$ and $\eps\in[0,1]$, an \ivme state of $\inst{D}$ supporting the maintenance of the binary triangle query takes $\bigO{|\db|^{1 +\min\{\eps,1-\eps\}}}$ space.
\end{proposition}  
\begin{proof}
Figure~\ref{fig:view_definitions_binary} gives the space complexity of the materialized views. 
The space complexities of the auxiliary views follow from the proof of Proposition~\ref{prop:space_complexity_full}.
The sizes of $V^{\H(\L\H)\F}$, $V^{\F\H\L}$, and $V^{\L\F(\H\L)}$ are 
upper bounded by the sizes of $T$, $R$, and $S$, respectively, while 
the sizes of $\hat{V}^{\H(\L\H)\F}$ and $\hat{V}^{\L\F(\H\L)}$ are upper bounded by the number of distinct $C$-values in $S^{\L\H}$ and respectively $T^{\H\L}$.

\end{proof}

\subsection{Processing a Single-Tuple Update}\label{sec:single_update_binary}
\label{sec:update_binary}
We analyze the time complexity of maintaining an \ivme state for the binary triangle query under a single-tuple update. 

\begin{proposition}\label{prop:single_step_time_binary}
Given a database $\db$, $\eps\in[0,1]$, and an \ivme state $\astate$ of $\inst{D}$ supporting the maintenance of the binary triangle query, \ivme maintains $\astate$ under a single-tuple update to any input relation in $\bigO{|\db|^{\max\{\eps,1-\eps\}}}$ time.
\end{proposition} 
\begin{proof}
Almost all the materialized views from Figure~\ref{fig:view_definitions_binary} can be maintained in time $\bigO{N^{\max\{\eps, 1-\eps\}}}$ under single-tuple updates by following the maintenance strategies described in the proof of Proposition~\ref{prop:single_step_time_full}.
The only new challenge is to maintain the refined views
 $\triangle_2^{\H(\L\L)\F}$ and 
 $\triangle_2^{\L\F(\H\H)}$.

We analyze the maintenance time for $\triangle_2^{\H(\L\L)\F}$.
For updates to $R^\H$, we need to iterate over less than $\frac{3}{2}N^\eps$ $C$-values in $S^{\L\L}$ for a fixed $B$-value from $\delta{R^\H}$
and do lookups in $T$. For updates to $T$, we need 
 to iterate over less than $\frac{3}{2}N^\eps$ $B$-values in $S^{\L\L}$ for a fixed $C$-value from $\delta{T}$ and do lookups in $R^{\H}$.
For updates to $S^{\L\L}$, we need to iterate over at most $2N^{1-\eps}$ distinct $A$-values in $R^\H$ and do lookups in $T$. 
Thus, $\triangle_2^{\H(\L\L)\F}$ can be maintained in $\bigO{N^{\max\{\eps, 1-\eps\}}}$ time.

We now consider the maintenance time for $\triangle_2^{\L\F(\H\H)}$.
For updates to $R^\L$, we need to iterate over at most $2N^{1-\eps}$ $C$-values in $T^{\H\H}$ and do lookups in $S$.
For updates to $S$, we need to iterate over at most $2N^{1-\eps}$ $A$-values in $T^{\H\H}$ and do lookups in $R^{\L}$.
For updates to $T^{\H\H}$, we need to iterate over less than 
$\frac{3}{2}N^{\eps}$  $B$-values in $R^\L$ for a fixed $A$-value from 
$\delta{T^{\H\H}}$ and do lookups in $S$.
Thus, $\triangle_2^{\L\F(\H\H)}$ can be maintained in $\bigO{N^{\max\{\eps, 1-\eps\}}}$ time.

The proposition follows from the above analysis and the invariant $N = \Theta(|\inst{D}|)$. 
\end{proof}

\begin{figure}[t]
\begin{center}
\renewcommand{\arraystretch}{1.3}
\setcounter{magicrownumbers}{0}

\begin{tabular}{l}
\toprule
\textsc{EnumerateBinary}(state $\astate$) \\
\midrule
\linenumber \LET $\astate = (\,\eps, N, 
\{R^\H, R^\L, S^{\H\H}, S^{\H\L}, S^{\L\H}, S^{\L\L}, 
T^{\H\H}, T^{\H\L}, T^{\L\H}, T^{\L\L}\},$ \\
\TAB\TAB\TAB\TAB\TAB\TAB \TAB\TAB\TAB \TAB\TAB\TAB $\{\, \triangle_2^{\H\H\H},\, \triangle_2^{\L\L\L},\, \triangle_2^{\H(\L\L)\F}, \, 
\triangle_2^{\L\F(\H\H)}, \, V^{\F\H\L}\,\} \cup \inst{V}\,)$ \\

\linenumber $\inst{I}_1 = \{\, \triangle_2^{\H\H\H}\!.\mathit{iter()},\, \triangle_2^{\L\L\L}\!.\mathit{iter()},\, \triangle_2^{\H(\L\L)\F}\!.\mathit{iter()},\, \triangle_2^{\L\F(\H\H)}\!.\mathit{iter()}, \, V^{\F\H\L}\!.\mathit{iter()} \,\}$ \\

\linenumber $\inst{I}_2 = \{\hspace{0.3mm} \triangle_2^{\H(\L\H)\F}.\mathit{iter}\left(\,\textsc{CandidateBuckets}^{\H(\L\H)\F}\hspace{0.3mm} \right), \,
\triangle_2^{\L\F(\H\L)}.\mathit{iter}\left(\hspace{0.3mm}\textsc{CandidateBuckets}^{\L\F(\H\L)}\hspace{0.3mm}\right) \,\}$ \\

\linenumber \WHILE $(\,((\deltaA, \deltaB) = \textsc{UnionNext}(\, \inst{I}_1 \cup \inst{I}_2 \,)) \neq$ \EOF\,) \\

\linenumber \TAB $\p_1 = 
\triangle_2^{\H\H\H}(\deltaA,\deltaB) + 
\triangle_2^{\L\L\L}(\deltaA,\deltaB) +
\triangle_2^{\H(\L\L)\F}(\deltaA,\deltaB) +
\triangle_2^{\L\F(\H\H)}(\deltaA,\deltaB) +
V^{\F\H\L}(\deltaA,\deltaB)$ \\

\linenumber \TAB $\p_2 = 
\sum_{t\in\{\H,\L\}^2} \sum_c R^\H(\deltaA,\deltaB) \cdot S^{\L\H}(\deltaB, c) \cdot T^t(c,\deltaA)$\\

\linenumber \TAB $\p_3 = 
\sum_{s\in\{\H,\L\}^2} \sum_c R^\L(\deltaA,\deltaB) \cdot S^{s}(\deltaB, c) \cdot T^{\H\L}(c,\deltaA)$\\

\linenumber \TAB \OUTPUT $(\deltaA,\deltaB)\mapsto (\p_1+\p_2+\p_3)$ \\
\bottomrule
\end{tabular}

\end{center}\vspace{-1em}
\caption{
Enumerating the result of the binary triangle query given an \ivme state of database $\db$.
Line~2 creates iterators over materialized skew-aware views. 
Line~3 creates hop-based iterators over the non-materialized skew-aware views, 
parameterized by the $\textsc{CandidateBuckets}^{\H(\L\H)\F}$ and $\textsc{CandidateBuckets}^{\L\F(\H\L)}$ functions.
Lines~5-7 compute the multiplicity of pair $(\deltaA,\deltaB)$ reported by the union algorithm.}

\label{fig:enum_binary}
\end{figure}

\subsection{Enumeration Delay}
\label{sec:enumeration_binary}
We construct an iterator for each skew-aware view 
of the binary triangle query and use 
the union algorithm from Section~\ref{sec:union_algorithm}
to enumerate the distinct tuples in the union  
of these views. 
For the materialized skew-aware views
$\triangle_2^{\H\H\H}$, $\triangle_2^{\L\L\L}$,
$\triangle_2^{\H(\L\L)\F}$, $\triangle_2^{\L\F(\H\H)}$, and
$\triangle_2^{\F\H\L}$ (materialized by $V^{\F\H\L}$), 
we construct iterators with constant lookup time and enumeration delay (see Section~\ref{sec:view_iterators}). 
For each of the non-materialized views $\triangle_2^{\H(\L\H)\F}$ and $\triangle_2^{\L\F(\H\L)}$, 
we first instantiate its view tree for the distinct $C$-values appearing at its root 
and then construct a hop-based iterator (see Section~\ref{sec:skip_pointers}) to enumerate the distinct $(A,B)$-values in the union of these instantiated view trees.

Given a materialized view $V$, we write $V.\mathit{iter}(\,)$ to denote the iterator for $V$.
We also call the function $\triangle_2^{\H(\L\H)\F}.\mathit{iter}(\textsc{CandidateBuckets}^{\H(\L\H)\F})$  
to get the hop-based iterator for $\triangle_2^{\H(\L\H)\F}$ parameterized by the $\textsc{CandidateBuckets}^{\H(\L\H)\F}$ function. This function intersects the $C$-values from the root $\hat{V}^{\H(\L\H)\F}$ and the $C$-values paired with a given $(A,B)$-value in the view $V_{RS}$. 
Similarly, the hop-based iterator for $\triangle_2^{\L\F(\H\L)}$ uses the $\textsc{CandidateBuckets}^{\L\F(\H\L)}$ function that intersects the $C$-values from the root $\hat{V}^{\L\F(\H\L)}$ and the $C$-values paired with a given $(A,B)$-value in the view $V_{TR}$.
Both functions return a set of indices that identify the view trees instantiated for the computed $C$-values.

The procedure \textsc{EnumerateBinary} from Figure~\ref{fig:enum_binary} enumerates the result of the binary triangle query given an \ivme state $\astate$. 
The procedure first creates the iterators over 
the (possibly non-disjoint) results of the skew-aware views.
The union algorithm from Figure~\ref{fig:enum_union} takes these iterators as input and reports distinct $(A,B)$-values as output. 
For each reported $(a,b)$, \textsc{EnumerateBinary} computes the 
multiplicity of $(a,b)$ by summing up the multiplicities in each skew-aware view.

\begin{proposition}
\label{prop:delay_binary}
Given a database $\db$, $\eps\in[0,1]$, an \ivme state $\astate$ of \hspace{0.5mm}$\db$ supporting the maintenance of the binary triangle query, \ivme enumerates the result of the query with 
$\bigO{|\db|^{\min\{\eps, 1-\eps\}}}$ delay and 
$\bigO{|\db|^{1+\min\{\eps, 1-\eps\}}}$ additional space.
\end{proposition}
\begin{proof}
We analyze the procedure 
\textsc{EnumerateBinary} in 
Figure~\ref{fig:enum_binary}.
Creating the iterators over materialized views takes constant time (Line~2); 
the same holds for the hop-based iterators in $\inst{I}_2$, per Lemma~\ref{lem:enumerate_hop_iterator} (Line~3). 
The iterators in $\inst{I}_1$ allow constant-time lookups and constant-delay enumeration of $(A,B)$-values. 
The hop-based iterator for $\triangle_2^{\H(\L\H)\F}$ is over at most $2N^{1-\eps}$ view trees instantiated for the distinct $C$-values appearing at the root $\hat{V}^{\H(\L\H)\F}$.
Each view tree supports constant-time lookups and constant-delay enumeration of $(A,B)$-values.
$\textsc{CandidateBuckets}^{\H(\L\H)\F}$ intersects at most $\min\{\frac{3}{2}N^\eps,2N^{1-\eps}\}$ $C$-values from $V_{RS}$ for a fixed $(A,B)$-value and at most $2N^{1-\eps}$ $C$-values from $\hat{V}^{\H(\L\H)\F}$;
thus, the returned set of indices is of size at most $\min\{\frac{3}{2}N^\eps, 2N^{1-\eps}\}$.
This function runs in $\bigO{N^{\min\{\eps, 1-\eps\}}}$ time.
Per Lemma~\ref{lem:enumerate_pointer_new}, the enumeration delay of the hop-based iterator for $\triangle_2^{\H(\L\H)\F}$ is $\bigO{N^{\min\{\eps, 1-\eps\}}}$. 
A similar analysis for $\triangle_2^{\L\F(\H\L)}$ gives the same enumeration delay. 

The iterators over materialized views need constant space during enumeration.
The hop-based iterators over $\triangle_2^{\H(\L\H)\F}$ and $\triangle_2^{\L\F(\H\L)}$ need space linear in the total number of their $(A,B)$-values, per Lemma~\ref{lem:enumerate_hop_iterator}.
This number is upper bounded by the size of $V_{RS}$ for the former and by the size of $V_{TR}$ for the latter.
By Proposition~\ref{prop:space_complexity_binary}, both of these views take $\bigO{N^{1 + \min\{\eps, 1-\eps\}}}$ space.

Computing the total multiplicity $m$ of a pair $(\deltaA,\deltaB)$ 
requires computing the multiplicity of $(\deltaA,\deltaB)$ in the result of each skew-aware view.
For the materialized views with schema $(A,B)$, this operation takes constant time (Line~5).
For the non-materialized views $\triangle_2^{\H(\L\H)\F}$ and $\triangle_2^{\L\F(\H\L)}$,
computing the multiplicities of $(\deltaA,\deltaB)$ requires iterating over the matching $C$-values in $S^{\L\H}$ and respectively $T^{\H\L}$ (Lines~6-7). 
In both cases, the number of distinct $C$-values for a fixed $(\deltaA,\deltaB)$ is at most $\min\{\frac{3}{2}N^\eps, 2N^{1-\eps}\}$.
Thus, the multiplicity of the pair $(\deltaA,\deltaB)$ can be computed in $\bigO{N^{\min\{\eps, 1-\eps\}}}$ time. 

Overall, \textsc{EnumerateBinary} enumerates the result of $\triangle_2$ from $\astate$ with $\bigO{N^{\min\{\eps, 1-\eps\}}}$ delay and 
\\
$\bigO{N^{1 + \min\{\eps, 1-\eps\}}}$ additional space. 
The proposition follows from the invariant $N = \Theta(|\db|)$.

\end{proof}

\subsection{Summing Up}
The additional space used during the enumeration 
of the result of the binary triangle query is linear in the 
size of the maintained views. Hence, 
our main result in Theorem~\ref{theo:main_result_triangle} for the binary triangle query follows from Propositions~\ref{prop:preprocessing_step_binary}, \ref{prop:space_complexity_binary}, \ref{prop:single_step_time_binary}, and \ref{prop:delay_binary} shown in the previous subsections, complemented by 
Proposition~\ref{prop:amortized_update_time}, which shows that the amortized
rebalancing time is $\bigO{|\inst{D}|^{\max\{\eps, 1-\eps\}}}$.

%% file: unary.tex

\section{Maintaining the Unary Triangle Query}
\label{sec:unary}
We now focus on the maintenance and enumeration of the unary triangle query 
$$\triangle_1(a) = \sum_{b,c}  R(a,b)\ztimes S(b,c) \ztimes T(c,a)$$
under a single-tuple update. As with the binary triangle query, the results of the skew-aware views in the unary case  are not necessarily disjoint. 
To report only the distinct $A$-values in the union of skew-aware views, 
we again rely on the union algorithm, presented in Section~\ref{sec:union_algorithm}.

We discuss the enumeration of distinct $A$-values 
in the result of skew-aware views that are not materialized but represented as view trees.
As a starting point for our discussion, we consider the view trees created for the ternary triangle query, see Figure~\ref{fig:view_definitions_full}. 
The view trees for $\triangle_3^{\H\L\F}$ and $\triangle_3^{\F\H\L}$ contain $A$-values at the root, thus they can support the enumeration of $A$-values in constant time. 
The view tree $T$ for $\triangle_3^{\L\F\H}$, however, contains $(B,C)$-values at its root, meaning that we need to find the distinct $A$-values that occur under $(B,C)$-values. 
The number of distinct $(B,C)$-values paired with any given $A$-value can be linear, 
meaning that a hop-based iterator from Section~\ref{sec:skip_pointers}
would enumerate distinct $A$-values with at least linear delay.

To improve the enumeration delay for the skew-aware view 
$\triangle_3^{\L\F\H}$, we refine our partitioning strategy to get a tighter bound on the number of $(B,C)$-values paired with any given  
$A$-value.
We double partition relation $R$ on $(A,B)$ and relation $T$ on $(C,A)$ 
while keeping $S$ partitioned on $B$.
This refinement further divides $\triangle_3^{\L\F\H}$
into three 
skew-aware views. 
One skew-aware view involves $R^{\L\H}$ and $T^{\H\L}$ and ensures that the number of distinct 
$(B,C)$-values 
paired with any $A$-value 
is bounded by $\bigO{N^{2\min\{\eps,1-\eps\}}}$ since $A$ is light 
in both relation parts and each of the variables $B$ and $C$ is heavy in at least one of the relation parts. 
The other two skew-aware views either involve $R^{\L\L}$ or involve 
$R^{\L\H}$ and $T^{\H\H}$, which enables their materialization and enumeration with constant delay. 
Overall, our maintenance strategy for the unary triangle query with double partitioning for $R$
and $T$ achieves $\bigO{N^{2\min\{\eps,1-\eps\}}}$ enumeration delay, which is sublinear for 
$\eps \neq \frac{1}{2}$.

\begin{figure}[t!]
  \begin{minipage}[b]{\textwidth}
  \begin{center}
    \renewcommand{\arraystretch}{1.3}  
    \begin{tabular}{@{\hskip 0.0in}l@{\hskip 0.15in}l@{\hskip 0.0in}}
      \toprule
      Materialized View Definition & Space Complexity \\    
      \midrule 
      $\triangle_1^{\H\H\H}(a) = 
      \sum_{r,t\in\{\H,\L\}}\sum_{b,c} R^{\H{r}}(a,b) \cdot S^\H(b,c) \cdot T^{\H{t}}(c,a)$ & 
      $\bigO{|\inst{D}|^{1-\eps}}$ \\[2pt]

      $\triangle_1^{\L\L\L}(a) = 
      \sum_{r,t\in\{\H,\L\}}\sum_{b,c} R^{\L{r}}(a,b) \cdot S^\L(b,c) \cdot T^{\L{t}}(c,a)$  & 
      $\bigO{|\inst{D}|}$ \\[2pt]

      $\triangle_1^{(\L\L)\F\H}(a) = 
      \sum_{s,t\in\{\H,\L\}}\sum_{b,c} R^{\L\L}(a,b) \cdot S^{s}(b,c) \cdot T^{\H{t}}(c,a)$  & 
      $\bigO{|\inst{D}|}$ \\[3pt]
      
      $\triangle_1^{(\L\H)\F(\H\H)}(a) = 
      \sum_{s\in\{\H,\L\}}\sum_{b,c} R^{\L\H}(a,b) \cdot S^{s}(b,c) \cdot T^{\H\H}(c,a)$ & 
      $\bigO{|\inst{D}|^{1-\eps}}$ \\[3pt]
      
      View tree for $\triangle_1^{\H\L\F}(a) = 
      \sum_{r\in\{\H,\L\}}\sum_{t\in\{\H,\L\}^2} \sum_{b,c} R^{\H{r}}(a,b) \cdot S^\L(b,c) \cdot T^t(c,a)$\\[2pt]

      $\TAB\TAB V_{RS}(a,c) = \sum_{r\in\{\H,\L\}}\sum_{b} R^{\H{r}}(a,b) \cdot S^\L(b,c)$ & 
      $\bigO{|\inst{D}|^{1+\min{\{\,\eps, 1-\eps \,\}}}}$ \\[2pt]

      $\TAB\TAB V^{\H\L\F}(a) =
      \sum_{t\in\{\H,\L\}^2} \sum_{c}V_{RS}(a,c) \cdot T^t(c,a)$ & 
      $\bigO{|\inst{D}|^{1-\eps}}$\\[3pt]

      View tree for $\triangle_1^{\F\H\L}(a) =  
      \sum_{r\in\{\H,\L\}^2}\sum_{t\in\{\H,\L\}} \sum_{b,c} R^r(a,b) \cdot S^{\H}(b,c) \cdot T^{\L{t}}(c,a)$ & \\[2pt]

      $\TAB\TAB V_{ST}(b,a) = \sum_{t\in\{\H,\L\}} \sum_{c} S^\H(b,c) \cdot T^{\L{t}}(c,a)$ & 
      $\bigO{|\inst{D}|^{1+\min{\{\,\eps, 1-\eps \,\}}}}$ \\[2pt]

      $\TAB\TAB V^{\F\H\L}(a) =
      \sum_{r\in\{\H,\L\}^2}\sum_{b} R^r(a,b) \cdot V_{ST}(b,a)$ & 
      $\bigO{|\inst{D}|}$\\[3pt]

 View tree for $\triangle_1^{(\L\H)\F(\H\L)}(a) =  
       \sum_{s\in\{\H,\L\}} \sum_{b,c} R^{\L\H}(a,b) \cdot S^s(b,c) \cdot T^{\H\L}(c,a)$ & \\[2pt]

$\TAB\TAB V_{TR}(c,a,b) = T^{\H\L}(c,a) \cdot R^{\L\H}(a,b)$ &
 $\bigO{|\inst{D}|^{1+\min{\{\,\eps, 1-\eps \,\}}}}$ \\[2pt]

 $\TAB\TAB \hat{V}_{TR}(c,b) = \sum_a V_{TR}(c,a,b)$  & 
 $\bigO{|\inst{D}|^{1+\min{\{\,\eps, 1-2\eps \,\}}}}$ \\[2pt]

 $\TAB\TAB V^{(\L\H)\F(\H\L)}(b,c) =  
      \sum_{s\in\{\H,\L\}} S^s(b,c) \cdot \hat{V}_{TR}(c,b)$  & 
 $\bigO{|\inst{D}|^{\min\{1,2-2\eps\}}}$  \\[3pt]
      \bottomrule
    \end{tabular}
  \end{center}
  \end{minipage} 

    \vspace{0.1cm}
  \begin{minipage}{\textwidth}
  \begin{center}
    \begin{tikzpicture}
    \begin{scope}[xshift=-1cm]
      \node at (0,0.1) (G) {View tree for $\triangle_1^{\H\L\F}$};
      \node at (0, -0.5) (A) {$V^{\H\L\F}(a)$};
      \node at (1, -1.65) (B) {$\sum\limits_{t\in\{\H,\L\}^2}\!\!\!\!\!\!T^t(c,a)$} edge[-] (A);
      \node at (-1, -1.5) (D) {$V_{RS}(a,c)$} edge[-] (A);
      \node at (-2, -2.65) (E) {$\sum\limits_{r\in\{\H,\L\}}\!\!\!\!\!R^{\H{r}}(a,b)$} edge[-] (D);
      \node at (0, -2.5) (F) {$S^\L(b,c)$} edge[-] (D);
    \end{scope}

    \begin{scope}[xshift=3.5cm]
      \node at (0,0.1) (G) {View tree for $\triangle_1^{\F\H\L}$};
      \node at (0, -0.5) (A) {$V^{\F\H\L}(a)$};
      \node at (1, -1.65) (B) {$\sum\limits_{r\in\{\H,\L\}^2}\!\!\!\!\!R^r(a,b)$} edge[-] (A);
      \node at (-1, -1.5) (D) {$V_{ST}(b,a)$} edge[-] (A);
      \node at (-2, -2.5) (E) {$S^\H(b,c)$} edge[-] (D);
      \node at (0, -2.6) (F) {$\sum\limits_{t\in\{\H,\L\}}\!\!\!\!\!T^{\L{t}}(c,a)$} edge[-] (D);      
    \end{scope}

    \begin{scope}[xshift=8cm,yshift=0.05cm]
      \node at (0,0.1) (G) {View tree for $\triangle_1^{(\L\H)\F(\H\L)}$};
      \node at (0, -0.5) (A) {$V^{(\L\H)\F(\H\L)}(b,c)$};
      \node at (-1, -1.5) (C) {$\hat{V}_{TR}(c,b)$} edge[-] (A);
      \node at (1, -1.65) (B) {$\sum\limits_{s\in\{\H,\L\}}\!\!\!\!\!S^s(b,c)$} edge[-] (A);
      \node at (-1, -2.5) (D) {$V_{TR}(c,a,b)$} edge[-] (C);
      \node at (-2, -3.5) (E) {$T^{\H\L}(c,a)$} edge[-] (D);
      \node at (0, -3.5) (F) {$R^{\L\H}(a,b)$} edge[-] (D);     
    \end{scope}
    \end{tikzpicture}
      \end{center}
\end{minipage}

  \caption{
  (top) 
  The materialized views
  $\inst{V} = \{ 
  \triangle_1^{\H\H\H}, \triangle_1^{\L\L\L}, \triangle_1^{(\L\L)\F\H},
  \triangle_1^{(\L\H)\F(\H\H)}, 
  V_{RS}, V^{\H\L\F},
  V_{ST}, V^{\F\H\L},
  V_{TR}, \hat{V}_{TR},$ $V^{(\L\H)\F(\H\L)} \}$ 
  supporting the maintenance of the unary triangle query.
  The set $\inst{V}$ is part of an \ivme state of database $\inst{D}$. 
  (bottom) The view trees supporting the maintenance of 
  $\triangle_1^{\H\L\F}$, $\triangle_1^{\F\H\L}$, and $\triangle_1^{(\L\H)\F(\H\L)}$.}
  
  \label{fig:view_definitions_unary}
  \vspace{-6pt}
\end{figure}

Figure~\ref{fig:view_definitions_unary} shows the definition and space complexity of the views supporting the maintenance of the unary triangle query. 
The \ivme state supporting the maintenance of the unary triangle query has the partitions
$\dbeps = \{R^{\H\H}, R^{\H\L}, R^{\L\H}, R^{\L\L}, S^\H, S^\L, 
T^{\H\H}, T^{\H\L}, T^{\L\H}, T^{\L\L}\}$ of $R$ on $(A,B)$,
of $S$ on $B$, and of $T$ on $(C,A)$; 
$\inst{V} = \{ 
  \triangle_1^{\H\H\H}, \triangle_1^{\L\L\L}, \triangle_1^{(\L\L)\F\H},
  \triangle_1^{(\L\H)\F(\H\H)}, 
  V_{RS}, V^{\H\L\F},
  V_{ST}, V^{\F\H\L},
  V_{TR}, \hat{V}_{TR}, V^{(\L\H)\F(\H\L)} \}$. 

\subsection{Preprocessing Stage} 
The preprocessing stage builds the initial \ivme state $\astate = (\eps, \dbeps, \inst{V}, N)$ of database $\db$ supporting the maintenance of the unary triangle query. 

\begin{proposition}\label{prop:preprocessing_step_unary}
Given a database $\db$ and $\eps\in[0,1]$, constructing the initial \ivme state of $\inst{D}$ supporting the maintenance of the unary triangle query takes $\bigO{|\db|^{\frac{3}{2}}}$ time.
\end{proposition}
\begin{proof}
The proof is similar to the proof of Proposition~\ref{prop:preprocessing_step_binary}.
\end{proof}

\subsection{Space Complexity}\label{sec:space_unary}
We analyze the space complexity of the \ivme maintenance strategy for the unary triangle query.
  
\begin{proposition}
\label{prop:space_complexity_unary}
Given a database $\inst{D}$ and $\eps\in[0,1]$, an \ivme state of $\inst{D}$ supporting the maintenance of the unary triangle query takes $\bigO{|\db|^{1 +\min\{\eps,1-\eps\}}}$ space.
\end{proposition}  
\begin{proof}
Figure~\ref{fig:view_definitions_unary} gives the definition and space complexity of the materialized views. The complexity results follow mainly from the proof of Proposition~\ref{prop:space_complexity_full}. The remaining views take either linear space because of their unary schema or sublinear space because of the heavy part condition on $A$ in one of the relation parts. 
Two notable cases are the views $\hat{V}_{TR}$ and $V^{(\L\H)\F(\H\L)}$.
The size of $\hat{V}_{TR}$ is upper bounded by the size of $V_{TR}$, which is $\bigO{N^{1+\min\{\eps,1-\eps\}}}$ as discussed in the proof of Proposition~\ref{prop:space_complexity_full}, but also by at most $4N^{2-2\eps}$ $(B,C)$-values created by pairing the distinct heavy $B$-values from $R^{\L\H}$ and the distinct heavy $C$-values from $T^{\H\L}$. Thus, the view $\hat{V}_{TR}$ takes $\bigO{N^{1 + \min\{\eps,1-2\eps\}}}$ space.
The view view $V^{(\L\H)\F(\H\L)}$ is further upper bounded by the size of $S$, which gives its $\bigO{N^{\min\{1,2-2\eps\}}}$ space. The proposition follows from the invariant $N = \bigO{|\inst{D}|}$.
\end{proof}

\subsection{Processing a Single-Tuple Update}\label{sec:single_update_unary}

We analyze the time complexity of maintaining an \ivme state for the unary triangle query under a single-tuple update. 

\begin{proposition}\label{prop:single_step_time_unary}
Given a database $\db$, $\eps\in[0,1]$, and an \ivme state $\astate$ of $\inst{D}$ supporting the maintenance of the unary triangle query, \ivme maintains $\astate$ under a single-tuple update to any input relation in $\bigO{|\db|^{\max\{\eps,1-\eps\}}}$ time.
\end{proposition} 
\begin{proof}
Almost all materialized views in 
 Figure~\ref{fig:view_definitions_unary} can be maintained 
 following the same strategies as in the proof of  
Proposition~\ref{prop:single_step_time_full} and by 
ignoring the double partitioning of $R$ and $T$.
The only notable cases are the refined skew-aware views 
$\triangle_1^{(\L\L)\F\H}$ and $\triangle_1^{(\L\H)\F(\H\H)}$, considered next.

We analyze the time to maintain $\triangle_1^{(\L\L)\F\H}$. 
For updates to $R^{\L\L}$,
we need to iterate over at most $2N^{1-\eps}$ $C$-values in $T^{\H}$ and do lookups in $S$. 
For updates to $S$, we need to iteration over less than $\frac{3}{2}N^{\eps}$ $A$-values in $R^{\L\L}$ 
for a fixed $B$-value from $\delta{S}$ and do lookups in $T^{\H}$.
For updates to $T^{H}$, we need to iterate 
over less than $\frac{3}{2}N^{\eps}$ $B$-values for a fixed 
$A$-value from $\delta{T^{H}}$ and do lookups in $S$. 
Thus, maintaining $\triangle_1^{(\L\L)\F\H}$ takes $\bigO{N^{\max\{\eps, 1-\eps\}}}$ time.

The maintenance strategies 
for $\triangle_1^{(\L\H)\F(\H\H)}$ 
differ from 
the strategies above only in case of updates to $S$. 
For an update $S$, we iterate over at most $2N^{1-\eps}$ $A$-values 
in $T^{\H\H}$ and do lookups 
in $R^{\L\H}$.
This implies that the maintenance time is $\bigO{N^{1-\eps}}$.

Hence, the overall maintenance time is
$\bigO{N^{\max\{\eps, 1-\eps\}}}$.  
The result follows from $N = \bigO{|\inst{D}|}$.
\end{proof}

\begin{figure}[t]
\begin{center}
\renewcommand{\arraystretch}{1.3}
\setcounter{magicrownumbers}{0}

\begin{tabular}{l}
\toprule
\textsc{EnumerateUnary}(state $\astate$) \\
\midrule
\linenumber \LET $\astate = (\,\eps, N, 
\{R^{\H\H}, R^{\H\L}, R^{\L\H}, R^{\L\L}, S^\H, S^\L, 
T^{\H\H}, T^{\H\L}, T^{\L\H}, T^{\L\L}\},$ \\
\TAB\TAB\TAB\TAB\TAB\TAB\TAB\TAB\TAB\TAB\TAB $\{\, \triangle_1^{\H\H\H},\, \triangle_1^{\L\L\L},\, \triangle_1^{(\L\L)\F\H},\, 
\triangle_1^{(\L\H)\F(\H\H)}, 
V^{\H\L\F},\, V^{\F\H\L} \,\} \cup \inst{V}\,)$ \\

\linenumber $\inst{I}_1 = \{\, \triangle_1^{\H\H\H}\!.\mathit{iter()},\, \triangle_1^{\L\L\L}\!.\mathit{iter()},\, \triangle_1^{(\L\L)\F\H}\!.\mathit{iter()},\, \triangle_1^{(\L\H)\F(\H\H)}\!.\mathit{iter()},\, V^{\H\L\F}\!.\mathit{iter()},\, V^{\F\H\L}\!.\mathit{iter()},  \,\}$ \\

\linenumber $\inst{I}_2 = \{\, \triangle_1^{(\L\H)\F(\H\L)}.iter\left(\hspace{0.3mm}\textsc{CandidateBuckets}^{(\L\H)\F(\H\L)}\hspace{0.3mm}\right)\,\}$ \\

\linenumber \WHILE $(\,(\deltaA = \textsc{UnionNext}(\, \inst{I}_1 \cup 
\inst{I}_2  \,)) \neq$ \EOF\,) \\

\linenumber \TAB $\p_1 = 
\triangle_1^{\H\H\H}(\deltaA) + 
\triangle_1^{\L\L\L}(\deltaA) +
\triangle_1^{(\L\L)\F\H}(\deltaA) +
\triangle_1^{(\L\H)\F(\H\H)}(\deltaA) +
V^{\H\L\F}(\deltaA) +
V^{\F\H\L}(\deltaA)$ \\

\linenumber \TAB $\p_2 = 
\sum_{s\in\{\H,\L\}} \sum_{b,c} R^{\L\H}(\deltaA,b) \cdot S^s(b,c) \cdot T^{\H\L}(c,\deltaA)$\\
\linenumber \TAB \OUTPUT $\deltaA \mapsto (\p_1+\p_2)$ \\
\bottomrule
\end{tabular}
\end{center}
\caption{
Enumerating the result of the unary triangle query given an \ivme state of database $\db$.
Line 2 creates six iterators over the results of materialized views with schema $A$.
Line 3 creates a hop-based iterator over the non-materialized skew-aware view 
$\triangle_1^{(\L\H)\F(\H\L)}$, 
parameterized by the $\textsc{CandidateBuckets}^{(\L\H)\F(\H\L)}$ function. 
Lines 5 and 6 compute the multiplicity of $\deltaA$ reported by the union algorithm.
}

\label{fig:enum_unary}
\end{figure}

\subsection{Enumeration Delay}
\label{sec:enumeration_unary}
The enumeration procedure for the unary triangle query is similar to that of the binary triangle query. 
The skew-aware views from Figure~\ref{fig:view_definitions_unary} are all materialized 
except $\triangle_1^{(\L\H)\F(\H\L)}$.
For each materialized view, we construct an iterator with constant lookup time and enumeration delay. 
For the non-materialized view $\triangle_1^{(\L\H)\F(\H\L)}$,
we first instantiate its view tree for the distinct $(B,C)$-values appearing at the root $V^{(\L\H)\F(\H\L)}$
and then construct a hop-based iterator for enumerating the distinct $A$-values in the union of these view trees. The hop-based iterator is parameterized by the $\textsc{CandidateBuckets}^{(\L\H)\F(\H\L)}$ function that restricts the set of instantiated view trees to be explored during enumeration for a fixed $A$-value. 
This function first computes the $(B,C)$-values that exist in both the materialized view $V_{TR}$ for the given $A$-value and the root $V^{(\L\H)\F(\H\L)}$, and then returns a set of indices that identify the view trees instantiated for those $(B,C)$-values.

The procedure \textsc{EnumerateUnary} from Figure~\ref{fig:enum_unary} enumerates the result of the unary triangle query given an \ivme state $\astate$. 
The procedure first creates the iterators for all skew-aware views (Lines~2-3). 
The union algorithm (see Section~\ref{sec:union_algorithm}) takes these iterators as input and reports distinct $A$-values as output. 
For each reported $A$-value $\deltaA$, \textsc{EnumerateUnary} sums up the multiplicity of $\deltaA$ in each of the skew-aware views, 
which involves lookups in the materialized views with schema $A$ (Line~5) 
and an aggregation of $(B,C)$-values over the relation parts from 
$\triangle_1^{(\L\H)\F(\H\L)}$ (Line~6).

\begin{proposition}
\label{prop:delay_unary}
Given a database $\db$, $\eps\in[0,1]$, an \ivme state $\astate$ of \ $\db$ supporting the maintenance of the unary triangle query, \ivme enumerates the query result from $\astate$ with $\bigO{|\db|^{2\min\{\eps,1-\eps\}}}$ delay
and $\bigO{|\db|^{1+ \min\{\eps,1-\eps\}}}$ additional space.
\end{proposition}

\begin{proof}
Creating the iterators over materialized and the hop-based iterator over $\triangle_1^{(\L\H)\F(\H\L)}$ takes constant time (Line~2-3), 
The iterators over the materialized views with schema $A$ allow constant-time lookups and constant-delay enumeration of $A$-values. 
The hop-based iterator reports the distinct $A$-values from the union of at most $\min\{N, 4N^{2(1-\eps)}\}$ view trees instantiated for the distinct $(B,C)$-values in the root $V^{(\L\H)\F(\H\L)}$. Each such a view tree allows constant-time lookups and constant-delay enumeration of $A$-values. 

The $\textsc{CandidateBuckets}^{(\L\H)\F(\H\L)}$ function, which parameterizes the hop-based iterator, first intersects the $(B,C)$-values from $V_{TR}$ for a fixed $A$-value and from the root $V^{(\L\H)\F(\H\L)}$.
The number of $(B,C)$-values in $V_{TR}$ is at most $4N^{2-2\eps}$ due to the heavy part conditions on $B$ in $R^{\L\H}$ and on $C$ in $T^{\H\L}$, and less than $\frac{9}{4}N^{2\eps}$ for a fixed $A$-value due to the light part conditions on $A$ in $R^{\L\H}$ and on $A$ in $T^{\H\L}$.
The number of $(B,C)$-values in $V^{(\L\H)\F(\H\L)}$ is further upper bounded by the size of $S$.
Thus, computing the intersection and returning a set of indices that identify the matching view trees take $\bigO{N^{2\min\{\eps, 1-\eps\}}}$ time.
The returned set of indices is of size at most $\min\{N, 4N^{2-2\eps}, \frac{9}{4}N^{2\eps}\}$.
Per Lemma~\ref{lem:enumerate_pointer_new}, the enumeration delay for the view $\triangle_1^{(\L\H)\F(\H\L)}$ is $\bigO{N^{2\min\{\eps,1-\eps\}}}$.

The iterators over materialized views require constant space during enumeration.
The hop-based iterator over $\triangle_1^{(\L\H)\F(\H\L)}$ requires space linear in the total number of its $A$-value, per Lemma~\ref{lem:enumerate_hop_iterator}.
This number is upper bounded by the size of $V_{TR}$, which takes $\bigO{N^{1 + \min\{\eps, 1-\eps\}}}$ space by Proposition~\ref{prop:space_complexity_unary}.

Computing the total multiplicity of each reported $A$-value $\deltaA$ requires constant-time lookups in the materialized views with schema $A$ (Line~5) and iteration over the distinct $(B,C)$-values appearing in the join of $R^{\L\H}$, $S$, and $T^{\H\L}$ (Line~6);  
since $A$ is light in $R^{\L\H}$ and $T^{\H\L}$, and each of the variables 
$B$ and $C$ is heavy in one of these relation parts,  the number of such 
$(B,C)$-values 
is $\bigO{N^{2\min{\{\eps,1-\eps\}}}}$. Thus, the multiplicity of the output value $\deltaA$ can be computed in $\bigO{N^{2\min{\{\eps,1-\eps\}}}}$ time. 

Overall, \textsc{EnumerateUnary} enumerates the result of $\triangle_1$ from $\astate$ with 
$\bigO{N^{2\min{\{\eps,1-\eps\}}}}$ delay
and $\bigO{N^{1+\min{\{\eps,1-\eps\}}}}$ additional space.
The proposition follows from the invariant $|\db|=\Theta(N)$.

\end{proof}

\subsection{Summing Up}
The additional space used by the enumeration algorithm for the 
unary triangle query is linearly bounded by the overall space complexity 
of maintained views. We conclude that 
our main result in Theorem~\ref{theo:main_result_triangle} for the unary triangle query
follows from Propositions~\ref{prop:preprocessing_step_unary}, \ref{prop:space_complexity_unary}, \ref{prop:single_step_time_unary}, and \ref{prop:delay_unary} shown in the previous subsections, complemented by 
Proposition~\ref{prop:amortized_update_time}, which shows that the amortized
rebalancing time is $\bigO{|\inst{D}|^{\max\{\eps, 1-\eps\}}}$.

%% file: rebalancing.tex
\section{Rebalancing Relation Partitions}
\label{sec:rebalancing}
The partition of a relation may change after updates.
For instance,  an insert 
$\delta R^{\L} = \{(\deltaA,\deltaB) \mapsto 1\}$  may violate the size invariant
 $\floor{\frac{1}{4}N} \leq |\db| < N$ in an \ivme state or may violate the light part condition $|\sigma_{A=\deltaA}R^{\L}| < \frac{3}{2}N^{\eps}$ 
 on data value $\deltaA$
 and require moving all tuples with $A$-value $\deltaA$ from $R^{\L}$ to $R^{\H}$. 
As the database evolves under updates, \ivme performs \emph{major} and \emph{minor} rebalancing steps to ensure that the size invariant and the 
heavy and light part conditions always hold. This rebalancing also ensures that the upper bounds on the number of data values, such as the number of $B$-values paired with $\deltaA$ in 
$R^{\L}$ and the number of distinct $A$-values in $R^{\H}$, are valid. The rebalancing cost is amortized over multiple updates.

The rebalancing procedures introduced in this section
operate on \ivme states supporting any triangle query discussed 
in the previous sections.  
 The maintenance procedure \textsc{ApplyUpdate} used 
 by major and minor rebalancing is  
 polymorphic in the sense that its definition depends 
 on the maintained triangle query and used partitioning scheme (single or double partitioning).
Sections~\ref{sec:single-update-nullary} and \ref{sec:single_update_full} show the procedures \textsc{ApplyUpdate} for the nullary triangle query under single partitioning and respectively the ternary triangle query.
Sections~\ref{sec:nullary_triangle_double_partitioning}, \ref{sec:single_update_binary}, and \ref{sec:single_update_unary} describe how to adapt these procedures for the nullary triangle query under double partitioning, the binary triangle query, and the unary triangle query, respectively.

\paragraph*{Major Rebalancing}
If an update causes the database size to fall below $\lfloor \frac{1}{4} N \rfloor$ or reach $N$, \ivme halves or, respectively, doubles the threshold base $N$, and calls the 
procedure \textsc{MajorRebalance} shown in Figure~\ref{fig:major_minor}. 
The procedure strictly repartitions the database relations  
with the new threshold $N^{\eps}$ (Line 2) and recomputes the materialized views using the new relation parts (Line 3).

\begin{figure}[t]
\begin{center}
\begin{tikzpicture}
\node at(-3.5,0)[anchor=north] {
\renewcommand{\arraystretch}{1}
\setcounter{magicrownumbers}{0}
\begin{tabular}{l@{\hskip 0.08in}l@{}}
\toprule
\multicolumn{2}{l}{\textsc{MajorRebalance}(state $\astate$)} \\
\midrule 
\rownumber & \LET $\astate = (\eps, N, \dbeps,\inst{V})$ \\
\rownumber & $\dbeps = \textsc{StrictPartition}(\dbeps,N^{\eps})$ \\
\rownumber & $\inst{V} = \textsc{Recompute}(\inst{V},\dbeps)$\\
\rownumber & \RETURN $\astate$ \\
\bottomrule
\end{tabular}
};

\node at(3.9,0)[anchor=north] {

\renewcommand{\arraystretch}{1}
\setcounter{magicrownumbers}{0}
\begin{tabular}{l@{}}
\toprule
{\textsc{MoveTuples}}($\text{variable }X, \text{value }x,
 K_{\mathit{src}}\! \shortrightarrow\! K_{\mathit{dst}}, \text{state }\astate$)\\
\midrule
\FOREACH $\inst{x} \in \sigma_{X=x} K_{\mathit{src}}$ \DO \\
 \TAB
$\astate$ = {\textsc{ApplyUpdate}}($\delta K_{\mathit{dst}} = \{\, \inst{x} \mapsto K_{\mathit{src}}(\inst{x}) \,\},\astate$)\\
\TAB $\astate$ = {\textsc{ApplyUpdate}}($\delta K_{\mathit{src}} = \{\, \inst{x} 
\mapsto -K_{\mathit{src}}(\inst{x}) \,\},\astate$)\\
  \RETURN $\astate$\\
\bottomrule
\end{tabular}
};

\node at(1,-11)[anchor=south] {
\renewcommand{\arraystretch}{1.1}
\setcounter{magicrownumbers}{0}
\begin{tabular}{@{\hskip 0.12in}l}
\toprule
\hspace{-0.04in}\textsc{MinorRebalance}($\hspace{0.1mm}\text{relation } K, \text{variable } X, \text{value }x, \text{variable }Y, \text{value }y, \text{state }\astate$) \\
\midrule
\linenumber \IF ($K$ is single partitioned) \\
\linenumber \TAB \IF (
$x \in \pi_X K^{\H}$ \AND $|\sigma_{X = x} K^{\H}| < \frac{1}{2} N^{\eps}$) \\
\linenumber \TAB\TAB $\astate = \textsc{MoveTuples}(X,x,
K^{\H}\! \shortrightarrow\! K^{\L},\astate)$\\

\linenumber \TAB \ELSE\IF ( 
$x \in \pi_X K^{\L}$ \AND $|\sigma_{X = x} K^{\L}| \geq \frac{3}{2} N^{\eps}$) \\
\linenumber \TAB\TAB $\astate = \textsc{MoveTuples}(X,x,
K^{\L}\! \shortrightarrow\! K^{\H},\astate)$\\

\linenumber \ELSE\IF ($K$ is double partitioned) \\

\linenumber \TAB \IF (
$x \in (\pi_X K^{\H\H} \cup \pi_X K^{\H\L})$ \AND $|\sigma_{X = x} K | < \frac{1}{2} N^{\eps}$) \\
\linenumber \TAB\TAB $\astate = \textsc{MoveTuples}(X,x,
K^{\H\H}\! \shortrightarrow\! K^{\L\H},\astate)$; \TAB
$\astate = \textsc{MoveTuples}(X,x,
K^{\H\L}\! \shortrightarrow\! K^{\L\L},\astate)$\\

\linenumber \TAB \ELSE\IF (
$x \in (\pi_X K^{\L\H} \cup \pi_X K^{\L\L})$ \AND 
$|\sigma_{X = x} K | \geq \frac{3}{2} N^{\eps}$) \\
\linenumber \TAB\TAB $\astate = \textsc{MoveTuples}(X,x,
K^{\L\H}\! \shortrightarrow\! K^{\H\H},\astate)$; 
\TAB $\astate = \textsc{MoveTuples}(X,x,
K^{\L\L}\! \shortrightarrow\! K^{\H\L},\astate)$\\

\linenumber \TAB \IF (
$y \in (\pi_Y K^{\H\H} \cup \pi_Y K^{\L\H})$ \AND 
$|\sigma_{Y = y} K | < \frac{1}{2} N^{\eps}$) \\
\linenumber \TAB\TAB $\astate = \textsc{MoveTuples}(Y,y,
K^{\H\H}\! \shortrightarrow\! K^{\H\L},\astate)$; \TAB
$\astate = \textsc{MoveTuples}(Y,y,
K^{\L\H}\! \shortrightarrow\! K^{\L\L},\astate)$\\

\linenumber \TAB \ELSE\IF (
$y \in (\pi_Y K^{\H\L} \cup \pi_Y K^{\L\L})$ \AND 
$|\sigma_{Y = y} K | \geq \frac{3}{2} N^{\eps}$) \\
\linenumber \TAB\TAB $\astate = \textsc{MoveTuples}(Y,y,
K^{\H\L}\! \shortrightarrow\! K^{\H\H},\astate)$; \TAB 
$\astate = \textsc{MoveTuples}(Y,y,
K^{\L\L}\! \shortrightarrow\! K^{\L\H},\astate)$\\

\linenumber \RETURN $\astate$ \\
\bottomrule
\end{tabular}
};
\end{tikzpicture}

\end{center}
\caption{
$\textsc{MajorRebalance}(\astate)$ performs major rebalancing on a state $\astate = (\eps,N, \dbeps,\inst{V})$ supporting the maintenance of a triangle query. 
$\textsc{StrictPartition}(\dbeps,N^{\eps})$ strictly repartitions the relations 
in $\dbeps$ with threshold $N^{\eps}$, and 
$\textsc{Recompute}(\inst{V},\dbeps)$  
recomputes the views in $\inst{V}$ using the partitions in $\dbeps$.
Given a relation $K$ with schema $(X,Y)$, 
an $X$-value $x$ and a $Y$-value $y$, 
$\textsc{MinorRebalance}(K,X,x,Y,y,\astate)$ moves tuples 
between relation parts to ensure that the heavy and light part conditions
on values $x$ and $y$ hold.
\textsc{MoveTuples}($X,x,K_{\mathit{src}}\! \shortrightarrow\! K_{\mathit{dst}},\astate$) 
uses \textsc{ApplyUpdate} to move 
all tuples with $X$-value $x$ from relation part $K_{\mathit{src}}$ to relation part $K_{\mathit{dst}}$.
\textsc{ApplyUpdate} depends on the maintained triangle query, 
see Sections~\ref{sec:single-update-nullary},
\ref{sec:nullary_triangle_double_partitioning},
\ref{sec:single_update_full},
\ref{sec:single_update_binary}, and
\ref{sec:single_update_unary}.
}

\label{fig:major_minor}
\end{figure}

\begin{proposition}\label{prop:major_cost}
Given a database $\inst{D}$,
 major rebalancing of an \ivme state of $\db$ supporting the maintenance 
 of any triangle query
takes $\bigO{|\inst{D}|^{\frac{3}{2}}}$ time.  
\end{proposition}

\begin{proof}
Let $\astate = (\eps, N, \dbeps,\inst{V})$ be an \ivme state
supporting the maintenance of any triangle query.  
Consider the procedure \textsc{MajorRebalance} 
from Figure~\ref{fig:major_minor}.
The procedure strictly repartitions the relations
in $\inst{P}$ using the threshold $N^{\eps}$
and recomputes the materialized views in $\inst{V}$
based on the new relation partitions.   
Strictly partitioning the input relations takes $\bigO{|\db|}$ time. 
Propositions \ref{prop:preprocessing_step},
\ref{prop:preprocessing_step_space_double_partition},
\ref{prop:preprocessing_step_full},
\ref{prop:preprocessing_step_binary}, and
\ref{prop:preprocessing_step_unary} state that 
the computation of the initial \ivme state supporting the maintenance
of any triangle query takes $\bigO{|\db|^{\frac{3}{2}}}$
time. From the proofs of these propositions follows that the views in 
$\inst{V}$ can be recomputed in $\bigO{|\db|^{\frac{3}{2}}}$ time.
\end{proof}

The superlinear time of major rebalancing is amortized over $\Omega{(N)}$ updates. After a major rebalancing step, it holds that $|\db| = \frac{1}{2}N$ (after doubling), or $|\db| = \frac{1}{2}N - \frac{1}{2}$ or $|\db| = \frac{1}{2}N - 1$ (after halving, i.e., setting $N$ to $\floor{\frac{1}{2}N}-1$; the two options are due to the floor functions in the size invariant and halving expression). To violate the size invariant $\floor{\frac{1}{4}N} \leq |\db| < N$ and trigger another major rebalancing, the number of required updates is at least $\frac{1}{4}N$.
Section~\ref{sec:amortization} proves the amortized $\bigO{|\db|^{\frac{1}{2}}}$ time of major rebalancing.

\paragraph*{Minor Rebalancing}
After each update $\delta R = \{(\deltaA,\deltaB) \mapsto \p\}$, \ivme checks whether the light and heavy part conditions still hold for $\deltaA$ and $\deltaB$.
If $R$ is partitioned on variable $A$, the relation partition consists of the heavy part 
$R^{\H}$ and the light part $R^{\L}$. 
By Definition 
\ref{def:loose_relation_partition}, 
the heavy and light part conditions on $\deltaA$ are 
$|\sigma_{A=\deltaA} R^{\H}| \geq \frac{1}{2}N^{\eps}$
and
$|\sigma_{A=\deltaA} R^{\L}| < \frac{3}{2}N^{\eps}$, respectively. 
If the first condition is violated, all tuples in $R^{\H}$ with the $A$-value $\deltaA$ are moved to $R^{\L}$ and the affected views are updated; similarly, if the second condition is violated, all tuples with the $A$-value $\deltaA$ are moved from $R^{\L}$ to $R^{\H}$, followed by updating the affected views. 

If $R$ is double partitioned on $(A,B)$, the relation 
partition consists of the parts $R^{\H\H}$, $R^{\H\L}$, 
$R^{\L\H}$, and $R^{\L\L}$. Then, the heavy and light part conditions 
must be checked not only for the $A$-value $\deltaA$ but also for 
the $B$-value $\deltaB$. 
From Definition~\ref{def:loose_double_relation_partition}, 
the heavy and light part conditions on $\deltaA$  are 
$|\sigma_{A = \alpha} R| \geq \frac{1}{2} N^{\eps}$
and respectively
$|\sigma_{A = \alpha} R| < \frac{3}{2} N^{\eps}$, 
where $R$ is obtained by taking the union of the parts of $R$.
If the update $\delta R$ violates the first condition, 
all tuples with $A$-value $\deltaA$ are moved from the relation parts in which $A$
is heavy to the relation parts in which $A$ is light, that is,  
from $R^{\H\H}$ and $R^{\H\L}$ 
to $R^{\L\H}$ and $R^{\L\L}$, respectively.
If the update violates the second condition, all tuples with $A$-value
$\deltaA$ are moved in the opposite direction, 
from $R^{\L\H}$ and $R^{\L\L}$ to $
R^{\H\H}$ and $R^{\H\L}$. 
In both cases, the affected views are updated.
The heavy and light part conditions on $B$-value $\beta$ are ensured in a similar 
way. As a result of an update, 
both values $\deltaA$ and $\deltaB$ might change from 
light to heavy or vice-versa, but it is impossible 
that one value changes from light to heavy and the other one from 
heavy to light. The minor rebalancing steps followed by
updates to the other relations $S$ and $T$ are analogous.

The procedure \textsc{MinorRebalance} in 
Figure~\ref{fig:major_minor} describes a minor rebalancing step 
on an \ivme state
following an update $\delta K = \{(x,y) \mapsto \p\}$ 
to a relation $K$ over schema $(X,Y)$. 
If $K$ is single partitioned, the heavy and light part conditions 
are checked for $X$-value $x$ only (Lines 1-5).  
If it is double partitioned, the conditions are checked 
for both $X$-value $x$ and $Y$-value $y$ (Lines 6-14). 
Tuples are moved between relation parts using the procedure 
\textsc{MoveTuples} in 
Figure~\ref{fig:major_minor}.
Given a variable $X$ in the schema 
of relation $K$, an  $X$-value $x$,
a source 
relation part $K_{\mathit{src}}$, and a target relation
part $K_{\mathit{dst}}$, the procedure  
\textsc{MoveTuples} moves all tuples with $X$-value $x$
from $K_{\mathit{src}}$ to 
part $K_{\mathit{dst}}$.
A tuple $\inst{x}$ is moved from  $K_{\mathit{src}}$ to  $K_{\mathit{dst}}$
by using the procedure \textsc{ApplyUpdate} 
that updates the multiplicities of $\inst{x}$ in  $K_{\mathit{dst}}$ and $K_{\mathit{src}}$ 
and maintains the materialized views in the \ivme state. 
Sections~\ref{sec:single-update-nullary}, 
\ref{sec:nullary_triangle_double_partitioning}, 
\ref{sec:single_update_full}, 
\ref{sec:single_update_binary}, and
\ref{sec:single_update_unary} 
give the definition of \textsc{ApplyUpdate} for each triangle query.
If $K$ is single partitioned,
\textsc{MoveTuples} is called at most once in \textsc{MinorRebalance}. 
If $K$ is double partitioned,
\textsc{MoveTuples} can be called up to four times, two times per $x$ and $y$, to meet the heavy and light part conditions.

\begin{proposition}
  \label{prop:minor_cost}
  Given a database $\inst{D}$ and $\eps \in [0,1]$ minor rebalancing 
  of an \ivme state of $\inst{D}$ supporting the maintenance of any triangle query
  takes $\bigO{|\inst{D}|^{\eps + \max\{\eps,1-\eps\}}}$ time.  
\end{proposition}

\begin{proof}
  Consider an \ivme state $\astate=(\eps, N, \dbeps, \inst{V})$
  and an update $\delta R = \{(\deltaA, \deltaB) \mapsto \p\}$ to relation $R$.
  The analysis for updates to $S$ and $T$ is similar. 
  If $R$ is single partitioned, \textsc{MinorRebalance}
  calls \textsc{MoveTuples} at most once;
  if $R$ is double partitioned, \textsc{MinorRebalance}
  calls \textsc{MoveTuples} at most four times.  
  Consider the worst case when 
  $R$ is double partitioned and both values $\deltaA$ and $\deltaB$
  change from heavy to light or vice-versa. 
  If they change from heavy to light, the procedure 
moves fewer than $\frac{1}{2}N^{\eps}$ 
tuples with $A$-value $\deltaA$ 
and 
fewer than $\frac{1}{2}N^{\eps}$ 
tuples with $B$-value $\deltaB$. 
If the two values 
  change from light to heavy,
the procedure moves  
fewer than $\frac{3}{2}N^{\eps} + 1$ tuples
with $A$-value $\deltaA$ and
fewer than $\frac{3}{2}N^{\eps} + 1$ tuples
with $B$-value $\deltaB$.
  Each tuple move performs one delete and one insert 
  by executing \textsc{ApplyUpdate}.  
  From Propositions 
  \ref{prop:single_step_time},
  \ref{prop:single_step_time_double_partition},
  \ref{prop:single_step_time_full},
  \ref{prop:single_step_time_binary}, and
  \ref{prop:single_step_time_unary}
  follows that, regardless of the maintained triangle query,
\textsc{ApplyUpdate} runs in time 
  $\bigO{|\db|^{\max\{\eps,1-\eps\}}}$. 
    Since there are 
  $\bigO{N^{\eps}}$ such operations, the procedure \textsc{MinorRebalance} requires 
  $\bigO{|\db|^{\eps+\max\{\eps,1-\eps\}}}$ time. As 
  $|\db|=\Theta(N)$, minor rebalancing runs in time 
  $\bigO{|\db|^{\eps+\max\{\eps,1-\eps\}}}$.
\end{proof}

The (super)linear time of minor rebalancing is amortized over $\Omega(N^{\eps})$ updates. 
This lower bound on the number of updates comes from the relation partition conditions (see Definition~\ref{def:loose_relation_partition}), namely from the gap between the two thresholds in these conditions. 
Section~\ref{sec:amortization} proves the amortized $\bigO{|\db|^{\max\{\eps,1-\eps\}}}$ time of minor rebalancing.

\begin{figure}[t]
\begin{tikzpicture}
\node at(-0.5,0)[anchor=north west] {
\renewcommand{\arraystretch}{1.1}
\setcounter{magicrownumbers}{0}
\hspace{-3mm}
\begin{tabular}{@{\hspace{3mm}}l}
\toprule
\textsc{OnUpdate}(update $\delta R$, state $\astate$) \\
\midrule
\linenumber \LET $\delta R = \{(\deltaA,\deltaB) \mapsto \p\}$  \\
\linenumber \LET $\astate = (\eps, N, \inst{P}, \inst{V})$  \\

\linenumber \LET $R^{r} = \textsc{AffectedPart}(\delta R,\astate)$ \\

\linenumber $\textsc{ApplyUpdate}(\delta R^r = \{(\deltaA,\deltaB) \mapsto \p\},\astate)$\\

\linenumber \IF($|\inst{D}| = N$)\\
\linenumber \TAB $N = 2N$ \\
\linenumber \TAB $\astate = {\textsc{MajorRebalance}}(\astate)$ \\
\linenumber \ELSE \IF($|\inst{D}| < \floor{\frac{1}{4}N}$) \\
\linenumber \TAB $N = \floor{\frac{1}{2}N}-1$ \\
\linenumber \TAB $\astate = \textsc{MajorRebalance}(\astate)$ \\

\linenumber \ELSE \IF ($A$ is light in $R^r$ \AND 
$|\sigma_{A = \deltaA} R | \geq \frac{3}{2} N^{\eps}$ \OR \\
\linenumber \TAB\TAB\TAB $B$ is light in $R^r$ \AND 
$|\sigma_{B = \deltaB} R | \geq \frac{3}{2} N^{\eps}$ \OR \\
\linenumber \TAB\TAB\TAB $A$ is heavy in $R^r$ \AND 
$|\sigma_{A = \deltaA} R | < \frac{1}{2} N^{\eps}$ \OR \\
\linenumber \TAB\TAB\TAB $B$ is heavy in $R^r$ \AND 
$|\sigma_{B = \deltaB} R | < \frac{1}{2} N^{\eps}$) \\
\linenumber \TAB $\astate = \textsc{MinorRebalance}(R, A,\alpha, B, \beta,  \astate)$ \\ 
\linenumber \RETURN $\astate$\\
\bottomrule
\end{tabular}
};

\node at(15.6,0)[anchor=north east] {
\renewcommand{\arraystretch}{1.1}
\setcounter{magicrownumbers}{0}
\begin{tabular}{@{\hspace{3mm}}l}
\toprule
\textsc{AffectedPart}(update $\delta R$, state $\astate$) \\
\midrule 
\linenumber \LET $\delta R = \{(\deltaA,\deltaB) \mapsto \p\}$  \\
\linenumber \LET $\astate = (\eps, N, \inst{P}, \inst{V})$ \\
\linenumber \IF ($R$ is single partitioned) \\
\linenumber \TAB \IF ($\deltaA \in \pi_A R^{\H}$ \OR $\eps =0$) \\
\linenumber \TAB \TAB \RETURN $R^{\H}$\\

\linenumber \TAB \ELSE \\
\linenumber \TAB \TAB \RETURN $R^{\L}$\\

\linenumber \ELSE\IF ($R$ is double partitioned) \\
\linenumber \TAB \IF ($(\deltaA,\deltaB) \in R^{\H\H}$ \OR $\eps =0$) \\
\linenumber \TAB \TAB \RETURN $R^{\H\H}$\\

\linenumber \TAB \ELSE \IF ($(\deltaA,\deltaB) \in R^{\H\L}$) \\
\linenumber \TAB \TAB \RETURN $R^{\H\L}$\\

\linenumber \TAB \ELSE \IF ($(\deltaA,\deltaB) \in R^{\L\H}$) \\
\linenumber \TAB \TAB \RETURN $R^{\L\H}$\\

\linenumber \TAB \ELSE \\
\linenumber \TAB \TAB \RETURN $R^{\L\L}$\\

\bottomrule
\end{tabular}
};
\end{tikzpicture}

\caption{
Maintaining an \ivme state supporting the maintenance of any triangle 
query 
 under a single-tuple update and performing rebalancing. 
The procedure \textsc{OnUpdate} takes as input an update $\delta R$ and 
an \ivme state $\astate$ of database $\inst{D}$
and returns a new state that results from applying 
$\delta R$ to $\astate$ and, if necessary, rebalancing partitions.
The procedure \textsc{AffectedPart} determines the relation part in $\astate$ affected by the update. 
\textsc{ApplyUpdate} depends on the maintained triangle query, 
see Sections~\ref{sec:single-update-nullary},
\ref{sec:nullary_triangle_double_partitioning},
\ref{sec:single_update_full},
\ref{sec:single_update_binary}, and
\ref{sec:single_update_unary}.
\textsc{MajorRebalance} and \textsc{MinorRebalance}
are given in Figure \ref{fig:major_minor}.  
The \textsc{OnUpdate} procedures for updates to $S$ and $T$ are analogous. 
}
\label{fig:onUpdate}
\end{figure}

Figure~\ref{fig:onUpdate} gives the trigger procedure \textsc{OnUpdate} that maintains an \ivme state of a database $\inst{D}$ under a single-tuple update 
$\delta R = \{(\deltaA,\deltaB) \mapsto \p\}$
to relation $R$ and, if necessary, rebalances partitions; the procedures for updates to $S$ and $T$ are analogous.
The procedure first calls \textsc{AffectedPart} to determine in constant time 
which part $R^r$ of $R$ is affected by the update. We first consider the case 
when $R$ is single partitioned. 
The update targets $R^{\H}$ if 
this relation part  already contains a tuple with the same $A$-value $\deltaA$, 
or $\eps$ is set to $0$; otherwise, the update targets $R^{\L}$.
When $\eps = 0$, all tuples are in $R^{\H}$, while $R^{\L}$ remains empty.
Although this behavior is not required by \ivme (without the condition 
$\eps=0$, $R^{\L}$ would contain only tuples whose $A$-values have the degree of $1$, and $R^{\H}$ would contain all other tuples), it allows us to recover existing IVM approaches, such as classical IVM for the nullary and ternary triangle queries; 
by setting $\eps$ to $0$, \ivme ensures that all tuples are in $R^{\H}$.
The case when $R$ is double partitioned is analogous.
The update targets $R^{\H\H}$ if 
$R^{\H\H}$ contains the tuple $(\deltaA,\deltaB)$ or $\eps = 0$; 
the update targets $R^{\H\L}$ or $R^{\L\H}$ if they already contain $(\deltaA,\deltaB)$;
otherwise, the update targets $R^{\L\L}$. 
The procedure \textsc{OnUpdate} then invokes \textsc{ApplyUpdate}.
If the update causes a violation of the size invariant $\floor{\frac{1}{4}N} \leq |\db| < N$, the procedure invokes \textsc{MajorRebalance} from 
Figure~\ref{fig:major_minor}
to recompute the relation partitions and auxiliary views. 
Otherwise, if any heavy or light part condition is violated, it calls 
\textsc{MinorRebalance} from Figure~\ref{fig:major_minor} to
move tuples between the parts of relation $R$ and ensure that 
these conditions hold again.

%% file: amortization.tex
\section{Amortizing Rebalancing Time}
\label{sec:amortization}
Sections~\ref{sec:count}-\ref{sec:unary} show that 
any \ivme state supporting the maintenance 
of a triangle query can be maintained in sublinear time
under a single-tuple update. 
The sublinear maintenance time requires that the size invariant 
and the heavy and light part conditions are preserved 
for the relation partitions in \ivme states.
To guarantee this, \ivme performs major and minor rebalancing
steps, which can take superlinear time
as stated 
in Propositions~\ref{prop:major_cost} and \ref{prop:minor_cost}.
We nevertheless show in this section that the amortized rebalancing costs and thus
the overall amortized maintenance time  over a sequence of updates 
remains sublinear. 
\begin{proposition}\label{prop:amortized_update_time}
Given a database $\db$, $\eps\in[0,1]$, and an \ivme state 
$\astate$ of $\inst{D}$ supporting the maintenance of
any triangle query, \ivme maintains $\astate$ under a single-tuple update to any input relation and performs rebalancing in $\bigO{|\db|^{\max\{\eps,1-\eps\}}}$ amortized time.
\end{proposition}

\begin{proof}
Let $\astate_0 = (\eps,N_{0},\inst{P}_{0},\inst{V}_{0})$ be the initial \ivme state of a database $\db_0$ and 
$\upd_0, \upd_1, \ldots ,\upd_{n-1}$ a sequence of arbitrary single-tuple updates. 
The application of this update sequence to $\astate_0$ yields a sequence 
$\astate_0 \overset{\upd_0}{\longrightarrow} \astate_1 \overset{\upd_1}{\longrightarrow} \ldots \overset{\upd_{n-1}}{\longrightarrow} \astate_{n}$
of \ivme states, 
where $\astate_{i+1}$ is the result of executing the procedure $\textsc{OnUpdate}(\upd_{i},\astate_{i})$ from Figure~\ref{fig:onUpdate}, for $0\leq i < n$.
Let $c_i$ denote the actual execution cost of $\textsc{OnUpdate}(\upd_i,\astate_{i})$. 
For some $\Gamma>0$, we can decompose each $c_i$ as:
\begin{align*}
c_i = c_i^{\update} + c_i^{\major} + c_i^{\minor} + \Gamma, \text{\qquad for } 0 \leq i < n,
\end{align*} 
where $c_i^{\update}$, $c_i^{\major}$, and $c_i^{\minor}$ are the actual costs of the subprocedures \textsc{ApplyUpdate}, \textsc{MajorRebalance}, and
\textsc{MinorRebalance}, respectively, in \textsc{OnUpdate}.
If update $\upd_i$ causes no major rebalancing, then $c_i^{\major} = 0$; 
similarly, if $\upd_i$ causes no minor rebalancing, then $c_i^{\minor} = 0$. 
These actual costs admit the following worst-case upper bounds: 
\\[6pt]
\begin{tabular}{r@{\hskip 0.05in}l@{\hskip 0.1in}l}
$c^{\update}_i$ & $\leq \gamma N_{i}^{\max\{\eps,1-\eps\}}$ & (by 
Propositions~\ref{prop:single_step_time},
  \ref{prop:single_step_time_double_partition},
  \ref{prop:single_step_time_full},
  \ref{prop:single_step_time_binary},
  \ref{prop:single_step_time_unary}), \\[4pt]
$c^{\major}_i$ & $\leq \gamma N_{i}^{\frac{3}{2}}$& {(by Proposition~\ref{prop:major_cost})}, and \\[4pt]
$c^{\minor}_i$ & $\leq \gamma N_{i}^{\eps+\max\{\eps,1-\eps\}}$ & {(by Proposition~\ref{prop:minor_cost})}, 
\end{tabular}\\[6pt]
where $\gamma$ is a constant derived from their asymptotic bounds, and $N_i$ is the threshold base of $\astate_i$.
The costs of major and minor rebalancing can be superlinear in the database size. 

The crux of this proof is to show that assigning a {\em sublinear amortized cost} $\hat{c}_{i}$ to each update $\upd_{i}$ accumulates enough budget to pay for  expensive but less frequent rebalancing procedures.
For any sequence of $n$ updates, our goal is to show that the accumulated amortized cost is no smaller than the accumulated actual cost:
\begin{align}\label{eq:inequality_proof}
\sum_{i=0}^{n-1} \hat{c}_i \geq \sum_{i=0}^{n-1}c_i.
\end{align}  
The amortized cost assigned to an update $\upd_i$ is $\hat{c}_i = \hat{c}^{\update}_i + \hat{c}^{\major}_i + \hat{c}^{\minor}_i + \Gamma$, where
\begin{align*}
\hat{c}^{\update}_i = \gamma N_{i}^{\max\{\eps,1-\eps\}}, \quad
\hat{c}^{\major}_i = 4\gamma N_{i}^{\frac{1}{2}}, \quad
\hat{c}^{\minor}_i = 4\gamma N_{i}^{\max\{\eps,1-\eps\}}, \quad \text{ and }
\end{align*}
$\Gamma$ and $\gamma$ are the constants used to upper bound the actual cost of \textsc{OnUpdate}. As it will be explained in more detail, the number of updates 
between a major rebalancing step caused by update $\upd_i$ and the previous 
major rebalancing step can be as less as 
$\frac{1}{4} N_i$. In order to accumulate enough budget to 
pay for the major rebalancing cost triggered by update $\upd_i$,
the amortized cost $\hat{c}^{\major}_i$ is 
defined as $\gamma N_i^{\frac{3}{2}} / \frac{1}{4}N_i = 4\gamma N_{i}^{\frac{1}{2}}$.
Given that $\upd_i$ is of the form $\delta R = \{(\deltaA, \deltaB) \mapsto \p\}$
and invokes minor rebalancing for $\deltaA$, the 
number of updates since the 
previous minor rebalancing step for $\deltaA$ can be 
as less as $\frac{1}{2}N^{\eps}$. Hence, 
to pay for
the minor rebalancing step for $\deltaA$ invoked by $\upd_i$, our 
budget must be at least 
$\gamma N_i^{\eps + \max\{\eps, 1-\eps\}} / 
\frac{1}{2}N^{\eps} = 2\gamma N_{i}^{\max\{\eps,1-\eps\}}$.
Since we also need to take the rebalancing costs 
for $\deltaB$ into account, we 
define the amortized minor rebalancing cost $\hat{c}^{\minor}_i$ as
$4\gamma N_{i}^{\max\{\eps,1-\eps\}}$.
In contrast to the actual costs $c^{\major}_i$ and $c^{\minor}_i$, the amortized costs $\hat{c}^{\major}_i$ and $\hat{c}^{\minor}_i$ are always nonzero. 

We prove that such amortized costs satisfy Inequality~\eqref{eq:inequality_proof}.
Since $\hat{c}^{\update}_i \geq c^{\update}_i $ for $0 \leq  i < n$, it suffices to show that the following inequalities hold:

\begin{align}
\label{ineq:major}
\text{\em (amortizing major rebalancing) \TAB} & \sum_{i=0}^{n-1} \hat{c}_i^{\major}  \geq \sum_{i=0}^{n-1}c_i^{\major} \qquad\text{ and}\\
\label{ineq:minor}
\text{\em (amortizing minor rebalancing) \TAB} & \sum_{i=0}^{n-1} \hat{c}_i^{\minor}  \geq \sum_{i=0}^{n-1}c_i^{\minor}. 
\end{align}
We prove Inequalities~\eqref{ineq:major} and \eqref{ineq:minor} by induction on the length $n$ of the update sequence.

\paragraph{Major rebalancing.}
\begin{itemize}
\item {\em Base case}:
We show that Inequality~\eqref{ineq:major} holds for $n = 1$. 
The preprocessing stage sets $N_{0} = 2\ztimes|\inst{D}_0| + 1$. 
If the initial database $\inst{D}_0$ is empty, then $N_{0}=1$ and $\upd_0$ triggers major rebalancing (and no minor rebalancing).  
The amortized cost $\hat{c}_0^{\major} = 4 \gamma N_0^{\frac{1}{2}} = 4\gamma$ suffices to cover the actual cost $c_0^{\major}\leq \gamma N_0^{1 + \frac{1}{2}} = \gamma$.
If the initial database is nonempty, $\upd_0$ cannot trigger major rebalancing (i.e., violate the size invariant) because 
$\floor{\frac{1}{4}N_{0}} = \floor{\frac{1}{2}|\inst{D}_0|} \leq |\inst{D}_0| -1$ (lower threshold) and 
$|\inst{D}_0| + 1 < N_{0} = 2\ztimes|\inst{D}_0| + 1$ (upper threshold); 
then, $\hat{c}_0^{\major} \geq c_0^{\major} = 0$. Thus, Inequality~\eqref{ineq:major} holds for $n = 1$. 

\item {\em Inductive step}:
Assumed that Inequality~\eqref{ineq:major} holds for all update sequences of length up to $n-1$, we show it holds for update sequences of length $n$.
If update $\upd_{n-1}$ causes no major rebalancing, then 
$\hat{c}_{n-1}^{\major} = 4 \gamma N_{n-1}^{\frac{1}{2}} \geq 0$ and $c_{n-1}^{\major} = 0$, thus Inequality~\eqref{ineq:major} holds for $n$. 
Otherwise, if applying $\upd_{n-1}$ violates the size invariant, 
the database size $|\inst{D}_{n}|$ is either 
$\floor{\frac{1}{4}N_{n-1}}-1$ or $N_{n-1}$. 
Let $\astate_j$ be the state created after the previous major rebalancing or, 
if there is no such step, the initial state. 
For the former ($j>0$), the major rebalancing step ensures
$|\inst{D}_j| = \frac{1}{2}N_j$ after doubling and 
$|\inst{D}_j| = \frac{1}{2}N_j - \frac{1}{2}$ or $|\inst{D}_j| = \frac{1}{2}N_j - 1$ after halving the threshold base $N_j$;
for the latter ($j=0$), the preprocessing stage ensures $|\inst{D}_j| = \frac{1}{2}N_j - \frac{1}{2}$.
The threshold base $N_j$ changes only with major rebalancing, thus $N_j = N_{j+1}=\ldots=N_{n-1}$.
The number of updates needed to change the database size from 
$|\inst{D}_j|$ to $|\inst{D}_{n}|$ (i.e., between two major rebalancing) is at least $\frac{1}{4}N_{n-1}$ since
$\min\{ \frac{1}{2}N_j - 1 - (\floor{\frac{1}{4}N_{n-1}} - 1), N_{n-1} - \frac{1}{2}N_j \} \geq \frac{1}{4}N_{n-1}$. 
Then,

\begin{align*}
  \sum_{i=0}^{n-1} \hat{c}_i^{\major} &\geq \sum_{i=0}^{j-1} c_i^{\major} + \sum_{i=j}^{n-1} \hat{c}_i^{\major} \qquad\qquad\quad && \textit{(by induction hypothesis)}\\
  & = \sum_{i=0}^{j-1} c_i^{\major} + \sum_{i=j}^{n-1} 4\gamma 
  N_{n-1}^{\frac{1}{2}} && \textit{($N_j = \ldots = N_{n-1}$)}\\
  & \geq \sum_{i=0}^{j-1} c_i^{\major} + \frac{1}{4}N_{n-1} \,4 \gamma
  N_{n-1}^{\frac{1}{2}} && \textit{(at least $\frac{1}{4}N_{n-1}$ updates)}\\
  & = \sum_{i=0}^{j-1} c_i^{\major} + \gamma N_{n-1}^{\frac{3}{2}} \\
  & \geq \sum_{i=0}^{j-1} c_i^{\major} + c_{n-1}^{\major} = \sum_{i=0}^{n-1} c_i^{\major} && \textit{($c_{j}^{\major} = \ldots = c_{n-2}^{\major} = 0$)}.
\end{align*}
\end{itemize}

\noindent Thus, Inequality~\eqref{ineq:major} holds for update sequences of length $n$.

\paragraph{Minor rebalancing.}
When the degree of a value in a partition changes such that the heavy or light part condition no longer holds, minor rebalancing moves the affected tuples between the 
relation parts.
To prove Inequality~\eqref{ineq:minor}, we decompose the cost of minor rebalancing per relation and data value over a variable in the schema of the relation. 
\begin{align*}
c_i^{\minor} & =  \sum_{a \in \Dom(A)} (c_i^{R,a} + c_i^{T,a}) + \sum_{b \in \Dom(B)} (c_i^{R,b} + c_i^{S,b}) +\sum_{c \in \Dom(C)} (c_i^{T,c} + c_i^{R,c}) \\[3pt]
\hat{c}_i^{\minor} & =  \sum_{a \in \Dom(A)} (\hat{c}_i^{R,a} + \hat{c}_i^{T,a}) + \sum_{b \in \Dom(B)} (\hat{c}_i^{R,b} + \hat{c}_i^{S,b}) +\sum_{c \in \Dom(C)} (\hat{c}_i^{T,c} + \hat{c}_i^{R,c})
\end{align*}

We write $c_i^{R,\deltaA}$ and $\hat{c}_i^{R,\deltaA}$ to denote the actual and respectively amortized costs of minor rebalancing caused by update $\upd_i$, for relation $R$ and an $A$-value $\deltaA$.
Recall that if update $\upd_i$ is of the form $\delta R = \{(\deltaA,\deltaB) \mapsto \p\}$ and
$R$ is single partitioned, the update can cause 
 minor rebalancing for $A$-value $\deltaA$. If $R$ is double partitioned,
the update can cause minor rebalancing for $A$-value $\deltaA$, or $B$-value 
$\deltaB$, or for both. Hence, if $\upd_i$ is of the form 
$\delta R = \{(\deltaA,\deltaB) \mapsto \p\}$ 
and causes any rebalancing, we have 
$c_i^{R,\deltaA} + c_i^{R,\deltaB} = c_i^\minor \leq \gamma N_i^{\eps + \max\{\eps, 1-\eps\}}$;
otherwise, $c_i^{R,\deltaA} = c_i^{R,\deltaB} = 0$. 
If $\upd_i$ is of the form $\delta R = \{(\deltaA,\deltaB) \mapsto \p\}$,  
we set 
$\hat{c}_i^{R,\deltaA} = \hat{c}_i^{R,\deltaB} =\frac{1}{2} \hat{c}_i^\minor = 
2\gamma N_i^{\max\{\eps, 1-\eps\}}$ regardless of whether $\upd_i$ causes minor rebalancing or not; otherwise, $\hat{c}_i^{R,\deltaA} = \hat{c}_i^{R,\deltaB} = 0$.  
The actual costs $c_i^{S,b}$, $c_i^{S,c}$, $c_i^{T,c}$, and $c_i^{T,a}$ and the amortized costs $\hat{c}_i^{S,b}$, $\hat{c}_i^{S,c}$, $\hat{c}_i^{T,c}$, and $\hat{c}_i^{T,a}$ are defined similarly.

We prove that for $R$ and any $a \in \Dom(A)$, the following inequality holds:
\begin{align}
\label{ineq:minor_decomposed}
\sum_{i=0}^{n-1} \hat{c}_i^{R,a}  \geq \sum_{i=0}^{n-1}c_i^{R,a}.
\end{align}
The proof of the inequality $\sum_{i=0}^{n-1} \hat{c}_i^{R,b}  \geq \sum_{i=0}^{n-1}c_i^{R,b}$  for any $b \in \Dom(B)$ and the inequalities for the other two relations $S$
and $T$ are analogous.
Inequality~\eqref{ineq:minor} follows directly from these inequalities.

We prove Inequality~\eqref{ineq:minor_decomposed} for an arbitrary $a \in \Dom(A)$ by induction on the length $n$ of the update sequence.
\begin{itemize}
\item {\em Base case}:
We show that Inequality~\eqref{ineq:minor_decomposed} holds for $n = 1$. 
Assume that update $\upd_0$ is of the form $\delta R = \{(\deltaA,\deltaB) \mapsto \p\}$; 
otherwise, $\hat{c}_0^{R,\deltaA} = c_0^{R,\deltaA} = 0$, and Inequality~\eqref{ineq:minor_decomposed} follows trivially for $n=1$.  
If the initial database is empty, $\upd_0$ triggers major rebalancing but no minor rebalancing, thus $\hat{c}_0^{R,\deltaA} = 2\gamma N_{0}^{\max\{\eps,1-\eps\}} \geq  c_0^{R,\deltaA} = 0$.
If the initial database is nonempty, each relation is partitioned using the threshold $N_0^{\eps}$. 
For update $\upd_0$ to trigger minor rebalancing for $A$-value  $\deltaA$, the degree of 
$\deltaA$ in $R$ 
has to either decrease from $\ceil{N_0^{\eps}}$ to $\ceil{\frac{1}{2}N_0^{\eps}}-1$ (heavy to light) or increase from $\ceil{N_0^{\eps}}-1$ to $\ceil{\frac{3}{2}N_0^{\eps}}$ (light to heavy). The former happens only if $\ceil{N_0^{\eps}}=1$ and update $\upd_0$ removes the last tuple with the $A$-value $\deltaA$ from $R$, thus no minor rebalancing is needed; 
the latter cannot happen since update $\upd_0$ can increase $|\sigma_{A=\deltaA}R|$ to at most $\ceil{N_0^{\eps}}$, and $\ceil{N_0^{\eps}} < \ceil{\frac{3}{2}N_0^{\eps}}$. 
In any case, $\hat{c}_0^{R,\deltaA} \geq c_0^{R,\deltaA}$, which implies that Inequality~\eqref{ineq:minor_decomposed} holds for $n = 1$. 

\item {\em Inductive step}:
Assumed that Inequality~\eqref{ineq:minor_decomposed} holds for all update sequences of length up to $n-1$, we show that  it holds for update sequences of length $n$.
Consider that update $u_{n-1}$ is of the form $\delta R = \{(\deltaA,\deltaB) \mapsto \p\}$ and causes minor rebalancing for $\deltaA$; otherwise, $\hat{c}_{n-1}^{R,\deltaA} \geq 0$ and $c_{n-1}^{R,\deltaA} = 0$, and Inequality~\eqref{ineq:minor_decomposed} follows trivially for $n$.  
Let $\astate_j$ be the state created after the previous major rebalancing or, if there is no such step, the initial state. The threshold changes only with major rebalancing, thus $N_j=N_{j+1}=\ldots=N_{n-1}$. 
Depending on whether there exist minor rebalancing steps since state $\astate_j$, we distinguish two cases:

  \begin{itemize}
    \item[Case 1:] There is no minor rebalancing caused by an update of the form
    $\delta R = \{(\deltaA,\deltaB') \mapsto \p'\}$ since state $\astate_j$; thus, 
    we have $c_{j}^{R,\deltaA} = \ldots = c_{n-2}^{R,\deltaA} = 0$.
    From state $\astate_j$ to state $\astate_{n}$, the number of tuples with the
    $A$-value $\deltaA$ either decreases from at least $\ceil{N_{j}^{\eps}}$ to $\ceil{\frac{1}{2}N_{n-1}^{\eps}}-1$ (heavy to light) or increases from at most $\ceil{N_{j}^{\eps}}-1$ to $\ceil{\frac{3}{2}N_{n-1}^{\eps}}$ (light to heavy). 
    For this change to happen, the number of updates needs to be greater than $\frac{1}{2}N_{n-1}^{\eps}$ since 
    $N_j=N_{n-1}$ and
    $\min\{ \ceil{N_{j}^{\eps}} - (\ceil{\frac{1}{2}N_{n-1}^{\eps}} - 1),
      \ceil{\frac{3}{2}N_{n-1}^{\eps}} - (\ceil{N_{j}^{\eps}} -1)
     \} > \frac{1}{2}N_{n-1}^{\eps}$.

    \item[Case 2:] There is at least one minor rebalancing step for $\deltaA$ caused by an update of the form $\delta R = \{(\deltaA,\deltaB') \mapsto \p'\}$ since state $\astate_j$. 
    Let $\astate_\ell$ denote the state created after the previous minor rebalancing 
    for $\alpha$ caused by an update of this form; thus, $c_{\ell}^{R,\deltaA} = \ldots = c_{n-2}^{R,\deltaA} = 0$.
    The minor rebalancing steps creating $\astate_\ell$ and $\astate_{n}$ move tuples with the $A$-value $a$ between the relation parts of $R$ in opposite directions
    with respect to heavy and light.
    From state $\astate_\ell$ to state $\astate_n$, the number of such tuples either decreases from $\ceil{\frac{3}{2}N_{l}^{\eps}}$ to $\ceil{\frac{1}{2}N_{n-1}^{\eps}}-1$ (heavy to light) or increases from $\ceil{\frac{1}{2}N_{l}^{\eps}}-1$ to $\ceil{\frac{3}{2}N_{n-1}^{\eps}}$ (light to heavy). 
    For this change to happen, the number of updates needs to be greater than $N_{n-1}^{\eps}$ since
    $N_l = N_{n-1}$ and
    $\min\{\ceil{\frac{3}{2}N_{l}^{\eps}} - (\ceil{\frac{1}{2}N_{n-1}^{\eps}}-1), 
      \ceil{\frac{3}{2}N_{n-1}^{\eps}} - (\ceil{\frac{1}{2}N_{l}^{\eps}} - 1)
    \} > N_{n-1}^{\eps}$. 
  \end{itemize}

 Let $k = j$ if  Case 1 holds and $k = \ell$ if Case 2 holds. By the above analysis, 
  there must be more than $\frac{1}{2}N_{n-1}^{\eps}$
 updates between $\astate_{k}$ and $\astate_{n}$. Hence, 
  
\begin{align*}
  \hspace{-0.2cm}
  \sum_{i=0}^{n-1} \hat{c}_i^{R,\deltaA} &\geq \sum_{i=0}^{k-1} c_i^{R,\deltaA} + \sum_{i=k}^{n-1} \hat{c}_i^{R,\deltaA} \qquad\qquad\qquad && \textit{(by induction hypothesis)}\\
  & = \sum_{i=0}^{k-1} c_i^{R,\deltaA} + \sum_{i=k}^{n-1} 2 \gamma N_{n-1}^{\max\{\eps,1-\eps\}} && \textit{($N_k = \ldots = N_{n-1}$)} \\
  & > \sum_{i=0}^{k-1} c_i^{R,\deltaA} + \frac{1}{2} N_{n-1}^{\eps} 2 \gamma N_{n-1}^{\max\{\eps,1-\eps\}} && \textit{(more than $\frac{1}{2} N_{n-1}^{\eps}$ updates)} \\
  & \geq \sum_{i=0}^{k-1} c_i^{R,\deltaA} + c_{n-1}^{R,\deltaA} = \sum_{i=0}^{n-1} c_i^{R,\deltaA} && \textit{($c_{k}^{R,\deltaA} = \ldots = c_{n-2}^{R,\deltaA} = 0$)}.
\end{align*}

This implies that Inequality~\eqref{ineq:minor_decomposed} holds for update sequences of length $n$.
\end{itemize}

The inductive analysis shows that Inequality~\eqref{eq:inequality_proof} holds when the amortized cost of  $\textsc{OnUpdate}(\upd_i, \astate_i)$ is
$$\hat{c}_{i} = 
\gamma N_{i}^{\max\{\eps,1-\eps\}} + 
4\gamma N_{i}^{\frac{1}{2}} +
4\gamma N_{i}^{\max\{\eps,1-\eps\}} + 
\Gamma, \text{\quad  for } 0\leq i < n,$$
where $\Gamma$ and $\gamma$ are constants. 
The amortized cost $\hat{c}^{\major}_i$ 
of major rebalancing is $4\gamma N_{i}^{\frac{1}{2}}$, 
and the amortized cost $\hat{c}^{\minor}_i$ 
 of minor rebalancing is $4\gamma N_{i}^{\max\{\eps,1-\eps\}}$. 
From the size invariant $\floor{\frac{1}{4}N_{i}} \leq |\db_{i}| < N_{i}$ follows that $|\db_i| < N_i < 4(|\db_i|+1)$ for $0\leq i < n$, where $|\db_i|$ is the database size before update $\upd_i$. 
This implies that for any database $\inst{D}$,
the amortized major rebalancing time is $\bigO{|\db|^{\frac{1}{2}}}$,
the amortized minor rebalancing time is $\bigO{|\db|^{\max\{\eps, 1-\eps\}}}$, and
the overall amortized update time of \ivme is $\bigO{|\inst{D}|^{\max\{\eps, 1-\eps\}}}$.
\end{proof}

%% file: lowerBound.tex
\section{A Lower Bound on the Maintenance of Triangle Queries}
\label{sec:lowerbound}
In this section we prove Proposition~\ref{prop:lower_bound_triangle}, 
which states a lower bound on the trade-off between amortized 
update time and enumeration delay for the maintenance 
of triangle queries, conditioned on the \OMv conjecture~\cite{Henzinger:OMv:2015}.

\smallskip
Proposition~\ref{prop:lower_bound_triangle}.
\textit{
For any $\gamma > 0$ and database $\db$,
there is no algorithm that incrementally maintains the result of
any triangle query under single-tuple updates to $\db$ with arbitrary preprocessing 
time, $\bigO{|\db|^{\frac{1}{2} - \gamma}}$ amortized update time, and $\bigO{|\db|^{1 - \gamma}}$ enumeration delay, unless the \OMv conjecture fails.
} 

\smallskip

\nop{
Together with the upper bounds in Theorem~\ref{theo:main_result_triangle}, 
this lower bound implies the optimality result in Corollary \ref{cor:ivme_optimal_triangle}.
}

The proof relies on the Online Vector-Matrix-Vector Multiplication (\OuMv) conjecture, which is implied by the \OMv conjecture (Conjecture \ref{conj:omv}).
First, we give the definition of the \OuMv problem and state 
the corresponding conjecture.

\begin{definition}[Online Vector-Matrix-Vector Multiplication (\OuMv)~\cite{Henzinger:OMv:2015}]\label{def:OuMv}
We are given an $n \times n$ Boolean matrix $\vecnormal{M}$ and  receive $n$ pairs of Boolean column-vectors of size $n$, denoted by $(\vecnormal{u}_1,\vecnormal{v}_1), \ldots, (\vecnormal{u}_n,\vecnormal{v}_n)$; after seeing each pair $(\vecnormal{u}_i,\vecnormal{v}_i)$, we output the product 
$\vecnormal{u}_i^{\text{T}} \vecnormal{M} \vecnormal{v}_i$ before we see the next pair.
\end{definition}

\begin{conjecture}[\OuMv Conjecture, Theorem 2.7 in~\cite{Henzinger:OMv:2015}]\label{conj:OuMv}
For any $\gamma > 0$, there is no algorithm that solves \OuMv in time 
$\bigO{n^{3-\gamma}}$.
\end{conjecture}

The following proof of Proposition 
\ref{prop:lower_bound_triangle}
reduces the \OuMv problem to  
the problem of incrementally maintaining a triangle query.
This reduction implies that 
if there is an algorithm that incrementally maintains a triangle query 
under single-tuple updates 
with arbitrary preprocessing time, $\bigO{|\inst{D}|^{\frac{1}{2}-\gamma}}$
amortized update time, and  $\bigO{|\inst{D}|^{1-\gamma}}$ enumeration delay
for some $\gamma > 0$ and database $\inst{D}$, then  
the \OuMv problem can be solved in subcubic time. 
This contradicts the \OuMv conjecture and, consequently, the \OMv conjecture.

\begin{proof}[Proof of Proposition~\ref{prop:lower_bound_triangle}]
The proof is inspired by the lower bound proof for maintaining
non-hierarchical  
Boolean conjunctive queries~\cite{BerkholzKS17}.
 Let $\triangle$ be a triangle query of arbitrary 
arity.  
For the sake of contradiction, assume that there is an incremental maintenance algorithm $\mathcal{A}$ that maintains 
$\triangle$ under single-tuple updates 
with arbitrary preprocessing time, $\bigO{|\inst{D}|^{\frac{1}{2}-\gamma}}$ amortized update time, 
and $\bigO{|\inst{D}|^{1-\gamma}}$ enumeration delay, for some
$\gamma >0$. We show 
that this algorithm can be used to design an algorithm $\mathcal{B}$ that solves the $\OuMv$ problem in subcubic time, which contradicts the \OuMv conjecture.

\begin{figure}[t]
\begin{center}
\renewcommand{\arraystretch}{1.0}
\setcounter{magicrownumbers}{0}
\begin{tabular}{l}
\toprule 
{\textsc{SolveOuMv}}(matrix $\vecnormal{M}$, vectors $\vecnormal{u}_1,\vecnormal{v}_1, \ldots, \vecnormal{u}_n,\vecnormal{v}_n$)\\
\midrule
\linenumber \LET $\astate = \text{initial \ivme state of the empty database }$ \\
\linenumber $\FOREACH\STAB (i,j) \in \vecnormal{M} \STAB\DO$ \\
\linenumber $\TAB\STAB\delta S = \{\, (i,j) \mapsto \vecnormal{M}(i,j) \,\}$ \\
\linenumber $\TAB\STAB\astate = \textsc{OnUpdate}(\delta S, \astate)$\\
\linenumber $\FOREACH\STAB r = 1, \ldots ,n \STAB\DO$ \\
\linenumber $\TAB\STAB\FOREACH\STAB i = 1, \ldots ,n \STAB\DO$ \\
\linenumber $\TAB\STAB\TAB\STAB\delta R = \{\, (a,i) \mapsto (\vecnormal{u}_r(i) - R(a,i)) \,\}$ \\
\linenumber $\TAB\STAB\TAB\STAB\astate = \textsc{OnUpdate}(\delta R, \astate)$ \\
\linenumber $\TAB\STAB\TAB\STAB\delta T = \{\, (i,a) \mapsto (\vecnormal{v}_r(i) - T(i,a)) \,\}$ \\
\linenumber $\TAB\STAB\TAB\STAB\astate = \textsc{OnUpdate}(\delta T, \astate)$ \\
\linenumber $\TAB\STAB \textbf{output}\STAB(\triangle \neq \emptyset)$ \\
\bottomrule
\end{tabular}

\end{center}
\caption{The procedure \textsc{SolveOuMv} solves the \OuMv problem using an 
incremental algorithm that maintains a triangle query $\triangle$
of arbitrary arity under single-tuple updates.
The state $\astate$ is the initial \ivme state of a database  
with empty relations $R$, $S$ and $T$.
The procedure \textsc{OnUpdate} is given in Figure~\ref{fig:onUpdate} and maintains the triangle query under single-tuple updates.}

\label{fig:lower_bound_reduction}
\end{figure}

 \paragraph{The reduction}
Figure~\ref{fig:lower_bound_reduction} gives the pseudocode of the algorithm  
$\mathcal{B}$, which processes an $\OuMv$ input $(\vecnormal{M}, (\vecnormal{u}_1,\vecnormal{v}_1), \ldots ,(\vecnormal{u}_n,\vecnormal{v}_n))$.
We denote the entry of $\vecnormal{M}$ in row $i$ and column $j$ by $\vecnormal{M}(i,j)$ and the $i$-th component of $\vecnormal{v}$ by $\vecnormal{v}(i)$.
The algorithm first constructs the initial \ivme state $\astate$ from a
database $\db = \{R,S,T\}$ with empty relations $R$, $S$, and $T$.
Then, it executes at most $n^2$ updates to the relation $S$ such that 
$S = \{\, (i,j) \mapsto \vecnormal{M}(i,j) \,\mid\, i,j \in [n] \,\}$. 
In each round $r \in [n]$, the algorithm executes at most $2n$ updates to the relations $R$ and $T$ 
such that $R = \{\, (a, i) \mapsto \vecnormal{u}_r(i) \,\mid\, i \in [n] \,\}$
and $T = \{\, (i, a) \mapsto \vecnormal{v}_r(i) \,\mid\, i \in [n] \,\}$, where $a$
is some constant.
By construction, $\vecnormal{u}_r^{\text{T}}\vecnormal{M}\vecnormal{v}_r = 1$ if and only if there exist 
$i,j\in [n]$ such that $\vecnormal{u}_r(i) = 1$, $\vecnormal{M}(i,j)=1$, and $\vecnormal{v}_r(j) = 1$, which is equivalent to 
$R(a,i) \ztimes S(i,j) \ztimes T(j,a) = 1$ at the end of round $r$. 
Thus, the algorithm outputs 
$1$ at the end of round $r$ if and only if the result of the triangle query
is nonempty. Nonemptiness of the query result can be checked 
by triggering enumeration and checking whether at least 
one output tuple is reported.

\paragraph{Time analysis}
Constructing the initial state from a database with empty relations 
takes constant time. 
The construction of relation $S$ from $\vecnormal{M}$ requires at most $n^2$ updates. 
Given that the amortized time for each update is $\bigO{|\db|^{\frac{1}{2}-\gamma}}$
and the database size $|\db|$ stays $\bigO{n^2}$, the overall time 
for constructing relation $S$ is 
$\bigO{n^2 \ztimes n^{2 \cdot (\frac{1}{2}-\gamma)}} = \bigO{n^{3-2\gamma}}$. 
In each round, the algorithm performs at most $2n$ updates and needs 
$\bigO{|\inst{D}|^{1-\gamma}}$ time to report the first result tuple or to 
signalize that  the result is empty. 
Hence, the time to execute the updates in a single round is 
$\bigO{2n \cdot n^{2 \cdot (\frac{1}{2}-\gamma)}} = \bigO{n^{2-2\gamma}}$.
The time to report the first result tuple or signalize emptiness 
is $\bigO{n^{2 \cdot (1-\gamma)}} = \bigO{n^{2-2\gamma}}$.
Thus, the overall execution time is 
$\bigO{n^{2-2\gamma}}$ per round and $\bigO{n^{3-2 \gamma}}$ for $n$ rounds. Hence, algorithm $\mathcal{B}$ needs $\bigO{n^{3-2 \gamma}}$ time to solve the \OuMv problem, which contradicts the \OuMv conjecture and, consequently, the \OMv conjecture. 

\nop{
 \paragraph{Reducing the \OuMv problem to arbitrary  triangle queries}
The above reduction from the \OuMv problem to the incremental maintenance of 
the triangle count can easily be turned into 
a reduction from the same problem to the incremental maintenance of 
a triangle query $Q(\inst{F})$ with $1 \leq |\inst{F}| \leq 3$.
We use the same encoding
for the matrix
$\vecnormal{M}$ and the vectors $\inst{u}_r$ and $\inst{v}_r$
for $r \in \{1, \ldots, n\}$. At the end of each round $r$ 
we ask whether the query result 
contains at least one tuple. This is the case if and only if 
$\vecnormal{u}_r^{\text{T}}\vecnormal{M}\vecnormal{v}_r = 1$.
As in the above reduction, the time 
to construct relation $S$ is $\bigO{n^{3-2\gamma}}$. 
Likewise, the algorithm performs 
at most $2n$ updates in each round $r \in \{1, \ldots, n\}$. 
Since the enumeration delay is
$\bigO{|\inst{D}|^{1-\gamma}}$, non-emptiness
of the query result at the end of round $r$ can be decided in  
$\bigO{|\inst{D}|^{1-\gamma}}$ time. 
Hence, the time needed in a single round is 
$\bigO{n^{2-2\gamma}}$ and the time for $n$ rounds is
 $\bigO{n^{3-2 \gamma}}$. 
This reduction gives the 
same conditional lower bound for the maintenance of 
$Q(\inst{F})$ as for the maintenance of the triangle count.
}
\end{proof}

\nop{
Theorem~\ref{theo:main_result} and Proposition~\ref{prop:lower_bound_triangle_count} imply that for $\eps = \frac{1}{2}$, \ivme incrementally maintains the count of triangles under single-tuple updates to a database $\db$ with optimal amortized update time $\bigO{|\db|^{\frac{1}{2}}}$ and constant answer time, unless the \OMv conjecture fails (Corollary~\ref{cor:ivme_optimal}).
}

%% file: recovery.tex

\section{Recovering Existing Dynamic and Static Approaches}
\label{sec:recovery}

We next show how \ivme recovers the classical first-order IVM~\cite{Chirkova:Views:2012:FTD} on triangle queries (Section \ref{sec:classicalIVM}) and the worst-case optimal time of non-incremental algorithms for computing the result of the ternary triangle query (Section \ref{sec:staticFull}).

\nop{We assume that $\inst{D}= \{R,S,T\}$ is the input database.}

\subsection{Classical First-Order IVM} 
\label{sec:classicalIVM}
We start with a brief description of classical first-order IVM on the ternary
triangle query $\triangle_3$. The other triangle queries
are treated analogously. 
Classical first-order IVM materializes the query result.
Given a single-tuple update $\delta R = \{(\deltaA,\deltaB) \mapsto \p\}$
to relation $R$, 
it maintains query $\triangle_3$ under the update by computing 
the delta query 
$$\delta \triangle_3(\deltaA,\deltaB,c) = 
\delta R(\deltaA,\deltaB) \ztimes 
S(\deltaB,c) \ztimes T(c,\deltaA)$$ 
and updating the query result by setting 
 $\triangle_3(\deltaA,\deltaB,c) := \triangle_3(\deltaA,\deltaB,c) + 
 \delta \triangle_3(\deltaA,\deltaB,c)$ for each $C$-value 
 $c$ in  $\delta \triangle_3$.
The maintenance of the query requires
the iteration over possibly linearly many $C$-values paired with 
$\deltaB$
in relation $S$ and with $\deltaA$ in relation $T$.
Hence, the update time is $\bigO{|\inst{D}|}$.   
The evaluation of updates to the relations 
$S$ and $T$ is analogous. 
The preprocessing phase uses a worst-case optimal 
join algorithm to compute the initial query result in 
$\bigO{|\inst{D}|^{\frac{3}{2}}}$ time~\cite{NgoPRR18}.
Since the query result is materialized, the enumeration delay is  
constant.
The space complexity is dominated by the
size $\bigO{|\inst{D}|^{\frac{3}{2}}}$ of the query result~\cite{LW:1949}.   

\ivme becomes the classical first-order IVM algorithm by setting 
$\eps$ to $0$ or  $1$.
 
We first consider the case $\eps = 1$ and explain it for the ternary triangle query;
the other triangle queries are treated analogously.
If $\eps = 1$, then all tuples are in the light parts of the relations and 
the results of all materialized views in Figure~\ref{fig:view_definitions_full} become empty except for the skew-aware view
$$\triangle_3^{LLL}(a,b,c) = R^L(a,b) \ztimes S^L(b,c) \ztimes T^L(c,a),$$  
whose result becomes exactly that of $\triangle_3$.

We next explain in more detail. 
The preprocessing stage 
sets the threshold base 
$N$ of the initial \ivme state  
to $2 \cdot |\inst{D}| +1$
 and strictly 
 partitions 
each relation with threshold $N^{\eps} = N$.
Since for each relation $K \in \{R,S,T\}$, variable $X$ in the schema
of $K$, and value $x$ in the domain of $X$, 
it holds $|\sigma_{X=x}K| < N$, all tuples in $K$ end up in 
the light part of $K$.
Consequently, all materialized views 
in Figure~\ref{fig:view_definitions_full}
besides  $\triangle_3^{LLL}$
stay empty, since each of them
refers to at least one heavy relation part.
The only 
materialized view that is possibly non-empty is $\triangle_3^{LLL}$.
This also means that  the result of query $\triangle_3$ and 
$\triangle_3^{LLL}$ are equal.
 Given an update, the procedure  
 \textsc{OnUpdate} in Figure~\ref{fig:onUpdate}
  never performs 
 minor rebalancing, since  
  the degrees of data values can never reach   $\frac{3}{2} N$, due to
 the size invariant $\floor{\frac{1}{4}N} \leq |\db| < N$. 
The procedure 
\textsc{MajorRebalancing}, which might be 
invoked by \textsc{OnUpdate},
does not move 
tuples to the heavy relation parts, 
since the threshold for strict partitioning 
is always greater 
than  the database size.   
This implies that 
the  views in Figure~\ref{fig:view_definitions_full}
besides $\triangle_3^{LLL}$
stay empty after any update.

The case of $\eps = 0$ is symmetric and \ivme becomes the classical first-order IVM algorithm. 
In the preprocessing stage, the input relations are strictly partitioned with threshold 
$N^\eps = 1$, which means that all light relation parts and
 materialized views referring to these parts become empty. 
Only one skew-aware view is constructed and its result is equal to 
that of the triangle query under consideration. 
\ivme materializes this view and allows for 
constant-delay enumeration from it. 

We next discuss in more detail the ternary triangle query.
The result of the skew-aware view
$\triangle_3^{HHH}(a,b,c) = R^H(a,b) \ztimes S^H(b,c) \ztimes T^H(c,a)$ 
is equal to the result of  $\triangle_3$. 
The condition $\eps = 0$ in the third line of the procedure
\textsc{AffectedPart} in Figure~\ref{fig:onUpdate} avoids that any update affects 
the light relation parts. Since the degrees of 
data values in the heavy relation parts can never fall 
below $\frac{1}{2}N^{\eps} = \frac{1}{2}$, minor rebalancing 
is never invoked.  
Based on the threshold for strict relation partitioning, 
major rebalancing does not move tuples to the 
light relation parts. 

\nop{
The case of the nullary triangle 
query is analogous to the ternary triangle query. 
For the unary and binary triangle queries, Theorem~\ref{theo:main_result_triangle} gives 
linear-time delay for $\eps=0$. This upper bound is much looser than the constant-time delay that is actually achieved by \ivme.
}

\subsection{Computing the Ternary Triangle Query in a Static Database}
\label{sec:staticFull}
The worst-case optimal time to compute the result of the ternary 
triangle query over the database $\inst{D}$ is $\bigO{|\db|^{\frac{3}{2}}}$~\cite{NgoPRR18}.
\ivme recovers this computation time 
in the static case by using its update mechanism as follows.
We fix $\eps = \frac{1}{2}$ and insert all tuples from $\db$, one at a time, into a
database $\inst{D}'$ that is initially empty.  
For each insert, we call the procedure
 \textsc{OnUpdate} from Figure~\ref{fig:onUpdate}.
   The preprocessing time is constant. 
By Theorem~\ref{theo:main_result_triangle}, \ivme guarantees  
$\bigO{M^{\frac{1}{2}}}$ amortized update time, where $M$
is the size of $\inst{D}'$ at update time.
Thus, the total time to insert all tuples into $\inst{D}'$ 
is
$$\bigO{\sum_{M=0}^{|\db|-1} M^{\frac{1}{2}}} = \bigO{|\db|\ztimes |\db|^{\frac{1}{2}}} = \bigO{|\db|^{\frac{3}{2}}}.$$
Finally, we enumerate the query result with constant
delay. Since the number of tuples in the  result is 
bounded by $|\db|^{\frac{3}{2}}$~\cite{LW:1949}, the overall enumeration 
takes $\bigO{|\db|^{\frac{3}{2}}}$ time.  
Overall, we compute the result of the ternary triangle query  
in $\bigO{|\db|^{\frac{3}{2}}}$ time.

To avoid rebalancing while inserting the tuples into the empty database, 
we can preprocess the input relations in $\inst{D}$ 
to decide for each tuple its final relation part.
For instance, if for an $A$-value $a$, it holds 
$|\sigma_{A=a} R| \geq |\db|^{\frac{1}{2}}$, the tuple is 
inserted to the heavy part of $R$, otherwise to the light part.
Since we do not perform any rebalancing,
the worst-case (and not only amortized) time of each insert is $\bigO{|\db|^{\frac{1}{2}}}$.

%% file: related.tex

\section{Related Work}
\label{sec:related}
\paragraph{Triangle queries in the static setting}
The problems of finding, counting, and 
listing of given-length cycles in graphs have been 
extensively investigated 
since the  70s~\cite{ItaiR78,ChibaN85,YusterZ97}.
One important result that falls into the scope 
of this work is that,   given a graph with $n$ vertices and $m$ edges, 
finding a triangle if one exists and counting all triangles 
can be done in time $\bigO{n^{\omega}}$
where 
$\omega < 2.373$ is the exponent of matrix multiplication~\cite{ItaiR78}.
The same problem can be solved 
in time 
$\bigO{m^{\frac{2\omega}{\omega+1}}} \leq \bigO{m^{1.41}}$, 
which is better than the former time bound on sparse 
graphs~\cite{AYZ:Counting:1997}. 
The problem of computing for each edge the number 
of triangles using this edge can be solved in 
time $\bigO{m^{1.41}}$~\cite{DurajK0W20}. 
This problem corresponds to computing the result of the binary triangle 
query over the ring of integers. 
Given a number $k$,  a flavor 
of the 
triangle listing problem asks for the listing of $k$ triangles 
if the graph has at least $k$ triangles  and all triangles otherwise. 
This problem can be solved in time 
$\softO{n^{2.373} + n^{1.568}t^{0.478}}$ 
on dense graphs and in time  
$\softO{m^{1.408} + m^{1.222}t^{0.186}}$ on sparse graphs,
where $\widetilde{\mathcal{O}}$ suppresses multiplicative 
factors of size $n^{\smallO{1}}$~\cite{BjorklundPWZ14}.
All time bounds mentioned above 
rely on algebraic fast matrix multiplication.   
\ivme's preprocessing phase relies on an algorithm
like 
Leapfrog TrieJoin or Recursive-Join that does not use 
matrix multiplication and 
runs in time $\bigO{|\inst{D}|^{\frac{3}{2}}}$~\cite{NgoPRR18}
 to compute the initial query result 
 on a database $\inst{D}$.
Further works approximate
the triangle count in large graphs~\cite{Tsourakakis08,BecchettiBCG10,KolountzakisMPT12}
 and assess the practicability of triangle counting and listing algorithms 
in massive networks~\cite{ChuC12,SchankW05}. 

\paragraph{Complexity gap between single-tuple and bulk updates}
Our main result states that for $\eps = \frac{1}{2}$, 
\ivme maintains the triangle count (unary triangle query)   
 under single-tuple 
updates to a database $\inst{D}$ with $\bigO{|\inst{D}|^{\frac{1}{2}}}$
amortized update time and $\bigO{1}$ enumeration delay
(Theorem~\ref{theo:main_result_triangle}), 
which is worst-case optimal under the 
\OMv conjecture (Proposition~\ref{prop:lower_bound_triangle}). 
We also know 
 that triangle counting on a graph with $m$ edges 
can be solved in $\bigO{m^{1.41}}$ time~\cite{AYZ:Counting:1997}.
Corroborating these two results, 
we conclude that 
there is a gap in the worst-case complexity of counting triangles between the static and the dynamic case (or equivalently between bulk updates and single-tuple updates).  
If the tuples in $\inst{D}$
come as a stream of inserts and we do 
one insert at a time, the overall time 
to compute the triangle count on $\inst{D}$
is $\bigO{|\inst{D}| \cdot |\inst{D}|^{\frac{1}{2}}} = 
\bigO{|\inst{D}|^{\frac{3}{2}}}$. 
This is worse than $\bigO{|\inst{D}|^{1.41}}$,
which is achieved by processing 
all tuples in $\inst{D}$ in bulk. For the ternary triangle query, 
however, \ivme recovers the worst-case optimal time to list all triangles
in the static setting, cf.\@ 
Section~\ref{sec:staticFull}. 

\paragraph{Dynamic set intersection}
A prior result~\cite{KopelowitzPP15} on the dynamic 
evaluation of a class of Boolean queries
is closely related to  the maintenance of
the nullary triangle query. 
Assume that $\calF$ is a family of sets 
that are subject to inserts and deletes
and $N$ is the overall size of these sets. 
Given two sets from $\calF$, the emptiness 
query answers whether their intersection
is empty.  
There is a dynamic algorithm that uses 
$\bigO{N}$ space, executes updates to the 
sets in $\calF$ in $\bigO{N^{\frac{1}{2}}}$
expected time, and answers emptiness 
queries in $\bigO{N^{\frac{1}{2}}}$
expected time.
The proof of this result reveals that 
the algorithm categorizes the sets in $\calF$
into {\em small} and {\em large} sets using some threshold
and maintains the intersection size for any two 
large sets in a 
lookup table. 
The emptiness query for two sets, where one 
of the sets is small, is answered by iterating over the
elements in the small set and checking for each element
its containment in the other set. 
For two large sets, the emptiness query 
is answered by using the intersection-size 
table. Although not stated in that work, 
the intersection-size 
table can be constructed 
in $\bigO{N^{\frac{3}{2}}}$ 
expected time in the preprocessing phase.
The algorithm can be adapted 
to allow for an unbounded number of sets 
in $\calF$ and to return the intersection size 
for any two sets from $\calF$.
This prior work can be used to 
recover a restricted instance
of \ivme for the nullary query. 
Given a database 
$\inst{D} = \{R(A,B), S(B,C), T(C,A)\}$, 
an $A$-value $a \in \pi_{A}R$,
and a $B$-value $b \in \pi_{B}R$, 
we denote by $\calR_{a}^B$ and $\calR_{b}^A$
the set of $B$-values paired with 
$a$ in $R$ and respectively the set 
of $A$-values paired with $b$ in $R$. 
The sets $\calS_{b}^C$, $\calS_{c}^B$,
$\calT_{c}^A$, and $\calT_{a}^C$ 
are defined analogously. 
Let $\calF$ consist of these sets 
for all data values in the database. 
 Assuming that the current triangle count on $\inst{D}$
is materialized,
we can obtain the new triangle count upon an insert 
of a tuple $(a,b)$ to relation $R$ as follows.
If $\calR_{a}^B$ is already contained 
in $\calF$, we extend this set by $b$, otherwise 
we create a new set $\calR_{a}^B =\{b\}$.
The set $\calR_{b}^A$ is updated or created 
analogously. 
Then, we ask for
the intersection size 
$\calS_{b}^C \cap \calT_{a}^C$.
The new triangle count is the sum of the 
previous count and the size of this intersection. 
Updating the sets in $\calF$ and computing the 
intersection size require  
$\bigO{|\inst{D}|^{\frac{1}{2}}}$ expected time~\cite{KopelowitzPP15}. 
Deletes to $R$ and updates to the  other relations 
are handled analogously. 
Since the triangle count is materialized, 
it allows constant-time access. 
Hence, we obtain  a maintenance strategy 
 for the nullary triangle query with 
 $\bigO{|\inst{D}|^{\frac{3}{2}}}$ expected preprocessing time,
 $\bigO{|\inst{D}|}$ space,   
$\bigO{|\inst{D}|^{\frac{1}{2}}}$ expected update time, 
and 
$\bigO{1}$ enumeration delay.
While meeting the complexity bounds 
of Proposition~\ref{prop:tighter-upper-bound-space-nullary} (for 
$\eps = \frac{1}{2}$), 
this alternative approach does not  support tuple multiplicities
or arbitrary rings beyond the ring of integers.

\paragraph{Fine-grained lower bounds}
Investigations  on fine-grained complexity have led to 
important 
conjectures and hypotheses
on finding and listing triangles in graphs that have served 
as conditional lower bounds for many other problems  \cite{Patrascu10,AbboudW14}. 
The strong triangle conjecture states that in the word-RAM 
model with words of length $\bigO{\log n}$, there is no
algorithm that decides whether a graph with $n$ nodes 
and $m$ edges contains a triangle in 
$\bigO{\min\{n^{\omega - \gamma}, m^{\frac{2\omega}{\omega+1} - \gamma}\}}$
expected time for any $\gamma >0$, where $\omega$
is the exponent of matrix multiplication. 
Moreover, there is no combinatorial 
algorithm that solves this problem in $\bigO{m^{\frac{3}{2}-\gamma}}$
time, for any $\gamma >0$. According to this conjecture, 
 the best known algorithms for this problem, the combinatorial ones as
 well as those based on fast matrix multiplication,   
 are optimal.
The \OMv conjecture (Conjecture~\ref{def:OMv}) \cite{Henzinger:OMv:2015}  
was used to derive conditional lower bounds on the maintenance 
 of conjunctive queries~\cite{BerkholzKS17}.
 It states that 
for any $\gamma >0$, there is no algorithm that solves
the \OMv problem (Definition~\ref{def:OuMv}) in $\bigO{n^{3-\gamma}}$ time.
The best known algorithm solving the \OMv problem  runs in
$\bigO{\frac{n^3}{\log^2 n}}$ time~\cite{Williams07}. 
Let $Q$ be a Boolean conjunctive query whose homomorphic core is 
not q-hierarchical~\cite{BerkholzKS17}. Then, 
for  any $\gamma >0$ and database of domain size $n$, there is no algorithm that incrementally maintains the result of $Q$ under single-tuple updates with arbitrary preprocessing time, $\bigO{n^{1-\gamma}}$ update time, and $\bigO{n^{2-\gamma}}$ answer time, unless the \OMv conjecture fails~\cite{BerkholzKS17}. 
Triangle queries are not q-hierarchical and their homomorphic cores
are the queries themselves in case they do not have repeating relation symbols.
Hence, the above lower bound holds for all triangle queries without repeating relation symbols. The proof of this lower bound is similar to that for the query $\varphi = \exists x \exists y (S(x) \wedge E(x,y) \wedge T(y))$, which is 
 the simplest Boolean
conjunctive query that is not q-hierarchical~\cite{BerkholzKS17}. 
Our lower bound proof in Section~\ref{sec:lowerbound} 
 adapts the proof for $\varphi$ to triangle queries,
 strengthens it to allow for amortized update time, and 
expresses complexities in terms of the database size.

\paragraph{Enumeration with skip pointers}
Skip pointers have been previously used for 
constant-delay enumeration of  
distinct elements in the union of a fixed number of 
sets~\cite{Berkholz:ICDT:2018}.
Section~\ref{sec:skip_pointers} introduces this approach using the abstraction of hop iterators. 
Our approach extends the original method~\cite{Berkholz:ICDT:2018} with second-level skip pointers 
and parameterizes it by a search function to enable tighter bounds on enumeration delay.
We use iterators with skip pointers 
in the enumeration procedures
for the binary and unary triangle queries.

\paragraph{Approximation schemes in the dynamic setting}
A distinct line of work investigates randomized approximation schemes with an arbitrary relative error for counting triangles in a graph given as a stream of 
edges~\cite{Bar:reductions:SODA:2002,Jowhari:new:COCOON:2005,Buriol:counting:PODS:2006,Mcgregor:better:PODS:2016,Cormode:secondlook:TCS:2017}. Each edge in the data stream corresponds to 
a tuple insert, and tuple deletes are not considered. The emphasis of these approaches is on space efficiency, and they express the space utilization as a function of the number of nodes and edges in the input graph and of the number of triangles. The space utilization is generally sublinear but may become superlinear if, for instance, the number of edges is greater than the square root of the number of triangles. The update time is polylogarithmic in the number of nodes in the graph. There is also work estimating 
the number of triangles in graph streams with both edge inserts
and deletes~\cite{BulteauFKP16}.

\paragraph{Dynamic descriptive complexity}
Further away from our line of work is the development of dynamic descriptive complexity, starting with the DynFO complexity class and the much-acclaimed result on FO expressibility of the maintenance for graph reachability under edge inserts and deletes, see a recent survey~\cite{Schwentick:DynamicComplexity:2016}. 
The $k$-clique query can be maintained under edge inserts by a quantifier-free update program of arity $k-1$ but not of arity $k-2$~\cite{Zeume:Clique:2017}.

%% file: extensions.tex
\section{Extensions}
\label{sec:extensions}

\paragraph{Relations over task-specific rings}
Different rings can be used as the domain of tuple multiplicities (or payloads). We used here the ring $(\mathbb{Z},+,\ztimes,0,1)$ of integers to support counting. Previous work shows how the data-intensive 
computation of many applications can be captured by application-specific 
rings, which define sum and product operations over data values~\cite{Nikolic:SIGMOD:18}. 
The relational data ring supports payloads with listing and factorized representations of relations, and the degree-$m$ matrix ring supports 
payloads that can be used for maintaining gradients of square loss functions for linear regression models~\cite{Nikolic:SIGMOD:18}.

\paragraph{\ivme variants}
\ivme can be used to maintain 
triangle queries with repeating relation symbols, 
the counting versions of any query built using three relations and the 4-path query~\cite{KaraNNOZ_triangle_arxiv}
in worst-case optimal update time.
The same conditional lower bound on the update time shown for the triangle count (nullary triangle query) applies for most of the mentioned queries, too.
This leads to the striking realization that, while in the static setting the counting versions of the cyclic query computing triangles and the acyclic query computing paths of length $3$ have different complexities and pose distinct computational challenges, they share the same complexity and can use a very similar approach in the dynamic setting.

\paragraph{Loomis Whitney queries}
The \ivme maintenance strategies also naturally extend from triangle to 
Loomis Whitney (LW) queries. LW queries 
generalize triangle queries from cliques of degree three to cliques of degree $n\geq 3$; they encode the Loomis Whitney inequality~\cite{LW:1949}.
Let $A_1,\ldots,A_n$ be the query variables and $R_1,\ldots,R_n$ relations over schemas ${\bf X}_1,\ldots,{\bf X}_n$, 
where $\forall i\in[n]: {\bf X}_i = (A_{((i+j)\mod n) + 1})_{-1\leq j\leq n-3}$. 
That is, the schema of $R_1$ is $(A_1,\ldots,A_{n-1})$, whereas the schema of $R_n$ is $(A_n,A_1,\ldots,A_{n-2})$.
The $n$-ary LW query of degree $n$  has the form 
$$\Diamond_n(\inst{x}) =   
R_1({\bf x}_1)\cdots R_n({\bf x}_n),$$ 
where 
$\inst{x} = (a_j)_{j \in [n]}$ and
for all $i\in[n]$, $\inst{x}_i = (a_{((i+j)\mod n) + 1})_{-1\leq j\leq n-3}$ is a value from the domain of the 
tuple ${\bf X}_i$ of variables.
As for triangle queries, a LW query of degree $n$ and arity 
$0\leq k\leq n-1$ has the same body as for arity $n$ but only keeps the first $k$ values in the result. For instance, for $n=4$ the binary LW query is 
$$\Diamond_2(a_1,a_2)=\sum\limits_{a_3,a_4} R_1(a_1,a_2,a_3)\cdot R_2(a_2,a_3,a_4)\cdot R_3(a_3,a_4,a_1)\cdot R_4(a_4,a_1,a_2).$$ 
In case $n=3$, each LW query $\Diamond_k$ becomes the triangle query $\triangle_k$, for $0\leq k\leq 3$.

\ivme achieves the following complexities 
for LW queries of degree $n$ (stated without proof):
\begin{itemize}
\item The preprocessing and amortized update time 
are the same as for triangle queries: 
$\bigO{|\inst{D}|^{\frac{3}{2}}}$ preprocessing time and
$\bigO{|\inst{D}|^{\max\{\eps, 1-\eps\}}}$ amortized update time.
  
\item In case all variables are free, the space complexity 
is the same as for the ternary triangle query, namely, 
$\bigO{|\inst{D}|^{\frac{3}{2}}}$; otherwise,
the space complexity is $\bigO{|\inst{D}|^{1 + \min\{\eps, 1-\eps\}}}$. 

\item For the nullary and $n$-ary LW queries,
the enumeration delay is constant; for $k$-ary LW queries
where $0<k<n$, the enumeration delay is  
$\bigO{|\inst{D}|^{\min\{1, (n - k) \cdot (1-\eps)\}}}$. The delay hence improves
with increasing arity. For $n = 3$, we get exactly the same 
enumeration delay as for the triangle queries. 

\item The lower bound on the update-delay trade-off for triangle queries stated 
in Proposition~\ref{prop:lower_bound_triangle} carry over 
to LW queries. This means that at $\eps = \frac{1}{2}$, \ivme is 
strongly Pareto
worst-case 
 optimal for the nullary and $n$-ary LW queries
and weakly Pareto worst-case 
optimal for all other LW queries. 

\end{itemize}

The result of the n-ary LW query $\Diamond_n$ of degree $n$ has size $\bigO{|\inst{D}|^{\frac{n}{n-1}}}$~\cite{LW:1949}. It can also be computed in the static setting in the same time, which is thus worst-case optimal~\cite{NgoPRR18}. \ivme cannot be used to recover the optimality in the static case, since it takes $\bigO{|\inst{D}|^{\frac{1}{2}}}$ amortized time per each single-tuple update and there are $|\inst{D}|$ tuples to insert. Since the combination of $\bigO{|\inst{D}|^{\frac{1}{2}}}$ amortized time and $\bigO{1}$ delay is 
 strongly Pareto
worst-case 
 optimal, it means that no dynamic algorithm can achieve a lower amortized single-tuple update time for the n-ary LW query. 
This shows the limitation of single-tuple updates. 
To achieve the overall $\bigO{|\inst{D}|^{\frac{n}{n-1}}}$ time for $|\inst{D}|$ tuple inserts, one would need to process several inserts at the same time, that is, in bulk, such that the amortized time per insert should be $\bigO{|\inst{D}|^{\frac{1}{n-1}}}$. 
A characterization of the difference between bulk updates and single-tuple updates remains an interesting open problem.

%% file: conclusion.tex
\section{Conclusion and Future Work}
\label{sec:conclusion}
This article introduces \ivme, an incremental maintenance approach for 
triangle queries under updates that exhibits a trade-off between the update time on one hand and 
the space and enumeration delay on the other hand.
\ivme captures classical first-order IVM as a special case that has suboptimal linear update time.

There are worst-case optimal algorithms for {\em join} queries in the {\em static} setting~\cite{NgoPRR18}. In contrast, \ivme is worst-case optimal for the nullary and ternary triangle join queries in the {\em dynamic} setting. The dynamic setting case poses challenges beyond the static setting.
First, the optimality argument for static join algorithms follows from their runtime being linear(ithmic) in their output size; this argument does not apply to our nullary triangle query, since its output is a scalar and hence of constant size. 
Second, optimality in the dynamic setting requires a more fine-grained argument that exploits the skew in the data for different evaluation strategies, view materialization, and delta computation; 
in contrast, there are static worst-case optimal join algorithms 
that do not need to exploit skew, materialize views, nor delta computation. 

We conclude with a discussion on possible directions for future work.

\paragraph{Worst-case optimal dynamic query evaluation}
This article opens up a line of work on  
worst-case optimal dynamic query evaluation algorithms. The goal is a complete characterization of the complexity of incremental maintenance for arbitrary functional aggregate queries~\cite{FAQ:PODS:2016}. 
We  would first like to find a syntactical characterization of all queries that admit incremental maintenance in (amortized) sublinear time. Using known (first-order, fully recursive, or factorized) incremental maintenance techniques, cyclic and even acyclic joins require at least linear update time. Our intuition is that this characterization is given by a notion of diameter of the query hypergraph. This class strictly contains the q-hierarchical queries, which admit constant-time updates~\cite{BerkholzKS17}. A first step towards this goal is a characterization 
of the update-delay trade-off for hierarchical queries with arbitrary free variables~\cite{KNOZ19_arxiv}.

\paragraph{Space-delay trade-off}
\ivme does not admit any trade-off 
between 
the space complexity and the enumeration delay: 
for all queries, 
there is either no or positive correlation between the two measures
(cf.\@ Figure~\ref{fig:complexity_plots_triangle}). 
Prior work investigates the trade-off between space and delay 
for the evaluation of conjunctive queries in the static 
setting~\cite{DeepK18}.
An interesting future direction is to design 
a maintenance approach with focus on the space-delay trade-off.

\paragraph{Implementation of \ivme}
We would like to implement \ivme and benchmark against
 existing IVM systems. The implementation 
  of \ivme may pose some challenges. 
For instance, maintaining the exact heavy-light partitions of relations 
is computationally expensive. 
One way to handle this problem is to loosen the partition thresholds
so that relation partitions are rebalanced less frequently while 
accepting temporarily  suboptimal maintenance strategies.  
A further challenge is the maintenance of the index structures 
 of \ivme. 
For each materialized view $V$ with some schema $\inst{X}$ 
and sub-schema  $\inst{Y} \subseteq \inst{X}$,
 \ivme assumes the existence
 of an index that allows to check containment of any tuple $\inst{y}$ over 
 $\inst{Y}$ in $\pi_{\inst{Y}}V$ in constant time and to enumerate all tuples in 
 $V$ matching $\inst{y}$ with constant delay
We need to address the trade-off between the cost of maintaining this indices and the cost of access times without them.

%% file: main.bbl
\begin{thebibliography}{10}

\bibitem{AbboudW14}
A.~Abboud and V.~V. Williams.
\newblock Popular conjectures imply strong lower bounds for dynamic problems.
\newblock In {\em {FOCS}}, pages 434--443, 2014.

\bibitem{FAQ:PODS:2016}
M.~Abo~Khamis, H.~Q. Ngo, and A.~Rudra.
\newblock {FAQ: Questions Asked Frequently}.
\newblock In {\em PODS}, pages 13--28, 2016.
\newblock DOI:
  \href{https://doi.org/10.1145/2902251.2902280}{10.1145/2902251.2902280}.

\bibitem{AYZ:Counting:1997}
N.~Alon, R.~Yuster, and U.~Zwick.
\newblock Finding and {Counting Given Length Cycles}.
\newblock {\em Algorithmica}, 17(3):209--223, 1997.
\newblock DOI: \href{https://doi.org/10.1007/BF02523189}{10.1007/BF02523189}.

\bibitem{Bar:reductions:SODA:2002}
Z.~Bar-Yossef, R.~Kumar, and D.~Sivakumar.
\newblock {Reductions in Streaming Algorithms, with an Application to Counting
  Triangles in Graphs}.
\newblock In {\em SODA}, pages 623--632, 2002.

\bibitem{BecchettiBCG10}
L.~Becchetti, P.~Boldi, C.~Castillo, and A.~Gionis.
\newblock {Efficient algorithms for large-scale local triangle counting}.
\newblock {\em TKDD}, 4(3):13:1--13:28, 2010.
\newblock DOI:
  \href{https://doi.org/10.1145/1839490.1839494}{10.1145/1839490.1839494}.

\bibitem{BerkholzKS17}
C.~Berkholz, J.~Keppeler, and N.~Schweikardt.
\newblock {Answering Conjunctive Queries Under Updates}.
\newblock In {\em PODS}, pages 303--318, 2017.
\newblock DOI:
  \href{https://doi.org/10.1145/3034786.3034789}{10.1145/3034786.3034789}.

\bibitem{Berkholz:ICDT:2018}
C.~Berkholz, J.~Keppeler, and N.~Schweikardt.
\newblock {Answering UCQs Under Updates and in the Presence of Integrity
  Constraints}.
\newblock In {\em ICDT}, pages 8:1--8:19, 2018.
\newblock DOI:
  \href{https://doi.org/10.4230/LIPIcs.ICDT.2018.8}{10.4230/LIPIcs.ICDT.2018.8}.

\bibitem{BjorklundPWZ14}
A.~Bj{\"{o}}rklund, R.~Pagh, V.~V. Williams, and U.~Zwick.
\newblock {Listing Triangles}.
\newblock In {\em ICALP}, pages 223--234, 2014.
\newblock DOI:
  \href{https://doi.org/10.1007/978-3-662-43948-7_19}{10.1007/978-3-662-43948-7\_19}.

\bibitem{BulteauFKP16}
L.~Bulteau, V.~Froese, K.~Kutzkov, and R.~Pagh.
\newblock Triangle counting in dynamic graph streams.
\newblock {\em Algorithmica}, 76(1):259--278, 2016.
\newblock DOI:
  \href{https://doi.org/10.1007/s00453-015-0036-4}{10.1007/s00453-015-0036-4}.

\bibitem{Buriol:counting:PODS:2006}
L.~S. Buriol, G.~Frahling, S.~Leonardi, A.~Marchetti-Spaccamela, and C.~Sohler.
\newblock {Counting Triangles in Data Streams}.
\newblock In {\em PODS}, pages 253--262, 2006.
\newblock DOI:
  \href{https://doi.org/10.1145/1142351.1142388}{10.1145/1142351.1142388}.

\bibitem{ChibaN85}
N.~Chiba and T.~Nishizeki.
\newblock {Arboricity and Subgraph Listing Algorithms}.
\newblock {\em {{SIAM} J. Comput.}}, 14(1):210--223, 1985.
\newblock DOI: \href{https://doi.org/10.1137/0214017}{10.1137/0214017}.

\bibitem{Chirkova:Views:2012:FTD}
R.~Chirkova and J.~Yang.
\newblock {Materialized Views}.
\newblock {\em Found. \& Trends DB}, 4(4):295--405, 2012.
\newblock DOI: \href{https://doi.org/10.1561/1900000020}{10.1561/1900000020}.

\bibitem{ChuC12}
S.~Chu and J.~Cheng.
\newblock {Triangle Listing in Massive Networks}.
\newblock {\em TKDD}, 6(4):17:1--17:32, 2012.
\newblock DOI:
  \href{https://doi.org/10.1145/2382577.2382581}{10.1145/2382577.2382581}.

\bibitem{Cormode:secondlook:TCS:2017}
G.~Cormode and H.~Jowhari.
\newblock {A Second Look at Counting Triangles in Graph Streams (Corrected)}.
\newblock {\em Theor. Comput. Sci.}, 683:22--30, 2017.
\newblock DOI:
  \href{https://doi.org/10.1016/j.tcs.2016.06.020}{10.1016/j.tcs.2016.06.020}.

\bibitem{DeepK18}
S.~Deep and P.~Koutris.
\newblock Compressed representations of conjunctive query results.
\newblock In {\em PODS}, pages 307--322, 2018.
\newblock DOI:
  \href{https://doi.org/10.1145/3196959.3196979}{10.1145/3196959.3196979}.

\bibitem{DurajK0W20}
L.~Duraj, K.~Kleiner, A.~Polak, and V.~V. Williams.
\newblock Equivalences between triangle and range query problems.
\newblock In {\em SODA}, 2020.
\newblock DOI:
  \href{https://doi.org/10.1137/1.9781611975994.3}{10.1137/1.9781611975994.3}.

\bibitem{Durand:CSL:11}
A.~Durand and Y.~Strozecki.
\newblock Enumeration complexity of logical query problems with second-order
  variables.
\newblock In {\em CSL}, pages 189--202, 2011.
\newblock DOI:
  \href{https://doi.org/10.4230/LIPIcs.CSL.2011.189}{10.4230/LIPIcs.CSL.2011.189}.

\bibitem{Eden:approximately:FOCS:2015}
T.~Eden, A.~Levi, D.~Ron, and C.~Seshadhri.
\newblock {Approximately Counting Triangles in Sublinear Time}.
\newblock In {\em FOCS}, pages 614--633, 2015.
\newblock DOI:
  \href{https://doi.org/10.1109/FOCS.2015.44}{10.1109/FOCS.2015.44}.

\bibitem{Henzinger:OMv:2015}
M.~Henzinger, S.~Krinninger, D.~Nanongkai, and T.~Saranurak.
\newblock {Unifying and Strengthening Hardness for Dynamic Problems via the
  Online Matrix-Vector Multiplication Conjecture}.
\newblock In {\em STOC}, pages 21--30, 2015.
\newblock DOI:
  \href{https://doi.org/10.1145/2746539.2746609}{10.1145/2746539.2746609}.

\bibitem{Idris:dynamic:SIGMOD:2017}
M.~Idris, M.~Ugarte, and S.~Vansummeren.
\newblock {The Dynamic Yannakakis Algorithm: Compact and Efficient Query
  Processing Under Updates}.
\newblock In {\em SIGMOD}, pages 1259--1274, 2017.
\newblock DOI:
  \href{https://doi.org/10.1145/3035918.3064027}{10.1145/3035918.3064027}.

\bibitem{ItaiR78}
A.~Itai and M.~Rodeh.
\newblock {Finding a Minimum Circuit in a Graph}.
\newblock {\em {{SIAM} J. Comput.}}, 7(4):413--423, 1978.
\newblock DOI: \href{https://doi.org/10.1137/0207033}{10.1137/0207033}.

\bibitem{Jowhari:new:COCOON:2005}
H.~Jowhari and M.~Ghodsi.
\newblock {New Streaming Algorithms for Counting Triangles in Graphs}.
\newblock In {\em COCOON}, pages 710--716, 2005.
\newblock DOI:
  \href{https://doi.org/10.1007/11533719_72}{10.1007/11533719\_72}.

\bibitem{KaraNNOZ_triangle_arxiv}
A.~Kara, H.~Q. Ngo, M.~Nikolic, D.~Olteanu, and H.~Zhang.
\newblock Counting triangles under updates in worst-case optimal time.
\newblock {\em CoRR}, abs/1804.02780, 2018.

\bibitem{KaraNNOZ19}
A.~Kara, H.~Q. Ngo, M.~Nikolic, D.~Olteanu, and H.~Zhang.
\newblock Counting triangles under updates in worst-case optimal time.
\newblock In {\em ICDT}, pages 4:1--4:18, 2019.
\newblock DOI:
  \href{https://doi.org/10.4230/LIPIcs.ICDT.2019.4}{10.4230/LIPIcs.ICDT.2019.4}.

\bibitem{KNOZ19_arxiv}
A.~Kara, M.~Nikolic, D.~Olteanu, and H.~Zhang.
\newblock Trade-offs in static and dynamic evaluation of hierarchical queries.
\newblock {\em CoRR}, abs/1907.01988, 2019.
\newblock To appear in PODS 2020.

\bibitem{KochAKNNLS14}
C.~Koch, Y.~Ahmad, O.~Kennedy, M.~Nikolic, A.~N{\"{o}}tzli, D.~Lupei, and
  A.~Shaikhha.
\newblock {DBToaster: Higher-Order Delta Processing for Dynamic, Frequently
  Fresh Views}.
\newblock {\em {VLDB} J.}, 23(2):253--278, 2014.
\newblock DOI:
  \href{https://doi.org/10.1007/s00778-013-0348-4}{10.1007/s00778-013-0348-4}.

\bibitem{KolountzakisMPT12}
M.~N. Kolountzakis, G.~L. Miller, R.~Peng, and C.~E. Tsourakakis.
\newblock {Efficient Triangle Counting in Large Graphs via Degree-Based Vertex
  Partitioning}.
\newblock {\em Internet Mathematics}, 8(1-2):161--185, 2012.
\newblock DOI:
  \href{https://doi.org/10.1080/15427951.2012.625260}{10.1080/15427951.2012.625260}.

\bibitem{KopelowitzPP15}
T.~Kopelowitz, S.~Pettie, and E.~Porat.
\newblock Dynamic set intersection.
\newblock In {\em WADS}, pages 470--481, 2015.
\newblock DOI:
  \href{https://doi.org/10.1007/978-3-319-21840-3_39}{10.1007/978-3-319-21840-3\_39}.

\bibitem{Koutris:FTDB:2018}
P.~Koutris, S.~Salihoglu, and D.~Suciu.
\newblock {Algorithmic Aspects of Parallel Data Processing}.
\newblock {\em Found. \& Trends DB}, 8(4):239--370, 2018.
\newblock DOI: \href{https://doi.org/10.1561/1900000055}{10.1561/1900000055}.

\bibitem{LW:1949}
L.~H. Loomis and H.~Whitney.
\newblock An inequality related to the isoperimetric inequality.
\newblock {\em Journal: Bull. Amer. Math. Soc.}, 55(55):961--962, 1949.
\newblock DOI:
  \href{https://doi.org/10.1090/S0002-9904-1949-09320-5}{10.1090/S0002-9904-1949-09320-5}.

\bibitem{Mcgregor:better:PODS:2016}
A.~McGregor, S.~Vorotnikova, and H.~T. Vu.
\newblock {Better Algorithms for Counting Triangles in Data Streams}.
\newblock In {\em PODS}, pages 401--411, 2016.
\newblock DOI:
  \href{https://doi.org/10.1145/2902251.2902283}{10.1145/2902251.2902283}.

\bibitem{NgoPRR18}
H.~Q. Ngo, E.~Porat, C.~R{\'{e}}, and A.~Rudra.
\newblock Worst-case optimal join algorithms.
\newblock {\em J. {ACM}}, 65(3):16:1--16:40, 2018.
\newblock DOI: \href{https://doi.org/10.1145/3180143}{10.1145/3180143}.

\bibitem{Nikolic:SIGMOD:18}
M.~Nikolic and D.~Olteanu.
\newblock {Incremental View Maintenance with Triple Lock Factorization
  Benefits}.
\newblock In {\em SIGMOD}, pages 365--380, 2018.
\newblock DOI:
  \href{https://doi.org/10.1145/3183713.3183758}{10.1145/3183713.3183758}.

\bibitem{Patrascu10}
M.~Patrascu.
\newblock Towards polynomial lower bounds for dynamic problems.
\newblock In {\em {STOC}}, pages 603--610, 2010.

\bibitem{SchankW05}
T.~Schank and D.~Wagner.
\newblock {Finding, Counting and Listing All Triangles in Large Graphs, an
  Experimental Study}.
\newblock In {\em WEA}, pages 606--609, 2005.
\newblock DOI:
  \href{https://doi.org/10.1007/11427186_54}{10.1007/11427186\_54}.

\bibitem{Schwentick:DynamicComplexity:2016}
T.~Schwentick and T.~Zeume.
\newblock {Dynamic Complexity: Recent Updates}.
\newblock {\em {SIGLOG} News}, 3(2):30--52, 2016.
\newblock DOI:
  \href{https://doi.org/10.1145/2948896.2948899}{10.1145/2948896.2948899}.

\bibitem{Tsourakakis08}
C.~E. Tsourakakis.
\newblock Fast counting of triangles in large real networks without counting:
  Algorithms and laws.
\newblock In {\em ICDM}, pages 608--617, 2008.
\newblock DOI:
  \href{https://doi.org/10.1109/ICDM.2008.72}{10.1109/ICDM.2008.72}.

\bibitem{Williams07}
R.~Williams.
\newblock Matrix-vector multiplication in sub-quadratic time: (some
  preprocessing required).
\newblock In {\em {SODA}}, pages 995--1001, 2007.

\bibitem{Williams:2018:finegrained}
V.~V. Williams.
\newblock {On Some Fine-Grained Questions in Algorithms and Complexity}.
\newblock In {\em ICM}, volume~3, pages 3431--3472, 2018.
\newblock DOI:
  \href{https://doi.org/10.1142/9789813272880_0188}{10.1142/9789813272880\_0188}.

\bibitem{YusterZ97}
R.~Yuster and U.~Zwick.
\newblock {Finding Even Cycles Even Faster}.
\newblock {\em {{SIAM} J. Discrete Math.}}, 10(2):209--222, 1997.
\newblock DOI:
  \href{https://doi.org/10.1137/S0895480194274133}{10.1137/S0895480194274133}.

\bibitem{Zeume:Clique:2017}
T.~Zeume.
\newblock {The Dynamic Descriptive Complexity of k-Clique}.
\newblock {\em Inf. Comput.}, 256:9--22, 2017.
\newblock DOI:
  \href{https://doi.org/10.1016/j.ic.2017.04.005}{10.1016/j.ic.2017.04.005}.

\end{thebibliography}
